\documentclass[letterpaper,twocolumn,superscriptaddress,accepted=2023-09-12]{quantumarticle}
\pdfoutput=1
\usepackage{algorithm2e}
\usepackage{amsfonts}
\usepackage{amsmath}
\usepackage{amssymb}
\usepackage{mathtools}
\usepackage{xfrac}
\usepackage{array}
\usepackage{graphicx}
\usepackage{hyperref}
\usepackage[numbers]{natbib}
\providecommand{\U}[1]{\protect\rule{.1in}{.1in}}
\newtheorem{theorem}{Theorem}

\newtheorem{Algorithm}{Algorithm}

\newtheorem{corollary}[theorem]{Corollary}

\newtheorem{definition}{Definition}
\newtheorem{example}{Example}

\newtheorem{proposition}{Proposition}
\newtheorem{remark}{Remark}

\newenvironment{proof}[1][Proof]{\noindent\textbf{#1.} }{\ \rule{0.5em}{0.5em}}
\numberwithin{theorem}{section}
\numberwithin{proposition}{section}
\numberwithin{definition}{section}
\numberwithin{example}{section}

\usepackage{blindtext}

\newcommand{\bra}[1]{\langle#1|}
\newcommand{\ket}[1]{|#1\rangle}
\newcommand{\outerprod}[2]{\vert#1\rangle\!\langle#2 \vert}
\newcommand{\outerproj}[1]{\outerprod{#1}{#1}}
\newcommand{\intCZ}[0]{\int_{0}^{4\pi}}

\makeatletter
\def\l@subsubsection#1#2{}
\makeatother

\allowdisplaybreaks

\begin{document}
%\preprint{ }
\title[ ]{Testing symmetry on quantum computers}
\author{Margarite L. LaBorde}
\email{mlabo15@lsu.edu}
\affiliation{Hearne Institute for Theoretical Physics, Department of Physics and Astronomy, and Center for Computation and Technology, Louisiana State University, Baton Rouge, Louisiana 70803, USA}
\orcid{0000-0001-9754-1468}

\author{Soorya Rethinasamy}
\email{sr952@cornell.edu}
\affiliation{School of Applied and Engineering Physics, Cornell University, Ithaca, New
York 14850, USA}
\affiliation{Hearne Institute for Theoretical Physics, Department of Physics and Astronomy, and Center for Computation and Technology, Louisiana State University, Baton Rouge, Louisiana 70803, USA}
\orcid{0000-0002-8849-3748}

\author{Mark M. Wilde}
\email{wilde@cornell.edu}
\affiliation{School of Electrical and Computer Engineering, Cornell University, Ithaca, New
York 14850, USA}
\affiliation{Hearne Institute for Theoretical Physics, Department of Physics and Astronomy, and Center for Computation and Technology, Louisiana State University, Baton Rouge, Louisiana 70803, USA}
\orcid{0000-0002-3916-4462}

\begin{abstract}
Symmetry is a unifying concept in physics. In quantum information and beyond, it is known that quantum states possessing symmetry are not useful for certain information-processing tasks. For example, states that commute with a Hamiltonian realizing a time evolution are not useful for timekeeping during that evolution, and bipartite states that are highly extendible are not strongly entangled and thus not useful for basic tasks like teleportation. Motivated by this perspective, this paper details several quantum algorithms that test the symmetry of quantum states and channels. For the case of testing Bose symmetry of a state, we show that there is a simple and efficient quantum algorithm, while the tests for other kinds of symmetry rely on the aid of a quantum prover. We prove that the acceptance probability of each algorithm is equal to the maximum symmetric fidelity of the state being tested, thus giving a firm operational meaning to these latter resource quantifiers. Special cases of the algorithms test for incoherence or separability of quantum states. We evaluate the performance of these algorithms on choice examples by using the variational approach to quantum algorithms, replacing the quantum prover with a parameterized circuit. We demonstrate this approach for numerous examples using the IBM quantum noiseless and noisy simulators, and we observe that the algorithms perform well in the noiseless case and exhibit noise resilience in the noisy case. We also show that the maximum symmetric fidelities can be calculated by semi-definite programs, which is useful for benchmarking the performance of these algorithms for sufficiently small examples. Finally, we establish various generalizations of the resource theory of asymmetry, with the upshot being that the acceptance probabilities of the algorithms are resource monotones and thus well motivated from the resource-theoretic perspective.

\end{abstract}
%\date{\today}
%\startpage{1}
%\endpage{10}
\maketitle
\tableofcontents

\section{Introduction}

Symmetry plays a fundamental role in physics \cite{FR96,Gross96}. The
evolution of a closed physical system is dictated by a Hamiltonian, which
often possesses symmetry that limits transitions from one state to another in
the form of superselection rules \cite{Wick1952,PhysRev.155.1428}. Permutation symmetry in the extension of a bipartite quantum state indicates a lack of entanglement in that state \cite{W89a,DPS02,DPS04}. This permutation symmetry limits entanglement, which relates to fundamental principles of quantum information like the no-cloning theorem \cite{Park1970,D82,nat1982}\ and entanglement monogamy \cite{T04}. Additionally, the lack of a shared reference frame between two parties implies that a quantum state prepared relative to another party's reference frame respects a certain symmetry and is less useful than one breaking that symmetry \cite{BRS07}. In all of these cases, a state respecting a symmetry is less resourceful than one breaking it. In more recent years, quantum resource theories have been proposed for each of the above scenarios (asymmetry \cite{MS13,MS14}, unextendibility \cite{KDWW19,KDWW21}, and frameness \cite{GS08}) in order to quantify the resourcefulness of quantum states\ (see \cite{CG18}\ for a review). As such, it is useful to be able to test whether a quantum state possesses symmetry and to quantify how much symmetry it possesses. 

In this paper, we show how a quantum computer can test for symmetries of quantum states and channels generated by quantum circuits. In fact, our quantum-computational tests actually quantify how symmetric a state or channel is. Given that asymmetry (i.e., breaking of symmetry) is a useful resource in a wide variety of contexts while being potentially difficult for a classical computer to verify, our tests are helpful in determining how useful a state will be for certain quantum information processing tasks. Additionally, our tests are in the spirit of the larger research program of using quantum computers to understand fundamental quantum-mechanical properties of high-dimensional quantum states, such as symmetry and entanglement, that are out of reach for classical computers. %new sentences added to address point 1 in response
Here, we give explicit algorithmic descriptions of our tests, connect to known applications of interest, and provide a general framework that facilitates new applications and research in this area. We augment these contributions by providing novel resource-theoretic results as well.

We begin our development in Section~\ref{sec:sym-ext-G}\ by introducing a general form of symmetry of quantum states that captures both the extendibility of bipartite states \cite{W89a,DPS02,DPS04}, as well as symmetries of a single quantum system with respect to a group of unitary transformations \cite{MS13,MS14}. This generalization allows for incorporating several kinds of symmetry tests into a single framework. We call this notion $G$-symmetric extendibility, and we discuss two different forms of it.

In Section~\ref{sec:tests-o-sym} we move on to an important contribution of our paper---namely, how a quantum computer can test for and estimate quantifiers of symmetry. These quantifiers are collectively called \textit{maximum symmetric fidelities}, with more particular names given in what follows. We prove that our quantum computational tests of symmetry have acceptance probabilities precisely equal to the various quantifiers. These results endow these resource-theoretic measures with operational meanings and allow us to estimate them to arbitrary precision. Using complexity-theoretic language, we demonstrate that several of these quantum-computational tests of symmetry can be conducted in the form of a quantum interactive proof (QIP) system consisting of two quantum messages exchanged between a verifier and a prover \cite{W09,VW15}. Our results thus generalize previous results in the context of unextendibility and entanglement of bipartite quantum states \cite{HMW13,Hayden:2014:TQI}; additionally, we go on to clarify the relation between our results and previous ones (Section~\ref{sec:specialized-tests}). Simpler forms of the tests can be conducted without the aid of a prover and are thus efficiently computable on a quantum computer.

In Section~\ref{sec:specialized-tests}, we show how the established concepts of $k$-extendibility or $k$-Bose extendibility \cite{W89a,DPS02,DPS04} can be recovered as special cases of our symmetry tests for both bipartite and multipartite states. These examples are particularly interesting as they serve as tests of separability. We also show there how to test for the covariance symmetry of quantum channels and measurements, where the former includes testing the symmetries of Hamiltonian evolution as a special case \cite{LW22}. 

Section~\ref{sec:SDPs-sym-fids} shows that the maximum symmetric fidelities can be calculated by means of semi-definite programs, which is helpful for benchmarking the outputs of the quantum algorithms for sufficiently small circuits. This follows from combining the known semi-definite program for fidelity \cite{Wat13} with the semi-definite constraints corresponding to the symmetry tests. Furthermore, we employ representation theory \cite{S12} to simplify some of the semi-definite programs even further, by making use of the block-diagonal form that results from performing a group twirl on a state.

We follow this in Section~\ref{sec:var-algs} by demonstrating the use of variational quantum algorithms for estimating the maximum symmetric fidelities  for various example groups. (See \cite{CABBEFMMYCC20,bharti2021noisy}\ for reviews of variational quantum algorithms and \cite{RevModPhys.55.725} for a review of the variational principle). In general, this approach is not guaranteed to estimate the maximum symmetric fidelities precisely, as the parameterized circuit used is not able to realize an arbitrarily powerful quantum computation. This approach thus leads only to lower bounds on the maximum symmetric fidelities. However, we find that this heuristic approach performs well for a variety of example groups, including symmetry tests with respect to $\mathbb{Z}_2$, the triangular dihedral group, a collective unitary action, and a collective phase action. In Appendices~\ref{app:cyclic-c-3}--\ref{app:Q8}, we go on to provide further examples for cyclic groups and the quaternion group. We note that a recent work adopted a similar variational approach for estimating the fidelity of quantum states generated by quantum circuits \cite{CSZW20}. It is well known that this latter problem is QSZK-complete \cite{W02} and thus likely difficult for quantum computers to solve in general. It remains an open question to determine how well this variational approach performs generally, beyond the examples considered in this paper. We note that the algorithms defined in this work rely on local measurements alone and, as a consequence of the results of \cite{osti_1826512}, should not suffer from the barren plateau problem in which global cost functions become untrainable. Since we have only conducted simulations of our algorithms for small quantum systems, it remains open to provide evidence that our algorithms will avoid the barren plateau problem for larger systems.

Finally, we review the resource theory of asymmetry \cite{MS13,MS14}. After doing so, we define several generalized resource theories of asymmetry (Section~\ref{sec:res-theories}), including both the resource theory of asymmetry and the resource theory of $k$-unextendibility \cite{KDWW19,KDWW21}\ as special cases. As part of this contribution, we also define resource theories of Bose asymmetry, which to our knowledge have not been considered yet. This development shows that the acceptance probabilities of the aforementioned algorithms, i.e., maximum symmetric fidelities, are resource monotones and thus well-motivated from the resource-theoretic perspective.

In what follows, we proceed in the aforementioned order, and we finally
conclude in Section~\ref{sec:conclusion}\ with a brief summary and a
discussion of future questions.

\section{Notions of symmetry}

\label{sec:sym-ext-G}

We introduce the notions of $G$-symmetric
extendibility and $G$-Bose symmetric extendibility, as generalizations of the
notions of $G$-symmetry \cite[Section~2]{MS13} and extendibility
\cite{W89a,DPS02,DPS04}. Later on in Section~\ref{sec:tests-o-sym}, we devise
quantum algorithms to test for these symmetries.

Let $\rho_{S}$ be a quantum state of system $S$ with corresponding Hilbert space $\mathcal{H}_{S}$. Let $G$ be a finite group, and let $U_{RS}(g)$ be a unitary representation \cite[Section~2]{MS13}\ of the group element $g\in G$, where $R$ indicates another Hilbert space such that $U_{RS}(g)$ acts on the tensor-product Hilbert space $\mathcal{H}_{R} \otimes\mathcal{H}_{S}$. Let $\Pi_{RS}^{G}$ denote the following projection operator:
\begin{equation}
\Pi_{RS}^{G}\coloneqq\frac{1}{\left\vert G\right\vert }\sum_{g\in G}U_{RS}(g).
\end{equation}
Observe that
\begin{equation}\label{eq:unitaries-and-projs}
\Pi_{RS}^{G}=U_{RS}(g)\Pi_{RS}^{G}=\Pi_{RS}^{G}U_{RS}(g),
\end{equation}
for all $g\in G$, which follows from what is called the rearrangement theorem in group theory.

We now define $G$-symmetric extendible and $G$-Bose-symmetric extendible states.

\begin{definition}
[$G$-symmetric extendible]\label{def:g-sym-ext}A state $\rho_{S}$ is
$G$-symmetric extendible if there exists a state $\omega_{RS}$ such that

\begin{enumerate}
\item the state $\omega_{RS}$ is an extension of $\rho_{S}$, i.e.,
\begin{equation}
\operatorname{Tr}_{R}[\omega_{RS}]=\rho_{S}, \label{eq:G-ext-1}
\end{equation}

\item the state $\omega_{RS}$ is $G$-invariant, in the sense that
\begin{equation}
\omega_{RS}=U_{RS}(g)\omega_{RS}U_{RS}(g)^{\dag}\qquad\forall g\in G.
\label{eq:G-ext-2}
\end{equation}

\end{enumerate}
\end{definition}

\begin{definition}
[$G$-Bose symmetric extendible]\label{def:g-bose-sym-ext}A state $\rho_{S}$ is
$G$-Bose symmetric extendible (G-BSE) if there exists a state $\omega_{RS}$ such that

\begin{enumerate}
\item the state $\omega_{RS}$ is an extension of $\rho_{S}$, i.e.,
\begin{equation}
\operatorname{Tr}_{R}[\omega_{RS}]=\rho_{S},
\end{equation}

\item the state $\omega_{RS}$ satisfies
\begin{equation}
\omega_{RS}=\Pi_{RS}^{G}\omega_{RS}\Pi_{RS}^{G}.
\label{eq:bose-G-sym-ext-cond}
\end{equation}

\end{enumerate}
\end{definition}

Note that the condition in \eqref{eq:bose-G-sym-ext-cond} is equivalent to
$\omega_{RS}=\Pi_{RS}^{G}\omega_{RS}$ or $\omega_{RS}=U_{RS}(g)\omega_{RS}$
for all $g\in G$.
Also, observe that $\rho_{S}$ is $G$-symmetric extendible if it is $G$-Bose symmetric extendible, but the opposite implication does not necessarily hold.

We have made no assumptions about the unitary representation used thus far. It is important to mention the case of projective unitary representations, due to their physical relevance in the case of symmetries of density operators. See, e.g., Eqs.~(1.2) and (1.3) of \cite{marvian2012symmetry} for a definition of a projective unitary representation. Restricting to projective unitary representations helps in avoiding trivial representations, and when considering symmetries of density operators, they necessarily arise. Furthermore, when considering example algorithms in later sections, we limit ourselves to faithful representations of the groups involved. In principle, neither faithfulness nor a projective representation are required unless stated otherwise. The choice of representation does matter when considering the symmetry of a state; however, following conventions in existing literature, we describe all symmetries with respect to the group and omit the reliance on the representation in notation.

The notions of symmetry from Definitions~\ref{def:g-sym-ext} and \ref{def:g-bose-sym-ext} generalize both $k$-extendibility of bipartite states and $G$-symmetry of unipartite states, as we discuss below.

\begin{example}
[$k$-extendible]\label{ex:k-ext}Recall that a bipartite state $\rho_{AB}$ is
$k$-extendible \cite{W89a,DPS02,DPS04}\ if there exists an extension state
$\omega_{AB_{1}\cdots B_{k}}$ such that
\begin{equation}
\operatorname{Tr}_{B_{2}\cdots B_{k}}[\omega_{AB_{1}\cdots B_{k}}]=\rho_{AB}
\end{equation}
and
\begin{equation}
\omega_{AB_{1}\cdots B_{k}}=W_{B_{1}\cdots B_{k}}(\pi)\omega_{AB_{1}\cdots
B_{k}}W_{B_{1}\cdots B_{k}}(\pi)^{\dag},
\end{equation}
for all $\pi\in S_{k}$, where each system $B_{1}$, \ldots, $B_{k}$ is
isomorphic to the system $B$ and $W_{B_{1}\cdots B_{k}}(\pi)$ is a unitary
representation of the permutation $\pi\in S_{k}$, with $S_{k}$ the symmetric
group. Then the established notion of $k$-extendibility is a special case of
$G$-symmetric extendibility, in which we set
\begin{align}
S  &  =AB_{1},\label{eq:ident-k-to-g-1}\\
R  &  =B_{2}\cdots B_{k},\\
G  &  =S_{k},\\
U_{RS}(g)  &  =\mathbb{I}_{A}\otimes W_{B_{1}\cdots B_{k}}(\pi).
\label{eq:ident-k-to-g-4}
\end{align}

\end{example}

\begin{example}
[$k$-Bose-extendible]\label{ex:k-bose-ext}A bipartite state $\rho_{AB}$ is
$k$-Bose-extendible if there exists an extension state $\omega_{AB_{1}\cdots
B_{k}}$ such that
\begin{equation}
\operatorname{Tr}_{B_{2}\cdots B_{k}}[\omega_{AB_{1}\cdots B_{k}}]=\rho_{AB}
\end{equation}
and
\begin{equation}
\omega_{AB_{1}\cdots B_{k}}=\Pi_{B_{1}\cdots B_{k}}^{\operatorname{Sym}}
\omega_{AB_{1}\cdots B_{k}}\Pi_{B_{1}\cdots B_{k}}^{\operatorname{Sym}},
\end{equation}
where
\begin{equation}
\Pi_{B_{1}\cdots B_{k}}^{\operatorname{Sym}}\coloneqq\frac{1}{k!}\sum_{\pi\in
S_{k}}W_{B_{1}\cdots B_{k}}(\pi) \label{eq:sym-subspace-proj}
\end{equation}
is the projection onto the symmetric subspace. Thus, $k$-Bose-extendibility is
a special case of $G$-Bose-symmetric extendibility under the identifications
in \eqref{eq:ident-k-to-g-1}--\eqref{eq:ident-k-to-g-4}.
\end{example}

\begin{example}
[$G$-symmetric]\label{ex:usual-symmetry}Let $G$ be a group with projective
unitary representation $\{U_{S}(g)\}_{g\in G}$, and let $\rho_{S}$ be a
quantum state of system $S$. A state $\rho_{S}$ is symmetric with respect to
$G$ \cite{MS13,MS14}\ if
\begin{equation}
\rho_{S}=U_{S}(g)\rho_{S}U_{S}(g)^{\dag}\quad\forall g\in G.
\end{equation}
Thus, the established notion of symmetry of a state $\rho_{S}$ with respect to
a group $G$ is a special case of $G$-symmetric extendibility in which the
system $R$ is trivial.
\end{example}

\begin{example}[$G$-Bose-symmetric]
A state $\rho_{S}$ is Bose-symmetric with respect to $G$ if
\begin{equation}
\rho_{S}=U_{S}(g)\rho_{S}\quad\forall g\in G.
\label{eq:def-G-Bose-sym}
\end{equation}
The condition in \eqref{eq:def-G-Bose-sym} is equivalent to the condition
\begin{equation}
\rho_{S}=\Pi_{S}^{G}\rho_{S}\Pi_{S}^{G},
\end{equation}
where the projector $\Pi_{S}^{G}$ is defined as
\begin{equation}
\label{eq:group_proj_GBS}
\Pi_{S}^{G}\coloneqq \frac{1}{\left\vert G\right\vert }\sum_{g\in G}U_{S}(g).
\end{equation}
Thus, the established notion of Bose symmetry of a state~$\rho_{S}$ with respect to a group $G$ is a special case of $G$-Bose symmetric extendibility in which the system $R$ is trivial.
\label{ex:usual-Bose-symmetry}
\end{example}

Although the concepts of $G$-symmetric extendibility and $G$-Bose-symmetric extendibility, in Definitions~\ref{def:g-sym-ext} and~\ref{def:g-bose-sym-ext}, respectively, are generally different, we can relate them by purifying a $G$-symmetric extendible state to a larger Hilbert space, as stated in Theorem~\ref{thm:Bose-sym-purify} below. The ability to do so plays a critical role in the algorithms proposed in Section~\ref{sec:tests-o-sym}. We give a proof of Theorem~\ref{thm:Bose-sym-purify} in Appendix~\ref{app:Bose-sym-purify}.

\begin{theorem}
\label{thm:Bose-sym-purify}A state $\rho_{S}$ is $G$-symmetric extendible if
and only if there exists a purification $\psi_{RS\hat{R}\hat{S}}^{\rho}$ of
$\rho_{S}$ satisfying the following:
\begin{equation}
|\psi^{\rho}\rangle_{RS\hat{R}\hat{S}}=\left(  U_{RS}(g)\otimes\overline
{U}_{\hat{R}\hat{S}}(g)\right)  |\psi^{\rho}\rangle_{RS\hat{R}\hat{S}}
\quad\forall g\in G, \label{eq:g-sym-pur-cond}
\end{equation}
where the overbar denotes the entrywise complex conjugate. The condition in
\eqref{eq:g-sym-pur-cond} is equivalent to
\begin{equation}
|\psi^{\rho}\rangle_{RS\hat{R}\hat{S}}=\Pi_{RS\hat{R}\hat{S}}^{G}|\psi^{\rho
}\rangle_{RS\hat{R}\hat{S}}, \label{eq:g-sym-pur-cond-proj}
\end{equation}
where
\begin{equation}
\Pi_{RS\hat{R}\hat{S}}^{G}\coloneqq \frac{1}{\left\vert G\right\vert }
\sum_{g\in G}U_{RS}(g)\otimes\overline{U}_{\hat{R}\hat{S}}(g).
\label{eq:projector-ref-unitaries}
\end{equation}

\end{theorem}

\section{Testing symmetry and extendibility on quantum computers}
\label{sec:tests-o-sym}

\begin{table*}
\centering
\begin{tabular}
[c]{c|c|c}\hline\hline
\multicolumn{1}{c}{Test} & \multicolumn{1}{|c|}{Algorithm} &
\multicolumn{1}{c}{Acceptance Probability}\\\hline\hline
$G$-Bose symmetry & \ref{alg:simple} & $\max_{\sigma\in\text{B-Sym}_{G}}F(\rho,\sigma
)$\\\hline
$G$-symmetry & \ref{alg:g-sym-test} & $\max_{\sigma\in\text{Sym}_{G}}F(\rho,\sigma)$\\\hline
$G$-Bose symmetric extendibility & \ref{alg:G-BSE-test} & $\max_{\sigma\in\text{BSE}_{G}}
F(\rho,\sigma)$\\\hline
$G$-symmetric extendibility & \ref{alg:sym-ext} & $\max_{\sigma\in\text{SymExt}_{G}}
F(\rho,\sigma)$\\\hline\hline
\end{tabular}
\caption{Summary of the various symmetry tests proposed in Section~\ref{sec:tests-o-sym} and their acceptance probabilities. For more details, see Theorems~\ref{thm:acc-prob-g-Bose-sym}, \ref{thm:max-acc-prob-g-sym}, \ref{thm:G-BSE-acc-prob}, and \ref{thm:G-SE-acc-prob}.}
\label{tbl:theory-summary}
\end{table*}

We can use a quantum computer to test for $G$-symmetric
extendibility of a quantum state, as well as for other forms of symmetry discussed in the previous section. We assume the following in doing so:

\begin{enumerate}
\item there is a quantum circuit available that prepares a purification
$\psi_{S^{\prime}S}^{\rho}$ of the state $\rho_{S}$,

\item there is an efficient implementation of each of the unitary operators in
the set $\{U_{RS}(g)\}_{g\in G}$,

\item and there is an efficient implementation of each of the unitary
operators in the set $\{\overline{U}_{RS}(g)\}_{g\in G}$.
\end{enumerate}

The first assumption can be made less restrictive by employing the variational, purification-learning procedure from \cite{CSZW20}. That is, given a circuit that prepares the state~$\rho_{S}$, the variational algorithm from \cite{CSZW20} outputs a circuit that approximately prepares a purification of~$\rho_{S}$. We should note that the convergence of the algorithm from \cite{CSZW20} has not been established, and so the first assumption might be necessary for some applications. See also \cite{EBSWSCH22}.

The last assumption can be relaxed by the following reasoning: a standard gate set for approximating arbitrary unitaries in quantum computing consists of the controlled-NOT gate, the Hadamard gate, and the $T$ gate \cite{book2000mikeandike}. The first two gates have only real entries while the $T$ gate is a diagonal $2\times2$ unitary gate with the entries $1$ and $e^{i\pi/4}$. The complex conjugate of this gate is equal to $T^{\dag}$. Thus, if a circuit for $U_{RS}(g)$ is constructed from this standard gate set, then we can generate a circuit for $\overline{U}_{RS}(g)$ by replacing every $T$ gate in the original circuit with $T^{\dag}$.

We now consider various quantum computational tests of symmetry that have increasing complexity. Table~\ref{tbl:theory-summary} summarizes the main theoretical insight of this section, which is that the acceptance probability of each symmetry test can be expressed in terms of the fidelity of the state being tested to a set of symmetric states.

To give insight along the way, we provide an example along with the tests below. In particular, we consider the dihedral group of the triangle, $D_3$, which has order six and is isomorphic to the symmetric group on three elements, the smallest non-abelian group. Recall that dihedral groups are the symmetry groups of regular polygons. 

Our example $D_3$ is generated via a flip $f$ and a rotation $r$: $\langle e, f, r\ \vert\ r^3 = e, f^2 = e, frf = r^{-1} \rangle$. The group thus has six elements $\{e, f, r, r^2, fr, fr^2\}$, where $e$ is the identity element. We will specify elements $r^2, fr, fr^2$ in order to enforce the rules of the group. \\

The group table for this dihedral group is given by
\begin{center}
\begin{tabular}{
>{\centering\arraybackslash}p{0.06\textwidth} | >{\centering\arraybackslash}p{0.03\textwidth} | 
>{\centering\arraybackslash}p{0.03\textwidth} | 
>{\centering\arraybackslash}p{0.03\textwidth} | 
>{\centering\arraybackslash}p{0.03\textwidth} | 
>{\centering\arraybackslash}p{0.03\textwidth} | 
>{\centering\arraybackslash}p{0.03\textwidth}}
 \hline
  Group element & $e$ & $f$ & $r$ & $r^2$ & $fr$ & $fr^2$ \\ 
  \hline
 $e$      & $e$      & $f$      & $r$      & $r^2$  & $fr$     & $fr^2$\\ 
 $f$      & $f$      & $e$      & $fr$     & $fr^2$ & $r$      & $r^2$ \\ 
 $r$      & $r$      & $fr^2$ & $r^2$  & $e$      & $f$      & $fr$    \\ 
 $r^2$  & $r^2$  & $fr$     & $e$      & $r$      & $fr^2$ & $f$     \\ 
 $fr$     & $fr$     & $r^2$  & $fr^2$ & $f$      & $e$      & $r$     \\ 
 $fr^2$ & $fr^2$ & $r$      & $f$      & $fr$     & $r^2$  & $e$     \\ 
 \hline
\end{tabular}
\end{center}

To fully realize $D_3$, we use a two-qubit unitary representation and specify the generators as such: $\{e \rightarrow \mathbb{I}, f \rightarrow \textrm{CNOT}, r \rightarrow \textrm{CNOT} \circ \textrm{SWAP}\}$. A quick check confirms that these generators obey the commutation rules of the group and generate the table above. Throughout the next four sections, we substitute this group into the presented algorithms to demonstrate their construction.

\subsection{Testing \texorpdfstring{$G$}{G}-Bose symmetry}

\label{sec:simple-algorithm}

Let us begin by discussing the simplest version of the problem. Suppose that the state under consideration is pure, so that we can write it as $\psi_{S}\equiv|\psi\rangle\!\langle\psi|_{S}$, and suppose that the $R$ system is trivial. We recover the traditional case of $G$-Bose symmetry mentioned in Example~\ref{ex:usual-Bose-symmetry}. Thus, our goal is to decide if
\begin{equation}
|\psi\rangle_{S}=U_{S}(g)|\psi\rangle_{S}\quad\forall g\in G.
\end{equation}
This condition is equivalent to
\begin{equation}
\label{eq:GBSProjTest}
|\psi\rangle_{S}=\Pi_{S}^{G}|\psi\rangle_{S},
\end{equation}
where
\begin{equation}
\Pi_{S}^{G}\coloneqq\frac{1}{\left\vert G\right\vert }\sum_{g\in G}U_{S}(g),
\label{eq:simple-case-g-proj}
\end{equation}
which is in turn equivalent to
\begin{equation}
\left\Vert \Pi_{S}^{G}|\psi\rangle_{S}\right\Vert _{2}=1.
\label{eq:proj-pure-state}
\end{equation}
The equivalence
\begin{equation}
|\psi\rangle_{S}=\Pi_{S}^{G}|\psi\rangle_{S}\quad\Leftrightarrow
\quad\left\Vert \Pi_{S}^{G}|\psi\rangle_{S}\right\Vert _{2}=1
\label{eq:equiv-norm-eq-1-iff-proj}
\end{equation}
holds from the Pythagorean theorem and the positive definiteness of the norm. Indeed,
\begin{equation}
\left\Vert \Pi_{S}^{G}|\psi\rangle_{S}\right\Vert _{2}  =1 
\quad \Rightarrow  \quad \left\Vert\Pi_{S}^{G}|\psi\rangle_{S}\right\Vert _{2}^2  =1 = \left\Vert|\psi\rangle_{S}\right\Vert _{2}^2
\end{equation}
and since the Pythagorean theorem states that
\begin{equation}
\left\Vert\Pi_{S}^{G}|\psi\rangle_{S}\right\Vert _{2}^2 + \left\Vert (\mathbb{I}_S-\Pi_{S}^{G})|\psi\rangle_{S}\right\Vert _{2}^2 = \left\Vert|\psi\rangle_{S}\right\Vert _{2}^2, 
\end{equation}
we conclude that $\left\Vert (\mathbb{I}_S-\Pi_{S}^{G})|\psi\rangle_{S}\right\Vert _{2} = 0$, which implies that $ (\mathbb{I}_S-\Pi_{S}^{G})|\psi\rangle_{S} = 0$ from the positive definiteness of the norm. This in turn is equivalent to the left-hand side of \eqref{eq:equiv-norm-eq-1-iff-proj}. 
%\begin{equation}
%\left\Vert |\psi\rangle_{S}-\Pi_{S}^{G}|\psi\rangle_{S}\right\Vert _{2}^{2}=0
%\end{equation}
Thus, if we have a method to perform the projection onto $\Pi_{S}^{G}$, then we can decide whether \eqref{eq:proj-pure-state} holds.

There is a simple quantum algorithm to do so. This algorithm was originally proposed in \cite[Chapter~8]{harrow2005applications} under the name of ``generalized phase estimation.'' It proceeds as follows and can be summarized as \textquotedblleft performing the quantum phase estimation algorithm with respect to the unitary representation $\{U_{S}(g)\}_{g\in G}$\textquotedblright:

\begin{Algorithm}
[$G$-Bose symmetry test]\label{alg:simple} The algorithm consists of the following steps:

\begin{enumerate}
\item Prepare an ancillary register $C$ in the state $|0\rangle_{C}$.

\item Act on register $C$ with a quantum Fourier transform.

\item Append the state $|\psi\rangle_{S}$ and perform the following controlled unitary:
\begin{equation}
\sum_{g\in G}|g\rangle\!\langle g|_{C}\otimes U_{S}(g).
\end{equation}

\item Perform an inverse quantum Fourier transform on register~$C$, measure in the basis $\{|g\rangle\!\langle g|_{C}\}_{g\in G}$, and accept if and only if the zero outcome $|0\rangle\!\langle0|_{C}$ occurs.
\end{enumerate}
\end{Algorithm}

\begin{figure}[ptb]
\begin{center}
\includegraphics[
width=\linewidth
]{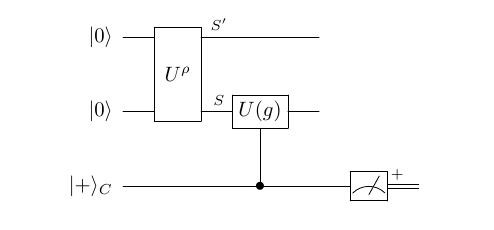}
\end{center}
\caption{Quantum circuit to implement Algorithm~\ref{alg:simple}. The unitary $U^{\rho}$ prepares a purification $\psi_{S^{\prime}S}$ of the state $\rho_{S}$. The final measurement box with the plus-sign to the right of it indicates that the measurement $\{|+\rangle\!\langle +|_C, \mathbb{I}_C - |+\rangle\!\langle +|_C\}$ is performed. (We use this same notation in several forthcoming figures.) Algorithm~\ref{alg:simple} tests whether the state $\rho_{S}$ is $G$-Bose symmetric, as defined in Example~\ref{ex:usual-Bose-symmetry}. Its acceptance probability is equal to $\operatorname{Tr}[\Pi_{S}^{G}\rho_{S}]$, where $\Pi_{S}^{G}$ is defined in \eqref{eq:simple-case-g-proj}.}
\label{fig:simple-case}
\end{figure}

Note that the register $C$ has dimension $\left\vert G\right\vert $. Also, we can write the state $|0\rangle_{C}$ as $|e\rangle_{C}$, where $e$ is the identity element of the group. The result of Step~2 of Algorithm~\ref{alg:simple}\ is to prepare the following uniform superposition state:
\begin{equation}
|+\rangle_{C}\coloneqq \frac{1}{\sqrt{\left\vert G\right\vert }}\sum_{g \in G}|g\rangle_{C}. \label{eq:plus-over-group}
\end{equation}
Although the quantum Fourier transform is specified in Algorithm~\ref{alg:simple}, in fact, any unitary that generates the desired superposition state $|+\rangle_{C}$ can serve as a replacement in Steps 2 and 4 above and oftentimes leads to an improvement in circuit depth. The same is true for all algorithms that follow.

Moving on, the overall state after Step~3 is as follows:
\begin{equation}
\frac{1}{\sqrt{\left\vert G\right\vert }}\sum_{g\in G}|g\rangle_{C} U_{S}(g)|\psi\rangle_{S}.
\end{equation}
The final step of Algorithm~\ref{alg:simple}\ projects the register $C$ onto the state $|+\rangle_{C}$.  According to the aforementioned convention, Algorithm~\ref{alg:simple}\ accepts if the identity element outcome $|e\rangle\!\langle e|_{C}$ occurs. The probability that Algorithm~\ref{alg:simple} accepts is equal to
\begin{align}
&  \left\Vert \left(  \langle+|_{C}\otimes \mathbb{I}_S\right)  \left(  \frac{1}{\sqrt{\left\vert G\right\vert}}\sum_{g\in G}|g\rangle_{C}U_{S} (g)|\psi\rangle_{S}\right) \right\Vert _{2}^{2}\nonumber\\
&  =\left\Vert \frac{1}{\left\vert G\right\vert} \sum_{g\in G} U_{S}(g)|\psi\rangle_{S}\right\Vert_{2}^{2}\label{eq:simple-acc-prob-1}\\
&  =\left\Vert \Pi_{S}^{G}|\psi\rangle_{S}\right\Vert _{2}^{2}.
\label{eq:simple-acc-prob-2}
\end{align}

Figure~\ref{fig:simple-case} depicts this quantum algorithm. Not only does it decide whether the state $|\psi\rangle_{S}$ is symmetric, but it also quantifies how symmetric the state is.  Since the acceptance probability is equal to $\left\Vert \Pi_{S}^{G}|\psi\rangle_{S}\right\Vert _{2}^{2}$, and this quantity is a measure of symmetry (see Theorem~\ref{thm:res-mono-G-Bose-sym}), we can repeat the algorithm a large number of times to estimate the acceptance probability to arbitrary precision.

The same quantum algorithm can decide whether a given mixed state $\rho_{S}$ is $G$-Bose symmetric (see Example~\ref{ex:usual-Bose-symmetry}). Similar to the above, it also can estimate how $G$-Bose symmetric the state $\rho_{S}$ is. To see this, consider that the acceptance probability for a pure state can be rewritten as follows:
\begin{equation}
\left\Vert \Pi_{S}^{G}|\psi\rangle_{S}\right\Vert _{2}^{2}=\operatorname{Tr}[\Pi_{S}^{G}|\psi\rangle\!\langle\psi|_{S}].
\label{eq:acc-prob-bose-test-pure-state-1}
\end{equation}
Then since every mixed state can be written as a probabilistic mixture of pure states, it follows that the acceptance probability of Algorithm~\ref{alg:simple}, when acting on the mixed state $\rho_{S}$, is equal to
\begin{equation}
\operatorname{Tr}[\Pi_{S}^{G}\rho_{S}]. \label{eq:acc-prob-bose-test}
\end{equation}
This acceptance probability is equal to one if and only if $\rho_{S}=\Pi_{S}^{G}\rho_{S}\Pi_{S}^{G}$, and so this test is a faithful test of $G$-Bose symmetry. The equivalence
\begin{equation}
\operatorname{Tr}[\Pi_{S}^{G}\rho_{S}]=1\quad\Leftrightarrow\quad\rho_{S} =\Pi_{S}^{G}\rho_{S}\Pi_{S}^{G} \label{eq:Bose-symmetric-equiv-cond}
\end{equation}
follows as a limiting case of the gentle measurement lemma \cite{itit1999winter,ON07} (see also \cite[Lemma~9.4.1]{Wbook17}):
\begin{equation}
\frac{1}{2}\left\Vert \rho_{S}-\frac{\Pi_{S}^{G}\rho_{S}\Pi_{S}^{G}}{\operatorname{Tr}[\Pi_{S}^{G}\rho_{S}]}\right\Vert_{1}\leq\sqrt{1-\operatorname{Tr}[\Pi_{S}^{G}\rho_{S}]}
\end{equation}
and the positive definiteness of the trace norm. Again, through repetition, we can estimate the acceptance probability $\operatorname{Tr}[\Pi_{S}^{G}\rho_{S}]$ and then employ it as a measure of $G$-Bose symmetry (see
Theorem~\ref{thm:res-mono-G-Bose-sym}).

Interestingly, the acceptance probability of Algorithm~\ref{alg:simple} can be expressed as the \textit{maximum $G$-Bose-symmetric fidelity}, defined for a state $\rho_S$ as
\begin{equation}
\max_{\sigma_{S}\in
\operatorname{B-Sym}_{G}}F(\rho_{S},\sigma_{S}),
\end{equation}
where
\begin{equation}
\operatorname{B-Sym}_{G}\coloneqq \left\{  \sigma_{S}\in\mathcal{D}(\mathcal{H}%
_{S}):\sigma_{S}=\Pi_{S}^{G}\sigma_{S}\Pi_{S}^{G}\right\}  ,
\end{equation}
and the fidelity of quantum states $\omega$ and $\tau$  is defined as \cite{U76}
\begin{equation}
F(\omega,\tau)\coloneqq \left\Vert \sqrt{\omega}\sqrt{\tau}\right\Vert_{1}^{2}.
\end{equation}
We state this claim in Theorem~\ref{thm:acc-prob-g-Bose-sym} below and provide a proof of Theorem~\ref{thm:acc-prob-g-Bose-sym} in Appendix~\ref{app:acc-prob-g-Bose-sym}.  
Thus, Algorithm~\ref{alg:simple} gives an operational meaning to the maximum $G$-Bose-symmetric fidelity in terms of its acceptance probability, and it can be used to estimate this fundamental measure of symmetry.

\begin{theorem}\label{thm:acc-prob-g-Bose-sym}
For a state $\rho_{S}$, the acceptance probability of Algorithm~\ref{alg:simple} is equal to the maximum $G$-Bose
symmetric fidelity. That is,
\begin{equation}
\operatorname{Tr}[\Pi_{S}^{G}\rho_{S}]=\max_{\sigma_{S}\in
\operatorname{B-Sym}_{G}}F(\rho_{S},\sigma_{S}).
\end{equation}
\end{theorem}

\begin{figure}
\begin{center}
\includegraphics[width=0.65\linewidth]{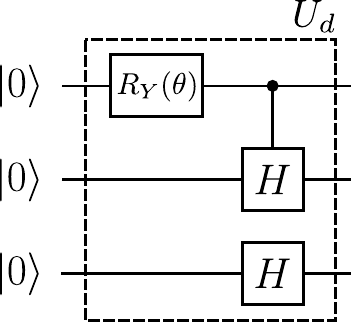}
\end{center}
\caption{Unitary $U_d$, with $\theta = 2 \arctan \!\left(\frac{1}{\sqrt{2}}\right)$, generates the equal superposition of six elements from \eqref{eq:Dihedral_Superposition}. Note that the controlled-Hadamard is controlled on the qubit being in the state zero.}
\label{fig:Dihedral_Superposition}
\end{figure}

\begin{example}
\label{ex:dih-gr-d3}
In the example of the dihedral group $D_3$, the $\ket{+}_C$ state is a uniform superposition of six elements. We use three qubits and the unitary $U_d$ shown in Figure~\ref{fig:Dihedral_Superposition} to generate an equal superposition of six elements:
\begin{multline}
    \label{eq:Dihedral_Superposition}
    U_d\ket{000} = \frac{1}{\sqrt{6}} (\ket{000} + \ket{001} + \ket{010} + \\ \ket{011} + \ket{100} + \ket{101}).
\end{multline}
These control register states need to be mapped to group elements to be meaningful; thus, we employ the mapping $\{\ket{000} \rightarrow e, \ket{001} \rightarrow fr^2, \ket{010} \rightarrow fr, \ket{011} \rightarrow r, \ket{100} \rightarrow f, \ket{101} \rightarrow r^2\}$ for our circuit constructions. The circuit to test for $D_3$-symmetry is shown in Figure~\ref{fig:Dihedral_GBS_Circuit}.
\end{example}

\begin{figure}
\begin{center}
\includegraphics[width=\linewidth]{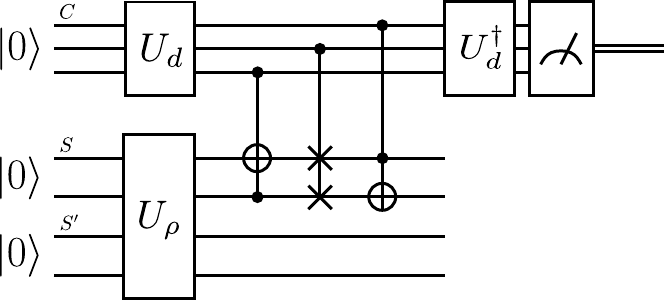}
\end{center}
\caption{Quantum circuit implementing Algorithm~\ref{alg:simple} to test $G$-Bose symmetry for $D_3$. Compared to Figure~\ref{fig:simple-case}, the systems $S$ and $S'$ are two qubits each, $C$ consists of three qubits, and $\ket{+}_C$ is defined as $U_d\ket{000}$.}
\label{fig:Dihedral_GBS_Circuit}
\end{figure}

\subsection{Testing \texorpdfstring{$G$}{G}-symmetry}

\label{sec:G-sym-single-S-triv-R} We now discuss how to modify Algorithm~\ref{alg:simple}\ to one that decides whether a state $\rho_{S}$ is $G$-symmetric (see Example~\ref{ex:usual-symmetry}), i.e., if
\begin{equation}
\rho_{S}=U_{S}(g)\rho_{S}U_{S}(g)^{\dag}\quad\forall g\in G\, .
\label{eq:rho-symmetric-single-sys}
\end{equation}
We also prove that the acceptance probability of the modified algorithm (Algorithm~\ref{alg:g-sym-test} below) is equal to the \textit{maximum }$G$\textit{-symmetric fidelity}, defined as
\begin{equation}
\max_{\sigma\in\text{Sym}_{G}}F(\rho_{S},\sigma_{S}), \label{eq:fid-of-asym}
\end{equation}
where
\begin{multline}
\operatorname{Sym}_{G}\coloneqq \\
\left\{ \sigma_{S}\in \mathcal{D} (\mathcal{H}_{S}):\sigma_{S} = U_{S}(g)\sigma_{S}U_{S}(g)^{\dag}\ \forall g\in G\right\}  ,
\end{multline}
and $\mathcal{D}(\mathcal{H}_{S})$ denotes the set of density operators acting on the Hilbert space $\mathcal{H}_{S}$.
Thus, Algorithm~\ref{alg:g-sym-test} gives an operational meaning to the maximum $G$-symmetric fidelity in terms of its acceptance probability, and it can be used to estimate this fundamental measure of symmetry.

In the modified approach, we suppose that the quantum computer (now called the verifier) is equipped with access to a ``quantum prover''---an agent who can perform arbitrarily powerful quantum computations. We suppose that the quantum computer is allowed to exchange two quantum messages with the prover. The resulting class of problems that can be solved using this approach is abbreviated QIP(2), for quantum interactive proofs with two quantum messages exchanged \cite{W09,VW15}, and we note here that computational problems related to entanglement of bipartite states \cite{HMW13,Hayden:2014:TQI}\ and recoverability of tripartite states \cite{PhysRevA.94.022310}\ were previously shown to be decidable in QIP(2). These latter problems were proven to be QSZK-hard, and it remains an open question to determine their precise computational complexity.

Let $|\psi\rangle_{S^{\prime}S}$ be a purification of the state $\rho_{S}$, and suppose that the verifier has access to a circuit $U^{\rho}$\ that prepares this purification of $\rho_{S}$.

\begin{Algorithm}
[$G$-symmetry test]\label{alg:g-sym-test}The algorithm consists of the following steps:

\begin{enumerate}
\item The verifier uses the circuit $U^{\rho}$ to prepare the state $|\psi\rangle_{S^{\prime}S}$.

\item The verifier transmits the purifying system $S^{\prime}$ to the prover.

\item The prover appends an ancillary register $E$ in the state $|0\rangle_{E}$ and performs a unitary $V_{S^{\prime}E\rightarrow\hat{S}E^{\prime}}$.

\item The prover sends the system $\hat{S}$ back to the verifier.

\item The verifier prepares a register $C$ in the state $|0\rangle_{C}$.

\item The verifier acts on register $C$ with a quantum Fourier transform.

\item The verifier performs the following controlled unitary:
\begin{equation}
\sum_{g\in G}|g\rangle\!\langle g|_{C}\otimes U_{S}(g)\otimes\overline{U}_{\hat{S}}(g).
\end{equation}

\item The verifier performs an inverse quantum Fourier transform on register~$C$, measures in the basis $\{|g\rangle\!\langle g|_{C}\}_{g\in G}$, and accepts if and only if the zero outcome $|0\rangle\!\langle0|_{C}$ occurs.
\end{enumerate}
\end{Algorithm}

Figure~\ref{fig:case-2} depicts this quantum algorithm. The overall state after Step~3 of Algorithm~\ref{alg:g-sym-test}\ is
\begin{equation}
V_{S^{\prime}E\rightarrow\hat{S}E^{\prime}}|\psi\rangle_{S^{\prime}S} |0\rangle_{E}.
\end{equation}
The result of Step~6 is to prepare the uniform superposition state $|+\rangle_{C}$, which is defined in \eqref{eq:plus-over-group}. After Step~7, the overall state is
\begin{equation}
\frac{1}{\sqrt{\left\vert G\right\vert }}\sum_{g\in G}|g\rangle_{C}\left(U_{S}(g) \otimes \overline{U}_{\hat{S}}(g)\right)  V_{S^{\prime}E \rightarrow \hat{S}E^{\prime}}|\psi\rangle_{S^{\prime}S}|0\rangle_{E}.
\end{equation}

For a fixed unitary $V_{S^{\prime}E\rightarrow\hat{S}E^{\prime}}$, the
probability of accepting, by following the same reasoning in \eqref{eq:simple-acc-prob-1}--\eqref{eq:simple-acc-prob-2}, is equal to
\begin{equation}
\left\Vert \Pi_{S\hat{S}}^{G}V_{S^{\prime}E\rightarrow\hat{S}E^{\prime}}
|\psi\rangle_{S^{\prime}S}|0\rangle_{E}\right\Vert _{2}^{2},
\end{equation}
where
\begin{equation}
\Pi_{S\hat{S}}^{G}\coloneqq \frac{1}{\left\vert G\right\vert }\sum_{g\in
G}U_{S}(g)\otimes\overline{U}_{\hat{S}}(g).
\end{equation}
Since the goal of the prover in a quantum interactive proof is to convince the verifier to accept \cite{W09,VW15}, the prover optimizes over every unitary $V_{S^{\prime}E\rightarrow\hat{S}E^{\prime}}$ and the acceptance probability of Algorithm~\ref{alg:g-sym-test}\ is given by
\begin{equation}
\max_{V_{S^{\prime}E\rightarrow\hat{S}E^{\prime}}}\left\Vert \Pi_{S\hat{S} }^{G}V_{S^{\prime}E\rightarrow\hat{S}E^{\prime}}|\psi\rangle_{S^{\prime} S}|0\rangle_{E}\right\Vert _{2}^{2}.
\end{equation}

\begin{figure}[ptb]
\begin{center}
\includegraphics[
width=3.4in
]{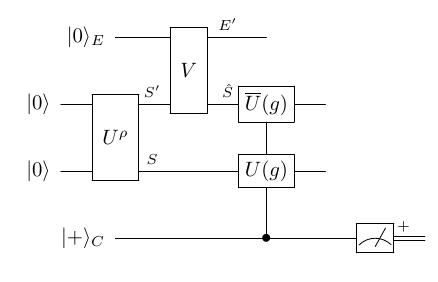}
\end{center}
\caption{Quantum circuit to implement Algorithm~\ref{alg:g-sym-test}. The
unitary $U^{\rho}$ prepares a purification $\psi_{S^{\prime}S}$ of the state
$\rho_{S}$. Algorithm~\ref{alg:g-sym-test} tests whether the state $\rho_{S}$
is $G$-symmetric, as defined in Example~\ref{ex:usual-symmetry}. Its
acceptance probability is equal to the maximum $G$-symmetric fidelity, as
defined in \eqref{eq:fid-of-asym}.}
\label{fig:case-2}
\end{figure}

The main idea behind Algorithm~\ref{alg:g-sym-test} is that if the state~$\rho_{S}$ possesses the symmetry in \eqref{eq:rho-symmetric-single-sys}, then
Theorem~\ref{thm:Bose-sym-purify} (with trivial reference system $R$)
guarantees the existence of a purification $\phi_{S\hat{S}}$ of $\rho_{S}$
such that
\begin{equation}
|\phi\rangle_{S\hat{S}}=\Pi_{S\hat{S}}^{G}|\phi\rangle_{S\hat{S}}.
\label{eq:proj-cond-purification}
\end{equation}
Since all purifications of a quantum state are related by a unitary acting on
the purifying system (see, e.g., \cite{Wbook17}), the prover is able to
apply a unitary taking the purification $|\psi\rangle_{S^{\prime}S}$ to the
purification $|\phi\rangle_{S\hat{S}}$. After the prover sends back the system
$\hat{S}$, the verifier then performs a quantum-computational test to determine if the condition in \eqref{eq:proj-cond-purification} holds. A discussion on how to choose the size of register $E$ can be found in Section~\ref{sec:var-algs}.
%added E system discussion here?
%A consequence of this fact is that although in principle the prover can maximize their action over both the $S^{\prime}$ and $E$ subsystems, it suffices to optimize over unitaries on $S^{\prime}$ and efffectively take $E$ to be a one-dimensional trivial system. This is illustrated in the proof in Appendix~\ref{app:max-acc-prob-g-sym}. {\color{red} NOT SURE ABOUT THIS PART}.

We now formally state the claim made just after \eqref{eq:rho-symmetric-single-sys}. See Appendix~\ref{app:max-acc-prob-g-sym} for a proof of Theorem~\ref{thm:max-acc-prob-g-sym}.

\begin{theorem}
\label{thm:max-acc-prob-g-sym}The acceptance probability of
Algorithm~\ref{alg:g-sym-test}\ is equal to the maximum $G$-symmetric fidelity
in \eqref{eq:fid-of-asym}, i.e.,
\begin{multline}
\max_{V_{S^{\prime}E\rightarrow\hat{S}E^{\prime}}}\left\Vert \Pi_{S\hat{S}
}^{G}V_{S^{\prime}E\rightarrow\hat{S}E^{\prime}}|\psi\rangle_{S^{\prime}
S}|0\rangle_{E}\right\Vert _{2}^{2}\\
=\max_{\sigma_{S}\in\operatorname{Sym}_{G}}F(\rho_{S},\sigma_{S}).
\end{multline}

\end{theorem}

\begin{example}
For the triangular dihedral group example (see Example~\ref{ex:dih-gr-d3}), we use the same unitary $U_d$ as in \eqref{eq:Dihedral_Superposition} to prepare the superposition $\ket{+}_C$ and the same mapping of control states to group elements. The circuit to test for $G$-symmetry is shown in Figure~\ref{fig:Dihedral_GS_Circuit}.
\end{example}

\begin{figure}
\begin{center}
\includegraphics[width=\linewidth]{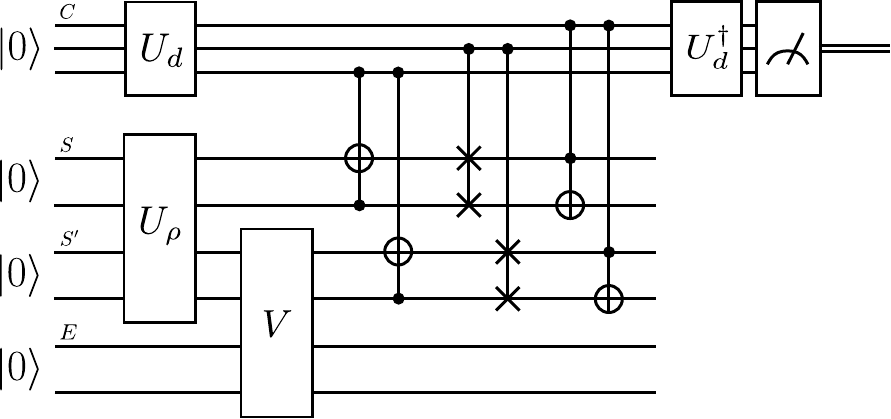}
\end{center}
\caption{Quantum circuit implementing Algorithm~\ref{alg:g-sym-test} to test $G$-symmetry in the case that the group $G$ is the triangular dihedral group. Compared to Figure~\ref{fig:case-2}, the systems $S$ and $S'$ are two qubits each, $C$ consists of three qubits, and $\ket{+}_C$ is defined as $U_d\ket{000}$. Both the SWAP and CNOT gates have no imaginary entries, and thus they are equal to their own complex conjugates.}
\label{fig:Dihedral_GS_Circuit}
\end{figure}

\begin{remark}[Testing incoherence]
\label{rem:incoherence}

We note here that testing the incoherence of a quantum state, in the sense of \cite{BCP14,SAP17}, is a special case of testing $G$-symmetry. To see this, we can pick $G$ to be the cyclic group over $d$ elements with unitary representation $\{Z(z)\}_z$, where $Z(z)$ is the generalized Pauli phase-shift unitary, defined as
\begin{equation}
    Z(z) \coloneqq \sum_{j=0}^{d-1} e^{2 \pi i j z /d} |j\rangle\!\langle j|.
\end{equation}
A state is symmetric with respect to this group if the condition in \eqref{eq:rho-symmetric-single-sys} holds. This condition is equivalent to the following one:
\begin{equation}
    \rho_S = \frac{1}{|G|} \sum_{g \in G} U_S(g) \rho_S U_S(g)^\dag.
    \label{eq:coh-sym-cond}
\end{equation}
For the choice mentioned above, the condition in \eqref{eq:coh-sym-cond} holds if and only if the state $\rho_S$ is diagonal in the incoherent basis, i.e., if it can be written as $\rho_S = \sum_j p(j) |j\rangle\!\langle j|$, where $p(j)$ is a probability distribution. Thus, Algorithm~\ref{alg:g-sym-test} can be used to test the incoherence of quantum states.
\end{remark}

\subsection{Testing \texorpdfstring{$G$}{G}-Bose symmetric extendibility}

\label{sec:G-Bose-sym-ext-test}We now describe an algorithm for testing
$G$-Bose symmetric extendibility of a quantum state $\rho_{S}$, as defined in
Definition~\ref{def:g-bose-sym-ext}. The algorithm bears some similarities
with Algorithms~\ref{alg:simple} and \ref{alg:g-sym-test}. Like
Algorithm~\ref{alg:g-sym-test}, it involves an interaction between a verifier
and a prover. We prove that its acceptance probability is equal to the maximum
$G$-BSE\ fidelity:
\begin{equation}
\max_{\sigma_{S}\in\operatorname*{BSE}_{G}}F(\rho_{S},\sigma_{S}),
\label{eq:fid-bose-asym-3}
\end{equation}
where BSE$_{G}$ is the set of $G$-Bose symmetric extendible states:
\begin{multline}
\text{BSE}_{G}\coloneqq\\
\left\{\begin{array}[c]{c}
\sigma_{S}:\exists\ \omega_{RS}\in\mathcal{D}(\mathcal{H}_{RS}
), \operatorname{Tr}_R[\omega_{RS}]=\sigma_S, \\ 
\omega_{RS}=U_{RS}(g)\omega_{RS},\ \forall g\in G
\end{array}
\right\}  .
\label{eq:G-BSE-states-set}
\end{multline}
Thus, the algorithm endows the maximum $G$-BSE fidelity with an operational
meaning. Note that the condition $\omega_{RS}=U_{RS}(g)\omega_{RS}$ for all $
g\in G$ is equivalent to
\begin{equation}
\omega_{RS}=\Pi_{RS}^{G}\omega_{RS}\Pi_{RS}^{G},
\end{equation}
where
\begin{equation}
\Pi_{RS}^{G}\coloneqq\frac{1}{\left\vert G\right\vert }\sum_{g\in G}U_{RS}(g).
\label{eq:Pi_RS-proj-again}
\end{equation}

The algorithm is similar to Algorithm~\ref{alg:g-sym-test}, but we
list it here for completeness. Let $|\psi\rangle_{S^{\prime}S}$ be a
purification of the state $\rho_{S}$, and suppose that the circuit $U^{\rho}$
prepares this purification of $\rho_{S}$.

\begin{Algorithm}
[$G$-BSE test]\label{alg:G-BSE-test}The algorithm proceeds as follows:

\begin{enumerate}
\item The verifier uses the circuit provided to prepare the state
$|\psi\rangle_{S^{\prime}S}$.

\item The verifier transmits the purifying system $S^{\prime}$ to the prover.

\item The prover appends an ancillary register $E$ in the state $|0\rangle
_{E}$ and performs a unitary $V_{S^{\prime}E\rightarrow RE^{\prime}}$.

\item The prover sends the system $R$ back to the verifier.

\item The verifier prepares a register $C$ in the state $|0\rangle_{C}$.

\item The verifier acts on register $C$ with a quantum Fourier transform.

\item The verifier performs the following controlled unitary:
\begin{equation}
\sum_{g\in G}|g\rangle\!\langle g|_{C}\otimes U_{RS}(g),
\end{equation}

\item The verifier performs an inverse quantum Fourier transform on
register~$C$, measures in the basis $\{|g\rangle\!\langle g|_{C}\}_{g\in G}$,
and accepts if and only if the zero outcome $|0\rangle\!\langle0|_{C}$ occurs.
\end{enumerate}
\end{Algorithm}

\begin{figure}[ptb]
\begin{center}
\includegraphics[
width=\linewidth
]{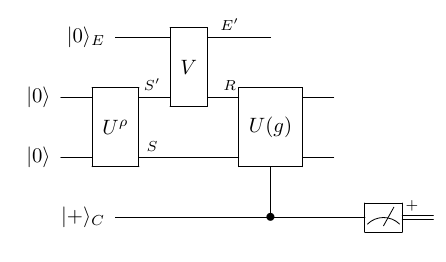}
\end{center}
\caption{Quantum circuit to implement Algorithm~\ref{alg:G-BSE-test}. The
unitary $U^{\rho}$ prepares a purification $\psi_{S^{\prime}S}$ of the state
$\rho_{S}$. Algorithm~\ref{alg:G-BSE-test} tests whether the state $\rho_{S}$
is $G$-Bose symmetric extendible, as defined in
Definition~\ref{def:g-bose-sym-ext}. Its acceptance probability is equal to
the maximum $G$-BSE fidelity, as defined in~\eqref{eq:fid-bose-asym-3}.}
\label{fig:case-3}

\end{figure}

Figure~\ref{fig:case-3} depicts this quantum algorithm. The overall state
after Step~3 is
\begin{equation}
V_{S^{\prime}E\rightarrow RE^{\prime}}|\psi\rangle_{S^{\prime}S}|0\rangle_{E}.
\end{equation}
Step~6 prepares the uniform superposition state $|+\rangle_{C}$, which is
defined in \eqref{eq:plus-over-group}. After Step~7, the overall state is
\begin{equation}
\frac{1}{\sqrt{\left\vert G\right\vert }}\sum_{g\in G}|g\rangle_{C}
U_{RS}(g)  V_{S^{\prime}E\rightarrow RE^{\prime}}|\psi\rangle
_{S^{\prime}S}|0\rangle_{E}.
\end{equation}
The last step can be understood as the verifier projecting the register $C$
onto the state $|+\rangle_{C}$.

The probability of accepting, following the
same reasoning as before,
is equal to
\begin{equation}
\left\Vert \Pi_{RS}^{G}V_{S^{\prime}E\rightarrow RE^{\prime}}|\psi
\rangle_{S^{\prime}S}|0\rangle_{E}\right\Vert _{2}^{2},
\end{equation}
where $\Pi_{RS}^{G}$ is defined in \eqref{eq:Pi_RS-proj-again}. As before, the
goal of the prover in a quantum interactive proof is to convince the verifier
to accept \cite{W09,VW15}, and so the prover optimizes over every unitary
$V_{S^{\prime}E\rightarrow\hat{S}E^{\prime}}$. The acceptance probability of
Algorithm~\ref{alg:G-BSE-test} is then given by
\begin{equation}
\max_{V_{S^{\prime}E\rightarrow RE^{\prime}}}\left\Vert \Pi_{RS}
^{G}V_{S^{\prime}E\rightarrow RE^{\prime}}|\psi\rangle_{S^{\prime}S}
|0\rangle_{E}\right\Vert _{2}^{2}.
\end{equation}

Our proof of the following theorem is similar to the proof given for
Theorem~\ref{thm:max-acc-prob-g-sym}; for completeness, we provide a proof
in Appendix~\ref{app:proof-thm-g-bse}.

\begin{theorem}
\label{thm:G-BSE-acc-prob}The maximum acceptance probability of
Algorithm~\ref{alg:G-BSE-test} is equal to the maximum $G$-BSE\ fidelity in
\eqref{eq:fid-bose-asym-3}, i.e.,
\begin{multline}
\max_{V_{S^{\prime}E\rightarrow RE^{\prime}}}\left\Vert \Pi_{RS}
^{G}V_{S^{\prime}E\rightarrow RE^{\prime}}|\psi\rangle_{S^{\prime}S}
|0\rangle_{E}\right\Vert _{2}^{2}\\
=\max_{\sigma_{S}\in\operatorname*{BSE}_{G}}F(\rho_{S},\sigma_{S}),
\end{multline}
where the set $\operatorname*{BSE}_{G}$ is defined in \eqref{eq:G-BSE-states-set}.
\end{theorem}

\begin{example}
For the triangular dihedral group example (see Example~\ref{ex:dih-gr-d3}), we use the same unitary $U_d$ to prepare the superposition $\ket{+}_C$ and the same mapping of control states to group elements. The circuit to test for $G$-Bose symmetric extendibility is shown in Figure~\ref{fig:Dihedral_GBSE_Circuit}.
\end{example}

\begin{figure}
\begin{center}
\includegraphics[width=\linewidth]{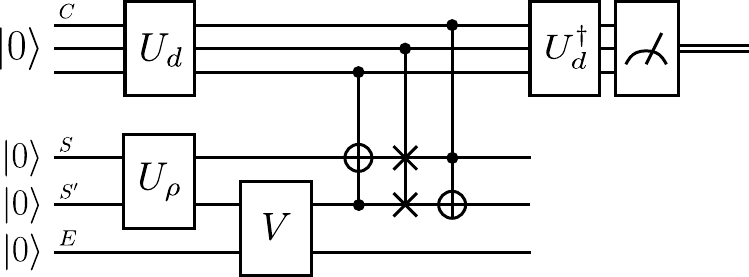}
\end{center}
\caption{Quantum circuit implementing Algorithm~\ref{alg:G-BSE-test} to test $G$-Bose symmetric extendibility for the triangular dihedral group. Compared to Figure~\ref{fig:case-3}, the systems $S$ and $S'$ are one qubit each, $C$ consists of three qubits, and $\ket{+}_C$ is defined as $U_d\ket{000}$.}
\label{fig:Dihedral_GBSE_Circuit}
\end{figure}

\subsection{Testing \texorpdfstring{$G$}{G}-symmetric extendibility}

\label{sec:test-g-sym-ext}The final algorithm that we introduce tests whether a
state $\rho_{S}$ is $G$-symmetric extendible (recall
Definition~\ref{def:g-sym-ext}). Similar to the algorithms in the previous
sections, not only does it decide whether $\rho_{S}$ is $G$-symmetric
extendible, but it also quantifies how similar it is to a state in the set of $G$-symmetric extendible states.  
The acceptance probability is equal
to the \textit{maximum }$G$\textit{-symmetric extendible fidelity}:
\begin{equation}
\max_{\sigma_{S}\in\operatorname{SymExt}_{G}}F(\rho_{S},\sigma_{S}
),\label{eq:max-g-sym-ext-fid}
\end{equation}
where
\begin{multline}
\operatorname{SymExt}_{G}\coloneqq
\\
\left\{
\begin{array}
[c]{c}
\sigma_{S}:\exists\ \omega_{RS}\in\mathcal{D}(\mathcal{H}_{RS}),\operatorname{Tr}_R[\omega_{RS}]=\sigma_S, \\
\omega_{RS}=U_{RS}(g)\omega_{RS}U_{RS}(g)^{\dag}\ \forall g\in G
\end{array}
\right\}  .\label{eq:g-sym-ext-set}
\end{multline}
We again operate in the model of quantum interactive proofs, in which a
verifier interacts with a prover.

We list the algorithm below for completeness, noting its similarity to the previous algorithms. Let $|\psi\rangle_{S^{\prime}S}$ be a
purification of the state $\rho_{S}$, and suppose that the circuit $U^{\rho}$
prepares this purification of $\rho_{S}$.

\begin{Algorithm}
[$G$-SE test]\label{alg:sym-ext}The algorithm proceeds as follows:

\begin{enumerate}
\item The verifier uses the circuit $U^{\rho}$ to prepare the state
$|\psi\rangle_{S^{\prime}S}$, which is a purification of the state $\rho_{S}$.

\item The verifier transmits the purifying system $S^{\prime}$ to the prover.

\item The prover appends an ancillary register $E$ in the state $|0\rangle
_{E}$ and performs a unitary $V_{S^{\prime}E\rightarrow R\hat{R}\hat
{S}E^{\prime}}$.

\item The prover sends the systems $R\hat{R}\hat{S}$ back to the verifier.

\item The verifier prepares a register $C$ in the state $|0\rangle_{C}$.

\item The verifier acts on register $C$ with a quantum Fourier transform.

\item The verifier performs the following controlled unitary:
\begin{equation}
\sum_{g\in G}|g\rangle\!\langle g|_{C}\otimes U_{RS}(g)\otimes\overline
{U}_{\hat{R}\hat{S}}(g),
\end{equation}

\item The verifier performs an inverse quantum Fourier transform on
register~$C$, measures in the basis $\{|g\rangle\!\langle g|_{C}\}_{g\in G}$,
and accepts if and only if the zero outcome $|0\rangle\!\langle0|_{C}$ occurs.
\end{enumerate}
\end{Algorithm}

Figure~\ref{fig:case-4} depicts this quantum algorithm. After Step~3, the
overall state is
\begin{equation}
V_{S^{\prime}E\rightarrow R\hat{R}\hat{S}E^{\prime}}|\psi\rangle_{S^{\prime}
S}|0\rangle_{E}.
\end{equation}
Step~5 prepares the uniform superposition state $|+\rangle_{C}$, which is
defined in \eqref{eq:plus-over-group}. After Step~7, the overall state is
\begin{equation}
\frac{1}{\sqrt{\left\vert G\right\vert }}\sum_{g\in G}|g\rangle_{C}\left(
U_{RS}(g)\otimes\overline{U}_{\hat{R}\hat{S}}(g)\right)  V|\psi\rangle
_{S^{\prime}S}|0\rangle_{E},
\end{equation}
where $V\equiv V_{S^{\prime}E\rightarrow R\hat{R}\hat{S}E^{\prime}}$. The last
step can be understood as the verifier projecting the register $C$ onto the
state $|+\rangle_{C}$.

The probability of accepting is equal to
\begin{equation}
\left\Vert \Pi_{RS\hat{R}\hat{S}}^{G} V_{S^{\prime}E\rightarrow R\hat{R}\hat
{S}E^{\prime}}|\psi\rangle_{S^{\prime}S}|0\rangle_{E}\right\Vert _{2}^{2},
\end{equation}
where $\Pi_{RS\hat{R}\hat{S}}^{G}$ is defined in
\eqref{eq:projector-ref-unitaries}. As before, the prover optimizes over every unitary $V_{S^{\prime}E\rightarrow R\hat{R}\hat
{S}E^{\prime}}$. The acceptance probability of
Algorithm~\ref{alg:sym-ext} is then given by
\begin{equation}
\left\Vert \Pi_{RS\hat{R}\hat{S}}^{G}V_{S^{\prime}E\rightarrow R\hat{R}\hat
{S}E^{\prime}}|\psi\rangle_{S^{\prime}S}|0\rangle_{E}\right\Vert _{2}^{2}.
\end{equation}

\begin{figure}[ptb]
\begin{center}
\includegraphics[
width=\linewidth
]{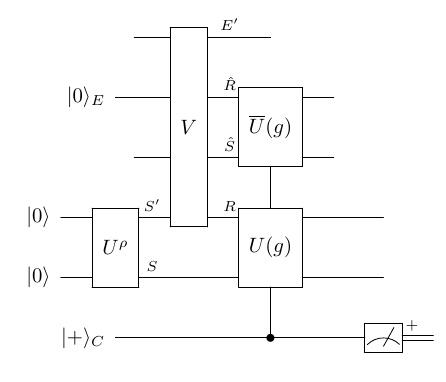}
\end{center}
\caption{Quantum circuit to implement Algorithm~\ref{alg:sym-ext}. The unitary $U^{\rho}$ prepares a purification $\psi_{S^{\prime}S}$ of the state $\rho_{S}$. Algorithm~\ref{alg:sym-ext} tests whether the state $\rho_{S}$ is
$G$-symmetric extendible, as defined in Definition~\ref{def:g-sym-ext}. Its
acceptance probability is equal to the maximum $G$-symmetric extendible
fidelity, as defined in~\eqref{eq:max-g-sym-ext-fid}.}
\label{fig:case-4}
\end{figure}

Our proof of the following theorem is similar to the proof given for
Theorem~\ref{thm:max-acc-prob-g-sym}. For completeness, we provide our proof in Appendix~\ref{app:proof-thm-g-se}.

\begin{theorem}
\label{thm:G-SE-acc-prob}The maximum acceptance probability of
Algorithm~\ref{alg:sym-ext} is equal to the maximum $G$-symmetric extendible
fidelity in \eqref{eq:max-g-sym-ext-fid}, i.e.,
\begin{multline}
\max_{V_{S^{\prime}E\rightarrow R\hat{R}\hat{S}E^{\prime}}}\left\Vert
\Pi_{RS\hat{R}\hat{S}}^{G}V_{S^{\prime}E\rightarrow R\hat{R}\hat{S}E^{\prime}
}|\psi\rangle_{S^{\prime}S}|0\rangle_{E}\right\Vert _{2}^{2}\\
=\max_{\sigma_{S}\in\operatorname*{SymExt}_{G}}F(\rho_{S},\sigma_{S}),
\end{multline}
where the set $\operatorname*{SymExt}_{G}$ is defined in \eqref{eq:g-sym-ext-set}.
\end{theorem}

\begin{example}
For the triangular dihedral group example (see Example~\ref{ex:dih-gr-d3}), we use the same unitary $U_d$ to prepare the superposition $\ket{+}_C$ and the same mapping of control states to group elements. The circuit to test for $G$-symmetric extendibility is shown in Figure~\ref{fig:Dihedral_GSE_Circuit}.
\end{example}

\begin{figure}
\begin{center}
\includegraphics[width=\linewidth]{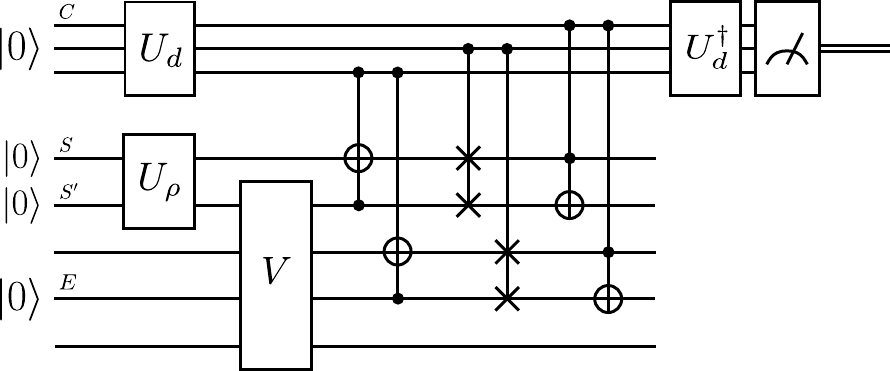}
\end{center}
\caption{Quantum circuit implementing Algorithm~\ref{alg:sym-ext} to test $G$-symmetric extendibility in the case that  the group $G$ is the triangular dihedral group. Compared to Figure~\ref{fig:case-4}, the systems $S$ and $S'$ are one qubit each, $C$ consists of three qubits, and $\ket{+}_C$ is defined as $U_d\ket{000}$. Both the SWAP and CNOT gates have no imaginary entries and thus are equal to their own complex conjugates.}
\label{fig:Dihedral_GSE_Circuit}
\end{figure}

\begin{remark}[Extensions to compact groups]
    Throughout our paper we have focused on discrete, finite groups; however, these notions of symmetry and the algorithms presented above in principle may be extended to continuous groups as well, permitting certain conditions hold. We leave a detailed investigation of this topic for future work and only discuss this extension briefly here. In particular, our algorithms can be generalized to any compact Lie group represented on a finite-dimensional quantum system. The primary limitation in cases of compact groups is realizing the following projection \cite{harrow2013church}
\begin{equation}
    \Pi^G \coloneqq  \int_{g \in G} d \mu(g) \ U(g)   \, ,
\end{equation}
where $U(g)$ is the unitary representation of $g$ and $\mu(g)$ is the Haar measure for the group. It follows from Caratheodory's theorem that there exists a probability mass function $\{p(g)\}_{g \in G'}$, where $G'$ is a finite set, such that the following equality holds:
\begin{equation}
    \Pi^{G} = \sum_{g \in G'} p(g) U(g).
\end{equation}
As such, since our algorithms ultimately realize this projection for the case in which $p(g)$ is uniform, they can be generalized in the following way. For concreteness, we consider the following generalization of Algorithm~\ref{alg:simple}, but we note that our other algorithms can be generalized similarly:
\begin{enumerate}
    \item Prepare an ancillary register $C$ in the state
    \begin{equation}
        |\varphi_p\rangle_C \coloneqq \sum_{g \in G'}\sqrt{p(g)} | g\rangle.
    \end{equation}
    \item Append the state $|\psi\rangle_S$ and perform the following controlled unitary:
    \begin{equation}
        \sum_{g \in G'} |g\rangle\!\langle g|_C \otimes U_S(g).
    \end{equation}
    \item Perform the measurement $\{|\varphi_p\rangle\!\langle \varphi_p |_C, \mathbb{I}_C - |\varphi_p\rangle\!\langle \varphi_p |_C\}$ on the register $C$, and accept if and only if the outcome $|\varphi_p\rangle\!\langle \varphi_p |_C$ occurs.
\end{enumerate}
Following similar calculations given in \eqref{eq:plus-over-group}--\eqref{eq:acc-prob-bose-test-pure-state-1}, we conclude that the acceptance probability of this algorithm is equal to $\operatorname{Tr}[\Pi^{G} |\psi\rangle\!\langle \psi|_S]$.

Although this abstract presentation of the generalized algorithm seems straightforward, there are some key questions to address before realizing it in practice. What is the probability mass function $\{p(g)\}_{g \in G'}$ that results from applying Caratheodory's theorem? This theorem only guarantees the existence of such a probability mass function, but it does not construct it. Once the probability mass function is known, is the state $|\varphi_p\rangle_C$ efficiently preparable? Addressing these two questions would lead to an efficient algorithm for estimating $\operatorname{Tr}[\Pi^{G} |\psi\rangle\!\langle \psi|_S]$.
%{\color{red}NEED TO REVISE} Thus, whenever there exists a finite group measure $\mu(g)$, the projector can be effectively realized by similar methods to those we have employed. Such measures exist for all finite groups and all compact Lie groups. In that case, finding an adequate measure is sufficient to implement the algorithm.
When the group representation permits a $t$-design \cite{roy2009unitary}, then it is straightforward to realize the algorithm, and we consider some examples in Sections~\ref{sec:collective-U} and \ref{sec:collective-z}. In general, addressing these questions may not be trivial; the topic of $t$-designs is addressed in a large body of work \cite{scott2008tomography,roy2009unitary,gross2007designs} beyond the scope considered here.
\end{remark}

\section{Tests of \texorpdfstring{$k$}{k}-extendibility of states and covariance symmetry of channels}

\label{sec:specialized-tests}

The theory developed in Section~\ref{sec:tests-o-sym}\ is rather general. In
the forthcoming subsections, we apply it to test for extendibility of bipartite and multipartite quantum states and to test for covariance symmetry of quantum channels and measurements. Later on in Section~\ref{sec:var-algs}, we consider many other example of groups and symmetry tests and simulate the performance of Algorithms~\ref{alg:simple}--\ref{alg:sym-ext}. 

\subsection{Separability test for pure bipartite states}

We illustrate the $G$-Bose symmetry test from Section~\ref{sec:simple-algorithm}\ on a case of interest:\ deciding whether a pure bipartite state is entangled. This problem is known to be BQP-complete \cite{GHMW15},\ and one can decide it by means of the SWAP test as considered in \cite{HM10}. The SWAP test as a quantum computational method of quantifying entanglement has been further studied in recent work \cite{FKS21,BGCC21}.

Let $\psi_{AB}$ be a pure bipartite state, and let $\psi_{AB}^{\otimes k}$
denote $k$ copies of it. Then we can consider the permutation unitaries
$W_{B_{1}\cdots B_{k}}(\pi)$ from Example~\ref{ex:k-ext}. This example is
a special case of $G$-Bose symmetry with the identifications
\begin{align}
S  &  \leftrightarrow A_{1}B_{1}\cdots A_{k}B_{k},\\
U_{S}(g)  &  \leftrightarrow \mathbb{I}_{A_{1}\cdots A_{k}}\otimes W_{B_{1}\cdots
B_{k}}(\pi).
\end{align}
The acceptance probability of Algorithm~\ref{alg:simple}\ is equal to
\begin{equation}
\operatorname{Tr}[\Pi_{B_{1}\cdots B_{k}}^{\text{Sym}}\rho_{B}^{\otimes k}],
\end{equation}
where the projection $\Pi_{B_{1}\cdots B_{k}}^{\text{Sym}}$ onto the symmetric
subspace is defined in \eqref{eq:sym-subspace-proj} and $\rho_{B}
\coloneqq \operatorname{Tr}_{A}[\psi_{AB}]$. We note that there is an efficient quantum algorithm to implement this test \cite[Section~4]{BBD+97}, which amounts to an instance of the abstract formulation in Algorithm~\ref{alg:simple}. For $k=2$, this reduces to the
well-known SWAP\ test with acceptance probability
\begin{equation}
p_{\text{acc}}^{(2)}\coloneqq \frac{1}{2}\left(  1+\operatorname{Tr}[\rho
_{B}^{2}]\right)  .
\end{equation}
For $k=3$, the acceptance probability is
\begin{equation}
p_{\text{acc}}^{(3)}\coloneqq \frac{1}{6}\left(  1+3\operatorname{Tr}[\rho
_{B}^{2}]+2\operatorname{Tr}[\rho_{B}^{3}]\right)  .
\end{equation}
For $k=4$, the acceptance probability is
\begin{multline}
p_{\text{acc}}^{(4)}\coloneqq
\frac{1}{24}\Big(  1+6\operatorname{Tr}[\rho_{B}^{2}]+3\left(
\operatorname{Tr}[\rho_{B}^{2}]\right)^{2} \\
+
8\operatorname{Tr}[\rho_{B}
^{3}]+6\operatorname{Tr}[\rho_{B}^{4}]\Big)  .
\end{multline}
We conclude that
\begin{equation}
p_{\text{acc}}^{(2)}\geq p_{\text{acc}}^{(3)}\geq p_{\text{acc}}^{(4)},
\label{eq:acc-prob-decreasing}
\end{equation}
because $\operatorname{Tr}[\rho^{k}]=\sum_{j}\lambda_{j}^{k}$, where the
eigenvalues of $\rho$ are $\{\lambda_{j}\}_{j}$, and for all $x,y\in\left[
0,1\right]  $,
\begin{align}
& \frac{1}{2}\left(  x+x^{2}\right) \notag \\
&  \geq\frac{1}{6}\left(  x+3x^{2}
+2x^{3}\right) \\
&  \geq\frac{1}{24}\left(  x+6x^{2}+3x^{2}y+8x^{3}+6x^{4}\right)  .
\end{align}
The inequalities in \eqref{eq:acc-prob-decreasing}\ imply that the tests
become more difficult to pass as $k$ increases. In a previous version of our paper \cite{LW21}, we speculated that this trend of decreasing acceptance probability continues as $k$ increases. Indeed, this was subsequently shown to be true in \cite{bradshaw2022cycle}.

We can interpret these findings in two different ways. For each $k$, the rejection probability $1-p_{\text{acc}}^{(k)}$ can be understood as an entanglement measure for pure states, similar to how the linear entropy $1-\operatorname{Tr}[\rho_{B}^{2}]$ 
is interpreted as an entanglement measure. Indeed, these quantities are non-increasing under local operations and classical communication that take pure states to pure states, as every R\'{e}nyi entropy (defined as $\frac{1}{1-\alpha}\log\operatorname{Tr}[\rho_{B}^{\alpha}]$ for $\alpha \in(0,1)\cup(1,\infty)$)
is an entanglement measure for pure states \cite{HHHH09}. Another interpretation is that, if using these tests to decide if a given pure state is product or entangled, a decision can be determined with fewer repetitions of the basic test by using tests with higher values of $k$.

\subsection{Separability test for pure multipartite states}

We can generalize the test from the previous section to one for pure
multipartite entanglement. Let $\psi_{A_{1}\cdots A_{m}}$ be a multipartite pure state, and let $\psi_{A_{1}\cdots A_{m}}^{\otimes k}$ denote $k$ copies of it. For $i\in\left\{  1,\ldots,m\right\}  $ and $\pi_{i}\in S_{k}$, let $W_{A_{i,1}\cdots A_{i,k}}(\pi_{i})$ denote a permutation unitary, where $i$ is an index for the $i$th party, and the notation $A_{i,j}$ for $j\in\left\{ 1,\ldots,k\right\}  $ indicates the $j$th system of the $i$th party. This example is a special case of $G$-Bose symmetry with the identifications:
\begin{align}
S &  \leftrightarrow A_{1,1}\cdots A_{1,k}\cdots A_{m,1}\cdots A_{m,k}
,\label{eq:pure-multipartite-idents-1}\\
U_{S}(g) &  \leftrightarrow\bigotimes\limits_{i=1}^{m}W_{A_{i,1}\cdots
A_{i,k}}(\pi_{i}), \\
G &  \leftrightarrow\overset{m\text{ times}}{\overbrace{S_{k}\times
\cdots\times S_{k}}},\\
g &  \leftrightarrow(\pi_{1},\ldots,\pi_{m}
),\label{eq:pure-multipartite-idents-4}
\end{align}
where $\times$ denotes the direct product of groups. The $G$-Bose symmetry test from Section~\ref{sec:simple-algorithm}\ has the following acceptance probability in this case:
\begin{equation}
\operatorname{Tr}\!\left[  \bigotimes\limits_{i=1}^{m}\Pi_{A_{i,1}\cdots
A_{i,k}}^{\operatorname{Sym}}\psi_{A_{1}\cdots A_{m}}^{\otimes k}\right]  .
\end{equation}
Note that one can again use the circuit from \cite[Section~4]{BBD+97} to implement this test.
For $k=2$, this test is known to be a test of multipartite pure-state entanglement \cite{HM10}, which has been considered in more recent works \cite{FKS21,BGCC21}. As far as we aware, the test proposed above, for larger values of $k$, has not been considered previously. Presumably, as was the case for the bipartite entanglement test mentioned above, the multipartite test is such that it becomes easier to detect an entangled state as $k$ increases. We leave its detailed analysis for future work.

\subsection{\texorpdfstring{$k$}{k}-Bose extendibility test for bipartite states}

\label{sec:k-bose-ext-test}We now demonstrate how the test for $G$-Bose
symmetric extendibility from Section~\ref{sec:G-Bose-sym-ext-test}\ can
realize a test for $k$-Bose extendibility of a bipartite state. Since every
separable state is $k$-Bose extendible, this test is then indirectly a test
for separability. To see this in detail, recall that a bipartite state
$\sigma_{AB}$ is separable if it can be written as a convex combination of
pure product states \cite{HHHH09,KW20book}:
\begin{equation}
\sigma_{AB}=\sum_{x}p_{X}(x)\psi_{A}^{x}\otimes\phi_{B}^{x},
\end{equation}
where $p_{X}$ is a probability distribution and $\{\psi_{A}^{x}\}_{x}$ and
$\{\phi_{B}^{x}\}_{x}$ are sets of pure states. A $k$-Bose extension for this
state is as follows:
\begin{equation}
\omega_{AB_{1}\cdots B_{k}}=\sum_{x}p_{X}(x)\psi_{A}^{x}\otimes\phi_{B_{1}
}^{x}\otimes\cdots\otimes\phi_{B_{k}}^{x}.
\end{equation}
By making the identifications discussed in Example~\ref{ex:k-bose-ext}, it
follows from Theorem~\ref{thm:G-BSE-acc-prob}\ that the test from
Section~\ref{sec:G-Bose-sym-ext-test} is a test for $k$-Bose extendibility.
For an input state $\rho_{AB}$, the acceptance probability of
Algorithm~\ref{alg:G-BSE-test}\ is equal to the maximum $k$-Bose extendible
fidelity
\begin{equation}
\max_{\omega_{AB}\in k\text{-BE}}F(\rho_{AB},\omega_{AB}),
\end{equation}
where $k$-BE denotes the set of $k$-Bose extendible states, as defined in
Example~\ref{ex:k-bose-ext}.

This test for $k$-Bose extendibility was proposed in
\cite{HMW13,Hayden:2014:TQI}\ for understanding the computational complexity
of the circuit separability problem. In that work, it was not mentioned that
the test employed is a test for $k$-Bose extendibility; instead, it was
suggested to be a test for $k$-extendibility. Thus, our observation here (also
made earlier by \cite{Marvian13}) is that the test proposed in
\cite{HMW13,Hayden:2014:TQI} is actually a test for $k$-Bose extendibility,
and we consider in the next section a true test for $k$-extendibility. The
main results of \cite{HMW13,Hayden:2014:TQI} were the computational
complexity of the circuit version of the separability problem, and so the precise
kind of test used was not particularly important there.

\subsection{\texorpdfstring{$k$}{k}-Extendibility test for bipartite states}

In this section, we discuss how the test for $G$-symmetric extendibility from
Section~\ref{sec:test-g-sym-ext}\ can realize a test for $k$-extendibility of
a bipartite state. Due to the known connections between $k$-extendibility
and separability \cite{CKMR08,BCY11,BCY11a,BH12}, this test is an indirect test for separability of a bipartite state. Since every separable state is $k$-Bose extendible, as discussed in Section~\ref{sec:k-bose-ext-test}, and every $k$-Bose extendible state is $k$-extendible, it follows that every separable state is $k$-extendible.

By making the identifications discussed in Example~\ref{ex:k-ext}, it follows
from Theorem~\ref{thm:G-SE-acc-prob}\ that the test from
Section~\ref{sec:test-g-sym-ext} is a test for $k$-extendibility. For an
input state $\rho_{AB}$, the acceptance probability of
Algorithm~\ref{alg:sym-ext}\ is equal to the maximum $k$-extendible fidelity
\begin{equation}
\max_{\omega_{AB}\in k\text{-E}}F(\rho_{AB},\omega_{AB}),
\end{equation}
where $k$-E denotes the set of $k$-extendible states, as defined in
Example~\ref{ex:k-ext}.

As far as we are aware, this quantum computational test for $k$-extendibility
is original to this paper, however inspired by the approach from
\cite{HMW13,Hayden:2014:TQI}. It was argued in \cite{HMW13,Hayden:2014:TQI}
that the acceptance probability of the test there is bounded from above by the
maximum $k$-extendible fidelity, which is consistent with the fact that the
set of $k$-Bose extendible states is contained in the set of $k$-extendible
states and our observation here that the acceptance probability of the test
in \cite{HMW13,Hayden:2014:TQI} is equal to the maximum $k$-Bose extendible fidelity.

\subsection{Extendibility tests for multipartite states}

We discuss briefly how the tests from
Sections~\ref{sec:G-Bose-sym-ext-test}\ and \ref{sec:test-g-sym-ext} apply to
the multipartite case, using identifications similar to those in \eqref{eq:pure-multipartite-idents-1}--\eqref{eq:pure-multipartite-idents-4}.

First, let us recall the definition of multipartite extendibility \cite{DPS05}. Let
$\sigma_{A_{1}\cdots A_{m}}$ be a multipartite state. Such a state is
$(k_{1},\ldots,k_{m})$-extendible if there exists a state $\omega
_{A_{1,1}\cdots A_{1,k_{1}}\cdots A_{m,1}\cdots A_{m,k_{m}}}$ such that
\begin{multline}
\sigma_{A_{1}\cdots A_{m}}=\label{eq:multipartite-ext-cond}\\
\operatorname{Tr}_{A_{1,2}\cdots A_{1,k_{1}}\cdots A_{m,2}\cdots A_{m,k_{m}}
}[\omega_{A_{1,1}\cdots A_{1,k_{1}}\cdots A_{m,1}\cdots A_{m,k_{m}}}]
\end{multline}
and
\begin{multline}
\omega_{A_{1,1}\cdots A_{1,k_{1}}\cdots A_{m,1}\cdots A_{m,k_{m}}} = \\
W_{_{A_{1,1}\cdots A_{1,k_{1}}\cdots A_{m,1}\cdots A_{m,k_{m}}}}
^{\mathbf{\pi}}\omega_{A_{1,1}\cdots A_{1,k_{1}}\cdots A_{m,1}\cdots
A_{m,k_{m}}}\times\\
(W_{_{A_{1,1}\cdots A_{1,k_{1}}\cdots A_{m,1}\cdots A_{m,k_{m}}}}
^{\mathbf{\pi}})^{\dag},
\end{multline}
for all $\mathbf{\pi}$, where $\mathbf{\pi}=(\pi_{1},\ldots,\pi_{m})\in
S_{k_{1}}\times\cdots\times S_{k_{m}}$ and
\begin{equation}
W_{_{A_{1,1}\cdots A_{1,k_{1}}\cdots A_{m,1}\cdots A_{m,k_{m}}}}^{\mathbf{\pi
}}\coloneqq \bigotimes\limits_{i=1}^{m}W_{A_{i,1}\cdots A_{i,k_{i}}}^{\pi_{i}}.
\end{equation}
A multipartite state is $(k_{1},\ldots,k_{m})$-Bose extendible if there exists
a state $\omega_{A_{1,1}\cdots A_{1,k_{1}}\cdots A_{m,1}\cdots A_{m,k_{m}}}$
such that \eqref{eq:multipartite-ext-cond} holds and
\begin{multline}
\omega_{A_{1,1}\cdots A_{1,k_{1}}\cdots A_{m,1}\cdots A_{m,k_{m}}}=\\
\Pi_{A_{1,1}\cdots A_{1,k_{1}}\cdots A_{m,1}\cdots A_{m,k_{m}}}\omega
_{A_{1,1}\cdots A_{1,k_{1}}\cdots A_{m,1}\cdots A_{m,k_{m}}}\\
\times\Pi_{A_{1,1}\cdots A_{1,k_{1}}\cdots A_{m,1}\cdots A_{m,k_{m}}},
\end{multline}
where
\begin{align}
\Pi_{A_{1,1}\cdots A_{1,k_{1}}\cdots A_{m,1}\cdots A_{m,k_{m}}}  &
\coloneqq \bigotimes\limits_{i=1}^{m}\Pi_{A_{i,1}\cdots A_{i,k_{i}}}
^{\operatorname{Sym}},\\
\Pi_{A_{i,1}\cdots A_{i,k_{i}}}^{\operatorname{Sym}}  & \coloneqq \frac{1}{k_{i}!}
\sum_{\pi_{i}\in S_{k_{i}}}W_{A_{i,1}\cdots A_{i,k_{i}}}^{\pi_{i}}.
\end{align}

By making the identifications
\begin{align}
S &  \leftrightarrow A_{1,1}\cdots A_{m,1},\\
R &  \leftrightarrow A_{1,2}\cdots A_{1,k_{1}}\cdots A_{m,2}\cdots A_{m,k_{m}
} , \\
U_{RS}(g) &  \leftrightarrow\bigotimes\limits_{i=1}^{m}W_{A_{i,1}\cdots
A_{i,k_{i}}}(\pi_{i}) , \\
G &  \leftrightarrow S_{k_{1}}\times\cdots\times S_{k_{m}},\\
g &  \leftrightarrow(\pi_{1},\ldots,\pi_{m}),
\end{align}
it follows that Algorithm~\ref{alg:G-BSE-test} is a test for multipartite
$(k_{1},\ldots,k_{m})$-Bose extendibility of a state $\rho_{A_{1}\cdots A_{m}}$, with acceptance probability equal to
\begin{equation}
\max_{\omega_{A_{1}\cdots A_{m}}\in(k_{1},\ldots,k_{m})\text{-BE}}
F(\rho_{A_{1}\cdots A_{m}},\omega_{A_{1}\cdots A_{m}}),
\end{equation}
and Algorithm~\ref{alg:sym-ext} is a test for multipartite $(k_{1}
,\ldots,k_{m})$-extendibility of a state $\rho_{A_{1}\cdots A_{m}}$, with
acceptance probability equal to
\begin{equation}
\max_{\omega_{A_{1}\cdots A_{m}}\in(k_{1},\ldots,k_{m})\text{-E}}F(\rho
_{A_{1}\cdots A_{m}},\omega_{A_{1}\cdots A_{m}}),
\end{equation}
where $(k_{1},\ldots,k_{m})$-BE and $(k_{1},\ldots,k_{m})$-E denote the sets
of $(k_{1},\ldots,k_{m})$-Bose extendible and $(k_{1},\ldots,k_{m}
)$-extendible states, respectively.

\subsection{Testing covariance symmetry of a quantum channel}

\label{sec:cov-sym-ch-test}

We can also use the test from Algorithm~\ref{alg:g-sym-test} to test for
covariance symmetry of a quantum channel. Before stating it, let us recall the
notion of a covariant channel \cite{Hol02}. Let $G$ be a group, and let $\{U_{A}(g)\}_{g\in
G}$ and $\{V_{B}(g)\}_{g\in G}$ denote projective unitary representations of~$G$. A channel $\mathcal{N}_{A\rightarrow B}$ is covariant if the following
$G$-covariance symmetry condition holds
\begin{equation}
\mathcal{N}_{A\rightarrow B}\circ\mathcal{U}_{A}(g)=\mathcal{V}_{B}
(g)\circ\mathcal{N}_{A\rightarrow B}\qquad\forall g\in G,\label{eq:ch-cov-def}
\end{equation}
where the unitary channels $\mathcal{U}_{A}(g)$ and $\mathcal{V}_{B}(g)$ are
respectively defined from $U_{A}(g)$ and $V_{B}(g)$ as
\begin{align}\label{eq:uandv}
\mathcal{U}_{A}(g)(\omega_{A}) &  \coloneqq U_{A}(g)\omega_{A}U_{A}(g)^{\dag
},\\
\mathcal{V}_{B}(g)(\tau_{B}) &  \coloneqq V_{B}(g)\tau_{B}V_{B}(g)^{\dag}.
\end{align}

It is well known that a channel is covariant in the sense above if and only if
its Choi state is invariant in the following sense \cite[Eq.~(59)]{CDP09}:
\begin{equation} 
\Phi_{RB}^{\mathcal{N}}=(\overline{\mathcal{U}}_{R}(g)\otimes\mathcal{V}
_{B}(g))(\Phi_{RB}^{\mathcal{N}})\quad\forall g\in G,
\label{eq:choi-cond-cov-ch}
\end{equation}
where
\begin{equation}
\overline{\mathcal{U}}_{R}(g)(\omega_{R})\coloneqq\overline{U}_{R}
(g)\omega_{R}U_{R}(g)^{T},
\end{equation}
and the superscript $T$ indicates the transpose. Also, the Choi state $\Phi_{RB}^{\mathcal{N}}$ is
defined as
\begin{align}
\Phi_{RB}^{\mathcal{N}}  & \coloneqq \mathcal{N}_{A\rightarrow B}(\Phi_{RA}),\\
\Phi_{RA}  & \coloneqq \frac{1}{\left\vert A\right\vert }\sum_{i,j}|i\rangle\!\langle
j|_{R}\otimes|i\rangle\!\langle j|_{A}.
\end{align}

Suppose then that a circuit is available that generates the channel
$\mathcal{N}_{A\rightarrow B}$. Similar to the first assumption in
Section~\ref{sec:tests-o-sym}, we suppose that the circuit realizes a unitary
channel $\mathcal{W}_{AE^{\prime}\rightarrow BE}$\ that extends the original
channel, in the sense that
\begin{equation}
\mathcal{N}_{A\rightarrow B}(\omega_{A})=(\operatorname{Tr}_{E}\circ
\mathcal{W}_{AE^{\prime}\rightarrow BE})(\omega_{A}\otimes|0\rangle\!\langle0|_{E^{\prime}}).
\end{equation}
Then to decide whether the channel is covariant, we send in one share of a
maximally entangled state to the unitary extension channel, such that the
overall state is
\begin{equation}
\mathcal{W}_{AE^{\prime}\rightarrow BE}(\Phi_{RA}\otimes|0\rangle\!\langle0|_{E^{\prime}}).
\end{equation}
Now making the identifications
\begin{align}
E  & \leftrightarrow S^{\prime},\\
RB  & \leftrightarrow S,\\
\overline{U}_{R}(g)\otimes V_{B}(g)  & \leftrightarrow U_{S}(g),
\end{align}
we apply Algorithm~\ref{alg:g-sym-test}, and as a consequence of
Theorem~\ref{thm:max-acc-prob-g-sym}, the acceptance probability is equal to
\begin{equation}
\max_{\sigma_{RB}\in\operatorname{Sym}_G}F(\Phi_{RB}^{\mathcal{N}}
,\sigma_{RB}),
\end{equation}
where
\begin{multline}
\operatorname{Sym}_G\coloneqq \\
\left\{
\begin{array}
[c]{c}
\sigma_{RB}\in\mathcal{D}(\mathcal{H}_{RB}):\\
\sigma_{RB}=(\overline{\mathcal{U}}_{R}(g)\otimes\mathcal{V}_{B}
(g))(\sigma_{RB})\ \forall g\in G
\end{array}
\right\}  .
\end{multline}
Thus, the test accepts with probability equal to one if and only if the channel is covariant in the sense of~\eqref{eq:ch-cov-def}.

We note here that a special kind of channel is a unitary channel induced by Hamiltonian evolution (i.e., $\mathcal{N}(\cdot) = e^{-iHt} (\cdot) e^{iHt}$, where $H$ is the Hamiltonian and $t$ is the evolution time). This special case was considered in \cite{LW22}, in which channel symmetry tests were employed as Hamiltonian symmetry tests.

\subsection{Testing covariance symmetry of a quantum measurement}

Recall that a quantum measurement is described by a positive operator-valued
measure (POVM), which is a set $\Lambda\coloneqq \{\Lambda^{x}\}_{x\in\mathcal{X}}$ of
positive semi-definite operators such that $\sum_{x\in\mathcal{X}}\Lambda
^{x}=\mathbb{I}$. From this set, we can define a quantum measurement channel as
follows:%
\begin{equation}
\mathcal{M}_{S\rightarrow X}(\rho_{S})\coloneqq \sum_{x\in\mathcal{X}}%
\operatorname{Tr}[\Lambda_{S}^{x}\rho_{S}]|x\rangle\!\langle x|_{X}%
,\label{eq:povm-channel}%
\end{equation}
where $\{|x\rangle_{X}\}_{x\in\mathcal{X}}$ is an orthonormal basis.

A POVM\ is $G$-symmetric (also called group covariant) if there exists a
projective unitary representation $\left\{  U(g)\right\}  _{g\in G}$ of a
group $G$ such that  
\begin{equation}
U(g)^{\dag}\Lambda^{x}U(g)\in\Lambda \quad \forall g\in G,\, 
x\in\mathcal{X} .
\label{eq:def-G-sym-povm}
\end{equation}
$G$-symmetric POVMs have been studied extensively in the literature \cite{D78,H11book,CdVT03,DJR05}, and they  arise in many applications, having to do with state
discrimination \cite{KGDdS15} and estimation \cite{CD04}. It is thus of interest to
determine whether a given POVM\ is $G$-symmetric.

Connecting to the previous section, a measurement channel $\mathcal{M}_{S\rightarrow X}$ is
$G$-symmetric if there exist projective unitary representations $\left\{
U(g)\right\}  _{g\in G}$ and $\left\{  W(g)\right\}  _{g\in G}$ such that%
\begin{equation}
\mathcal{M}_{S\rightarrow X}\circ\mathcal{U}(g)=\mathcal{W}(g)\circ
\mathcal{M}_{S\rightarrow X}\quad\forall g\in G.\label{eq:cov-sym-meas-ch}%
\end{equation}
Plugging into \eqref{eq:povm-channel}, the condition in
\eqref{eq:cov-sym-meas-ch} becomes%
\begin{multline}
\sum_{x\in\mathcal{X}}\operatorname{Tr}[U(g)^{\dag}\Lambda^{x}U(g)\rho
_{S}]|x\rangle\!\langle x|_{X} \\
=\sum_{x\in\mathcal{X}}\operatorname{Tr}%
[\Lambda_{S}^{x}\rho_{S}]W(g)|x\rangle\!\langle x|_{X}W(g)^{\dag} \quad\forall g\in G .
\end{multline}
Since the output system $X$ is classical, it is sensible to restrict the
unitary $W(g)$ to be a shift operator that realizes a permutation $\pi_{g}$ of
the classical letter $x$, so that we can write%
\begin{align}
& \sum_{x\in\mathcal{X}}\operatorname{Tr}[U(g)^{\dag}\Lambda^{x}U(g)\rho
_{S}]|x\rangle\!\langle x|_{X}\nonumber\\
& =\sum_{x\in\mathcal{X}}\operatorname{Tr}[\Lambda_{S}^{x}\rho_{S}]|\pi
_{g}(x)\rangle\!\langle\pi_{g}(x)|_{X}\\
& =\sum_{x\in\mathcal{X}}\operatorname{Tr}[\Lambda_{S}^{\pi_{g}^{-1}(x)}%
\rho_{S}]|x\rangle\!\langle x|_{X}.
\end{align}
Since this equation holds for every input state $\rho$, we conclude that the
following condition holds for a $G$-symmetric measurement channel:%
\begin{equation}
U(g)^{\dag}\Lambda^{x}U(g)=\Lambda_{S}^{\pi_{g}^{-1}(x)}\quad\forall g\in
G,\ x\in\mathcal{X},
\label{eq:cov-povm-conseq}
\end{equation}
coinciding with the definition given in \eqref{eq:def-G-sym-povm}.

As a consequence of the connection between \eqref{eq:cov-sym-meas-ch} and the definition in \eqref{eq:def-G-sym-povm}, we can use the methods from the previous section to test whether a
POVM\ is $G$-symmetric. Recall that the Choi state of a measurement channel
has the following form (see, e.g., \cite[Eq. (3.2.162)]{KW20book}):%
\begin{equation}
\Phi_{RX}^{\mathcal{M}}\coloneqq \mathcal{M}_{S\rightarrow X}(\Phi_{RS})=\frac
{1}{\left\vert \mathcal{X}\right\vert }\sum_{x\in\mathcal{X}}\left(
\Lambda_{R}^{x}\right)  ^{T}\otimes|x\rangle\!\langle x|_{X}%
.\label{eq:choi-state-meas-ch}%
\end{equation}
By appealing to \eqref{eq:choi-cond-cov-ch}, \eqref{eq:cov-sym-meas-ch}, and \eqref{eq:cov-povm-conseq}, it follows that a POVM is $G$-symmetric if
and only if its Choi state is $G$-symmetric in the following sense:%
\begin{equation}
(\overline{\mathcal{U}}_{R}(g)\otimes\mathcal{W}_{X}(g))(\Phi_{RX}%
^{\mathcal{M}})=\Phi_{RX}^{\mathcal{M}}\quad\forall g\in G,
\end{equation}
or equivalently, if%
\begin{multline}
\frac{1}{\left\vert \mathcal{X}\right\vert }\sum_{x\in\mathcal{X}}%
\left[  \mathcal{U}_{R}(g)(\Lambda_{R}^{x})  \right]^{T}\otimes
|\pi_{g}(x)\rangle\!\langle\pi_{g}(x)|_{X}\\
=\frac{1}{\left\vert \mathcal{X}\right\vert }\sum_{x\in\mathcal{X}}\left(
\Lambda_{R}^{x}\right)  ^{T}\otimes|x\rangle\!\langle x|_{X}\qquad\forall g\in
G.
\end{multline}

One method for performing a measurement  on a quantum
system $S$ is to employ a unitary circuit $U_{SX}$ acting on the system $S$
and a probe system $X$ prepared in the state $|0\rangle\!\langle0|_{X}$ (see, e.g.,  \cite[Figure~3.1]{KW20book}). This is
then followed by a projective measurement $\{|x\rangle\!\langle x|_{X}%
\}_{x\in\mathcal{X}}$ in the standard basis of the probe system~$X$. To
realize this process in a fully unitary manner, we can attach two probe
systems $X$ and $X^{\prime}$ to the system~$S$, prepared in the state
$|0\rangle\!\langle0|_{X}\otimes|0\rangle\!\langle0|_{X^{\prime}}$, perform the
unitary $U_{SX}$, followed by generalized controlled-NOT gates from $X$ to
$X^{\prime}$. If we send in one share $S$\ of a maximally entangled state
$\Phi_{RS}$, the resulting state is%
\begin{equation}
C_{XX^{\prime}}U_{SX}\left(  \Phi_{RS}\otimes|0\rangle\!\langle0|_{X}%
\otimes|0\rangle\!\langle0|_{X^{\prime}}\right)  U_{SX}^{\dag}C_{XX^{\prime}%
}^{\dag},
\end{equation}
where $C_{XX^{\prime}}$ denotes the generalized CNOT gate, defined through
$C_{XX^{\prime}}|x\rangle_{X}|0\rangle_{X^{\prime}}=|x\rangle_{X}%
|x\rangle_{X^{\prime}}$. Tracing over systems $S$ and $X^{\prime}$, the
resulting state is the Choi state of the measurement channel, as given in
\eqref{eq:choi-state-meas-ch}. Thus, by making the identifications%
\begin{align}
SX^{\prime}  & \leftrightarrow S^{\prime},\\
RX  & \leftrightarrow S,\\
\overline{\mathcal{U}}_{R}(g)\otimes\mathcal{W}_{X}(g)  & \leftrightarrow
U_{S}(g),
\end{align}
we apply Algorithm~\ref{alg:g-sym-test}, and as a consequence of Theorem~\ref{thm:max-acc-prob-g-sym}, the
acceptance probability is equal to%
\begin{equation}
\max_{\sigma_{RX}\in\text{Sym}_{G}}F(\Phi_{RX}^{\mathcal{M}},\sigma_{RX}),
\end{equation}
where
\begin{multline}
\text{Sym}_{G}\coloneqq \\
\left\{
\begin{array}
[c]{c}%
\sigma_{RX}\in\mathcal{D}(\mathcal{H}_{RX}):\\
\sigma_{RX}=(\overline{\mathcal{U}}_{R}(g)\otimes\mathcal{W}_{X}%
(g))(\sigma_{RX})\ \forall g\in G
\end{array}
\right\}  .
\end{multline}
Thus, the test accepts with probability equal to one if and only if the
POVM is $G$-symmetric, as defined in \eqref{eq:def-G-sym-povm}. Finally, we remark that it suffices to restrict the optimization over $\sigma_{RX}$ to be over quantum-classical states of the form $\sigma_{RX} = \sum_{x\in\mathcal{X}} \tilde{\sigma}_R^x \otimes |x\rangle \! \langle x|_X$, where each $\tilde{\sigma}_R^x$ is positive semi-definite and $\sum_{x \in \mathcal{X}} \operatorname{Tr}[\tilde{\sigma}_R^x]=1$. This follows because the Choi state $\Phi_{RX}^{\mathcal{M}}$ is quantum-classical (and thus invariant under such a dephasing), and the fidelity does not decrease under the action of a completely-dephasing channel on the classical system~$X$. It thus suffices to optimize over quantum-classical $\sigma_{RX}$ satisfying
\begin{multline}
\sum_{x\in\mathcal{X}} \tilde{\sigma}_R^x \otimes |x\rangle \! \langle x|_X = \\
\sum_{x\in\mathcal{X}} \overline{\mathcal{U}}_R(g)(\tilde{\sigma}_R^x) \otimes |\pi_g(x)\rangle \! \langle \pi_g(x)|_X \, ,
\end{multline}
for all $g \in G$, or equivalently, $\tilde{\sigma}_R^{\pi_g(x)} = \overline{\mathcal{U}}_R(g)(\tilde{\sigma}_R^x)$ for all $x \in \mathcal{X}$ and $g\in G$.

\section{Semi-definite programs for maximum symmetric fidelities}

\label{sec:SDPs-sym-fids}

In this section, we note that the acceptance
probabilities of Algorithms~\ref{alg:simple}--\ref{alg:sym-ext}\ can be
computed by means of semi-definite programming (see \cite{BV04,Wat18,KW20book}
for reviews). This is useful for comparing the true values of the acceptance probabilities of Algorithms~\ref{alg:simple}--\ref{alg:sym-ext} to estimates formed from executing them
on near-term quantum computers; however, this semi-definite programming approach
only works well in practice if the circuit $U^{\rho}$ acts on a small number of qubits. This limitation holds because the semi-definite programs (SDPs) run in a time polynomial in the dimension of the states involved, but the dimension of a state grows exponentially with the number of qubits involved.

We note that the fact that the acceptance probabilities of
Algorithms~\ref{alg:simple}--\ref{alg:sym-ext} can be computed by semi-definite programming follows from a more general fact that the acceptance probability of a QIP(2)\ algorithm can be computed in this manner \cite{W09,VW15}; however, it is helpful to have the explicit form
of the SDPs available.

We now list the SDPs for the acceptance probabilities of
Algorithms~\ref{alg:simple}--\ref{alg:sym-ext}. To begin with, let us note that the acceptance probability of Algorithm~\ref{alg:simple} is equal to
$\operatorname{Tr}[\Pi_{S}^{G}\rho_{S}]$, and so there is no need for an optimization. This quantity can be calculated directly if the projection matrix $\Pi_{S}^{G}$ and the density matrix $\rho_{S}$ are available. Alternatively, one could employ an optimization as given below.
Let us first note that the root fidelity of states $\omega$ and $\tau$ can be
calculated by the following SDP \cite{Wat13}:
\begin{equation}\label{eq:SDP-fidelity}
\sqrt{F}(\omega,\tau)=\max_{X\in\mathcal{L}(\mathcal{H})}\left\{
\operatorname{Tr}[\operatorname{Re}[X]]:
\begin{bmatrix}
\omega & X^{\dag}\\
X & \tau
\end{bmatrix}
\geq0\right\}, 
\end{equation}
where $\mathcal{L}(\mathcal{H})$ is the space of linear operators acting on the Hilbert space $\mathcal{H}$. Each of the sets $\text{B-Sym}_{G}$,  $\operatorname*{Sym}_{G}$, $\operatorname*{BSE}_{G}$, and $\operatorname*{SymExt}_{G}$ are specified by semi-definite constraints. Thus, combining the optimization in \eqref{eq:SDP-fidelity}\ with
various constraints, we find that the acceptance probabilities of Algorithms~\ref{alg:simple}--\ref{alg:sym-ext} can be
calculated by using the following SDPs, respectively:
\begin{multline}
\label{eq:SDP-rootfid-GBS}
\max_{\sigma_{S}\in\text{B-Sym}_{G}}\sqrt{F}(\rho_{S},\sigma_{S})\\
=\max_{\substack{X\in\mathcal{L}(\mathcal{H}_{S}),\\\sigma_{S}\geq0}}\left\{
\begin{array}
[c]{c}
\operatorname{Tr}[\operatorname{Re}[X]]:\\
\begin{bmatrix}
\rho_{S} & X^{\dag}\\
X & \sigma_{S}
\end{bmatrix}
\geq0,\\
\operatorname{Tr}[\sigma_{S}]=1,\\
\sigma_{S}=\Pi^G_{S}\sigma_{S}\Pi^G_{S}
\end{array}
\right\}  ,
\end{multline}
\begin{multline}
\label{eq:SDP-rootfid-GS}
\max_{\sigma_{S}\in\operatorname*{Sym}_{G}}\sqrt{F}(\rho_{S},\sigma_{S})\\
=\max_{\substack{X\in\mathcal{L}(\mathcal{H}_{S}),\\\sigma_{S}\geq0}}\left\{
\begin{array}
[c]{c}
\operatorname{Tr}[\operatorname{Re}[X]]:\\
\begin{bmatrix}
\rho_{S} & X^{\dag}\\
X & \sigma_{S}
\end{bmatrix}
\geq0,\\
\operatorname{Tr}[\sigma_{S}]=1,\\
\sigma_{S}=U_{S}(g)\sigma_{S}U_{S}(g)^{\dag}\ \forall g\in G
\end{array}
\right\}  ,
\end{multline}
\begin{multline}
\label{eq:SDP-rootfid-GBSE}
\max_{\sigma_{S}\in\operatorname*{BSE}_{G}}\sqrt{F}(\rho_{S},\sigma_{S})\\
=\max_{\substack{X\in\mathcal{L}(\mathcal{H}_{S}),\\\omega_{RS}\geq0}}\left\{
\begin{array}
[c]{c}
\operatorname{Tr}[\operatorname{Re}[X]]:\\
\begin{bmatrix}
\rho_{S} & X^{\dag}\\
X & \operatorname{Tr}_{R}[\omega_{RS}]
\end{bmatrix}
\geq0,\\
\operatorname{Tr}[\omega_{RS}]=1,\\
\omega_{RS}=\Pi_{RS}^{G}\omega_{RS}\Pi_{RS}^{G}
\end{array}
\right\}  ,
\end{multline}
\begin{multline}
\label{eq:SDP-rootfid-GSE}
\max_{\sigma_{S}\in\operatorname*{SymExt}_{G}}\sqrt{F}(\rho_{S},\sigma_{S}) = \\
\max_{\substack{X\in\mathcal{L}(\mathcal{H}_{S}),\\\omega_{RS}\geq0}}\left\{
\begin{array}
[c]{c}
\operatorname{Tr}[\operatorname{Re}[X]]:\\
\begin{bmatrix}
\rho_{S} & X^{\dag}\\
X & \operatorname{Tr}_{R}[\omega_{RS}]
\end{bmatrix}
\geq0,\\
\operatorname{Tr}[\omega_{RS}]=1,\\
\omega_{RS}=U_{RS}(g)\omega_{RS}U_{RS}(g)^{\dag}\ \forall g\in G
\end{array}
\right\}  .
\end{multline}

We note here that the complexity of the SDPs in \eqref{eq:SDP-rootfid-GS} and \eqref{eq:SDP-rootfid-GSE} can be greatly simplified by employing basic concepts from representation theory (i.e., Schur's lemma). See \cite{S12} for background on representation theory and Propositions~4.2.2 and 4.2.3 therein for Schur's lemma. Focusing on the SDP in 
\eqref{eq:SDP-rootfid-GS}, it is well known that there exists a unitary $W$ that block diagonalizes every unitary in the set $\{U(g)\}_{g\in G}$, as follows:
\begin{equation}
    U(g) = W\left(\bigoplus_\lambda \mathbb{I}_{m_\lambda} \otimes U_\lambda(g)  \right) W^\dag,
    \label{eq:unitary-decomp-group}
\end{equation}
where the variable $\lambda$ labels an irreducible representation (irrep) of $U(g)$, the matrix $\mathbb{I}_{m_\lambda}$ is an identity matrix of dimension $m_\lambda$, and the unitary $U_\lambda(g)$ is an irrep of $U(g)$ with multiplicity $m_\lambda$. This same unitary $W$
induces a direct-sum decomposition (called isotypic decomposition) of the Hilbert space $\mathcal{H}$ for $\rho_S$ and $\sigma_S$ as follows:
\begin{align}
    W^\dag \mathcal{H} & = \bigoplus_\lambda \mathcal{H}_\lambda,\\
    \mathcal{H}_\lambda & \coloneqq   \mathbb{C}^{m_\lambda} \otimes \mathcal{K}_\lambda,
\end{align}
where $\mathcal{H}_\lambda$ is the space on which $\mathbb{I}_{m_\lambda} \otimes U_\lambda(g)  $ acts and  $\mathcal{K}_\lambda$ is the factor on which $U_\lambda(g)$ acts. 
 Noting that the condition 
\begin{equation}
\sigma_{S}=U_{S}(g)\sigma_{S}U_{S}(g)^{\dag} \qquad \forall g\in G
\end{equation}
 is equivalent to 
\begin{equation}
\sigma_S = \mathcal{T}_G(\sigma_S)    ,
\label{eq:symmetrized-states}
\end{equation}
where the group twirl channel is defined as
\begin{equation}
    \mathcal{T}_G(\cdot) \coloneqq \frac{1}{|G|}\sum_{g \in G}U_S(g) (\cdot)U_S(g)^\dag ,
\end{equation}
 it then follows from \eqref{eq:unitary-decomp-group} and Schur's lemma that the twirl channel $\mathcal{T}_G$ has the following form (see page~8 of \cite{BRS07}):
\begin{equation}
    \mathcal{T}_G(\cdot) = \mathcal{W} \circ \left( \sum_\lambda (\operatorname{id}_{m_\lambda} \otimes \mathcal{D}_\lambda) \circ \mathcal{P}_\lambda\right)\circ \mathcal{W}^\dag ,
\end{equation}
where $\mathcal{W}(\cdot) \coloneqq W(\cdot)W^\dag$, the map $\mathcal{P}_\lambda$  projects onto $\mathcal{H}_\lambda$ (i.e., $\mathcal{P}_\lambda(\cdot) \coloneqq \Pi_\lambda (\cdot) \Pi_\lambda$, with $\Pi_\lambda$ the projection onto $\mathcal{H}_\lambda$), the map $\operatorname{id}_{m_\lambda}$ denotes the identity channel acting on the multiplicity space, and $\mathcal{D}_\lambda$ denotes a completely depolarizing channel with the action $\mathcal{D}_\lambda(\cdot) \coloneqq \operatorname{Tr}[\cdot] \pi_\lambda$, with $\pi_\lambda \coloneqq \mathbb{I}_{d_\lambda} / d_\lambda$ and $d_\lambda$ the dimension of $\mathcal{K}_\lambda$. The effect of the twirl $\mathcal{T}_G$ on a general input $\sigma$ is then
\begin{equation}
    \mathcal{T}_G(\sigma) = W\left(\bigoplus_\lambda \operatorname{Tr}_2[\Pi_\lambda W^\dag \sigma W \Pi_\lambda] \otimes \pi_\lambda\right ) W^\dag.
\end{equation}
It then follows that every state satisfying \eqref{eq:symmetrized-states} has the following form:
\begin{equation}
    \sigma_S = W\left( \bigoplus_\lambda  \tilde{\sigma}_\lambda  \otimes \pi_\lambda\right)W^\dag,
\end{equation}
where $\{\tilde{\sigma}_\lambda\}_\lambda$ is a set of positive semi-definite operators such that  $\sum_\lambda \operatorname{Tr}[\tilde{\sigma}_\lambda] = 1$. Thus, when performing the optimization in \eqref{eq:SDP-rootfid-GS}, it suffices to find the diagonalizing unitary $W$ for the representation $\{U(g)\}_{g\in G}$ (for which an algorithm is known \cite[Section~9.2.5]{AL12}) and then optimize over the set $\{\tilde{\sigma}_\lambda\}_\lambda$, thus greatly reducing the space over which the optimization needs to be conducted. This kind of reduction was recently exploited in \cite{FST22}, and a Matlab toolbox was provided in \cite{RMB21}. We note that we can employ similar reasoning to simplify the optimization in \eqref{eq:SDP-rootfid-GSE}.

It also follows from Schur's lemma that the group projection $\Pi^G_S$ has the following form \cite[Eqs.~(1)--(2)]{Cub18}:
\begin{align}
\Pi^G_S & = W\left( \bigoplus_\lambda \delta_{\lambda, \lambda_t} \mathbb{I}_{m_\lambda}  \otimes \mathbb{I}_{d_\lambda} \right)W^\dag, \\
&  = W \Pi_{\lambda_t} W^{\dag},
\end{align}
where $\lambda_t$ is the irrep for the trivial representation of $\{U_S(g)\}_{g \in G}$. Noting that $d_{\lambda_t} =1$ for this irrep, it follows that $\Pi_{\lambda_t}$ acts as $\mathbb{I}_{m_\lambda}$ on this subspace. Thus, in the optimization in \eqref{eq:SDP-rootfid-GBS}, it follows that every state $\sigma_S$ satisfying $\sigma_S = \Pi^G_S\sigma_S\Pi^G_S$ has the following form:
\begin{equation}
W\sigma_{\lambda_t} W^\dag,
\end{equation}
where $\sigma_{\lambda_t}$ is a state with support only in the space $\mathcal{H}_{\lambda_t}$, i.e., satisfying $\sigma_{\lambda_t} = \Pi_{\lambda_t}\sigma_{\lambda_t}\Pi_{\lambda_t}$. In this way, we can simplify the optimization task in \eqref{eq:SDP-rootfid-GBS}. We finally note that we can employ similar reasoning to simplify the optimization in \eqref{eq:SDP-rootfid-GSE}.

%The complexity of the twirl operation from \eqref{eq:SDP-rootfid-GS} and \eqref{eq:SDP-rootfid-GSE} can reduced by decomposing the Hilbert space \cite{BRS07}.

\section{Variational algorithms for testing symmetry}

\label{sec:var-algs}

Having established that the  acceptance probabilities can be computed by SDPs for circuits on a sufficiently small number of qubits, we now propose variational quantum algorithms (VQA) for use on quantum computers as a proof-of-concept implementation of these tests (see \cite{CABBEFMMYCC20,bharti2021noisy}\ for reviews of variational quantum algorithms). These algorithms make use of variational machine learning techniques to mimic the action of the prover in Algorithms~\ref{alg:g-sym-test}--\ref{alg:sym-ext}; however, these techniques are in general limited in terms of their capabilities and thus do not fully satisfy the all-powerful nature of the prover called for in quantum interactive proofs. Note also that training a VQA has been shown to be NP-hard \cite{qvanp}; nonetheless, implementing such methods on near-term quantum devices gives a rough lower bound on the symmetry measures of interest. In the future, more advanced techniques could be substituted into the prover's position in an equivalent manner to improve on these lower-bound estimates. We present here a series of examples and show the circuit diagrams and VQA performance for these tests. To demonstrate the wide-ranging applicability of these algorithms, we have performed symmetry tests for a variety of groups. We present a subset of them now and defer the rest of them to Appendices~\ref{app:cyclic-c-3} through \ref{app:Q8} in the interest of space. 

For the algorithms discussed in this section, all code was implemented in Python using Qiskit (a Python package used for quantum computing with IBM Quantum). For each algorithm, the noiseless variant was implemented using the IBM Quantum noiseless simulator. For the noisy versions, we use the noise model from the IBM-Jakarta quantum computer and conduct a noisy simulation. We find that the algorithms behave well in both scenarios, and for VQA tests, our results converge in a reasonable number of layers, typically less than five. 
In the noisy simulations, the algorithms converge well, and the parameters obtained exhibit a noise resilience as put forward in \cite{Sharma_2020}; that is, the relevant quantity can be accurately estimated by inputting the parameters learned from the noisy simulator into the noiseless simulator. Note that some sections show only a noiseless simulation; for these cases, the noisy simulation requires a noise model of a larger quantum system than is currently available to us.

As with many VQAs, it is necessary in these simulations to endeavor to avoid the barren plateau problem, in which global cost functions become untrainable. The algorithms specified in Section~\ref{sec:tests-o-sym} rely solely on local measurements alone in the regime in which the number of data qubits is much larger than the number of control qubits and thus should not suffer from this issue in this regime \cite{osti_1826512}. Furthermore, all VQAs utilized herein employ the SPSA optimization technique discussed in \cite{spall1998overview}, which aims to prevent local minima problems. Indeed, our simulations did not run into either issue for any of the results discussed. However, we have only considered simulations of small quantum systems; it remains open to provide evidence that our algorithms will avoid the barren plateau problem for larger systems.

%new
Lastly, consider that many of the algorithms in Section~\ref{sec:tests-o-sym} allow the prover access to an environmental system, labelled $E$. A natural question is how best to choose the dimension of this system. 
In general, we find that the $E$ system must be sufficiently large so as to match the input and output qubits, making the entire process unitary. For example, in $G$-symmetry tests,  the dimension of the $E$ system must be sufficiently large to provide a purification of the test state (recall Figure~\ref{fig:case-2}); for instance, if the state under test is a two-qubit state with a three-qubit purification, then $E$ must necessarily provide the remaining qubit to get from the initial three-qubit purification to the four-qubit purification being tested. By construction, the purification of a state under test is always provided to the prover and is not considered part of the environmental system. For all simulations, we have taken the dimension of $E$ to be the minimal viable dimension.

In what follows, we consider several groups and their unitary representations and test states for $G$-Bose symmetry, $G$-symmetry, $G$-Bose symmetric extendibility, and $G$-symmetric extendibility. We also test for two- and three-extendibility.

\subsection{\texorpdfstring{$\mathbb{Z}_2$}{Z2} Group}

\begin{figure*}[t!]
\begin{center}
\includegraphics[width=\linewidth]{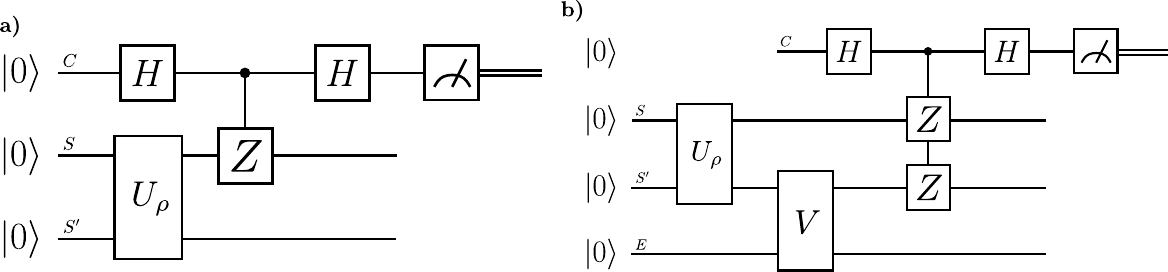}
\end{center}
\caption{Symmetry tests for the $\mathbb{Z}_2$ group: a) $G$-Bose symmetry and b) $G$-symmetry.}
\label{fig:Z2_Circuits}
\end{figure*}

In order to test membership in $\operatorname{Sym}_G$, a group with an established unitary representation is needed. One somewhat trivial, albeit easily testable, example is the group generated by the identity and the Pauli~$Z$ gate. The group table for the $\mathbb{Z}_2$ group is given by

\begin{center}
\begin{tabular}{>{\centering\arraybackslash}p{0.14\textwidth} | >{\centering\arraybackslash}p{0.05\textwidth} |  >{\centering\arraybackslash}p{0.05\textwidth}}
 \hline
  Group element & $e$ & $g$ \\ 
  \hline
 $e$ & $e$ & $g$ \\ 
 $g$ & $g$ & $e$ \\ 
 \hline
\end{tabular}
\end{center}

\noindent where $e$ denotes the identity element. The $\mathbb{Z}_2$ group has a simple one-qubit unitary representation $\{e \rightarrow \mathbb{I}, g \rightarrow Z\}$. Since $\mathbb{Z}_2$ has two elements, the $\ket{+}_C$ state is a uniform superposition of two elements. Thus, we use one qubit and the Hadamard gate to generate the necessary state:
\begin{equation}
    H\ket{0} = \frac{1}{\sqrt{2}}\left(\ket{0} + \ket{1}\right).
\end{equation}
The control register states need to be mapped to group elements. We employ the mapping $\{\ket{0} \rightarrow e, \ket{1} \rightarrow g\}$ for our circuit constructions.

\subsubsection{\texorpdfstring{$G$}{G}-Bose symmetry}
We begin with a test for Bose symmetry, which in this case is a test whether the state is equal to $|0\rangle\!\langle 0|$, because the group projector $\Pi^{\mathbb{Z}_2}_S = (\mathbb{I}+Z)/2 = |0\rangle\!\langle 0|$. Calculation by hand or classical computation can easily verify whether a state is Bose symmetric with respect to $\mathbb{I}$ and $Z$. Additionally, this simple gate set can be easily implemented on existing quantum computers. 

Figure~\ref{fig:Z2_Circuits}a) shows the circuit that tests for this $G$-Bose symmetry. Table~\ref{tab:Z2_GBS} shows the results for various input states. The true fidelity value is calculated using \eqref{eq:acc-prob-bose-test}, where $\Pi^G_S$ is defined in \eqref{eq:group_proj_GBS}.

\begin{table}[h]
\centering
\vspace{.05in}
\begin{tabular}{
>{\centering\arraybackslash}p{0.08\textwidth} | >{\centering\arraybackslash}p{0.08\textwidth} | 
>{\centering\arraybackslash}p{0.08\textwidth} | 
>{\centering\arraybackslash}p{0.08\textwidth} 
}
\hline
\textrm{State} & 
\textrm{True Fidelity} &
\textrm{Noiseless} &
\textrm{Noisy}\\
\hline\hline 
$\outerproj{0}$ & 1   & 1.0 & 0.9998 \\
$\outerproj{1}$ & 0   & 0.0 & 0.0013\\ 
$\outerproj{+}$ & 0.5 & 0.5 & 0.5002\\ 
$\mathbb{I}/2 $ & 0.5 & 0.5 & 0.5092\\ 
\hline
\end{tabular}
\caption{Results of $Z_2$-Bose symmetry tests.}
\label{tab:Z2_GBS}
\end{table}

\subsubsection{\texorpdfstring{$G$}{G}-symmetry}

We now consider a simple test for $G$-symmetry. As mentioned in Remark~\ref{rem:incoherence}, this is also a test for incoherence of the input state, i.e., to determine if it is diagonal in the computational basis. In the circuit depicted in Figure~\ref{fig:Z2_Circuits}b), a parameterized circuit substitutes the role of an all-powerful prover. 

A circuit that tests for $G$-symmetry is shown in Figure~\ref{fig:Z2_Circuits}b). As this circuit involves variational parameters, an example of the training process is shown in Figure~\ref{fig:Z2_GS_Training}. Table~\ref{tab:Z2_GS} shows the final results after training for various input states. The true fidelity is calculated using the semi-definite program given in \eqref{eq:SDP-rootfid-GS} and is used as a comparison point.

\begin{figure}
\begin{center}
\includegraphics[width=\linewidth]{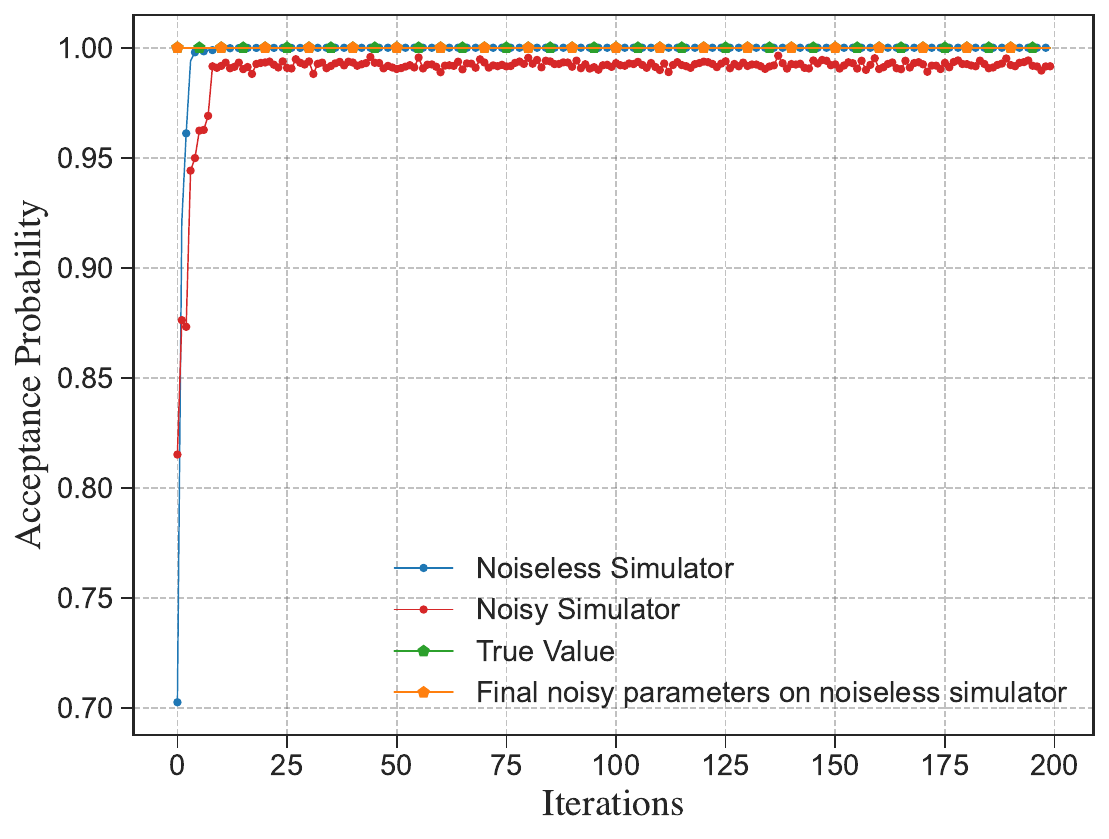}
\end{center}
\caption{Example of the training process for testing $\mathbb{Z}_2$-symmetry of $\rho = \mathbb{I}/2$. We see that the training exhibits a noise resilience.}
\label{fig:Z2_GS_Training}
\end{figure}

\begin{table}[h]
\vspace{.05in}
\centering
\begin{tabular}{
>{\centering\arraybackslash}p{0.05\textwidth} | >{\centering\arraybackslash}p{0.07\textwidth} | 
>{\centering\arraybackslash}p{0.08\textwidth} | 
>{\centering\arraybackslash}p{0.06\textwidth} |
>{\centering\arraybackslash}p{0.08\textwidth}}
\hline
\textrm{State} & 
\textrm{True Fidelity} &
\textrm{Noiseless} &
\textrm{Noisy} & 
\textrm{Noise Resilient}
\\
\hline\hline 
$\outerproj{0}$ & 1   & 0.9999 & 0.9987 & 0.9999 \\
$\outerproj{1}$ & 1   & 1.0 & 1.0 & 0.9999\\ 
$\outerproj{+}$ & 0.5 & 0.5 & 0.5087 & 0.5\\ 
$\mathbb{I}/2 $ & 1 & 0.9999 & 0.9932 & 0.9999\\ 
\hline
\end{tabular}
\caption{Results of $Z_2$-symmetry tests.}
\label{tab:Z2_GS}
\end{table}

\subsection{Triangular dihedral group \texorpdfstring{$D_3$}{D3}}

\subsubsection{\texorpdfstring{$G$}{G}-Bose symmetry}

Throughout Section~\ref{sec:tests-o-sym}, we have used the dihedral group of the equilateral triangle, abbreviated as $D_3$, as an example, and we continue to do so now. As a reminder, this group is generated by a flip of order two and a rotation of order three (denoted respectively by $f$ and $r$). Then the group is specified as $D_3 = \{e,f,r,r^2,fr,fr^2\}$ where $e$ is the identity element. General dihedral groups have previously been studied as non-abelian groups for which a quantum  algorithm to find a hidden subgroup is available \cite{kuperberg2005subexponential}. 

In the introduction of Section~\ref{sec:tests-o-sym}, we provided a faithful, projective unitary representation of this group given by letting $U(f)=\textrm{CNOT}$, $U(r)= \textrm{CNOT} \cdot \textrm{SWAP}$, and $U(e)=\mathbb{I}_4$. Figure~\ref{fig:Dihedral_GBS_Circuit} shows the circuit needed to test for $G$-Bose symmetry. Note that we do not generate the control register using a quantum Fourier transform; as the resultant control state is still equivalent to $\ket{+}_C = \frac{1}{\sqrt{6}} \sum_{g \in D_3} \ket{g}$, this simplification suffices for our calculations. Table~\ref{tab:Dihedral_GBS} shows the results for various input states. The true fidelity value is calculated using \eqref{eq:acc-prob-bose-test}, where $\Pi^G_S$ is defined in~\eqref{eq:group_proj_GBS}.

\begin{table}[h]
\centering
\vspace{.05in}
\begin{tabular}{
>{\centering\arraybackslash}p{0.07\textwidth} | >{\centering\arraybackslash}p{0.08\textwidth} | 
>{\centering\arraybackslash}p{0.08\textwidth} | 
>{\centering\arraybackslash}p{0.08\textwidth}}
\hline
\textrm{State} & 
\textrm{True Fidelity} &
\textrm{Noiseless} &
\textrm{Noisy}\\
\hline\hline 
$\outerproj{00}$ & 1 & 1.0000 & 0.9998 \\
$\rho$ & 1 & 0.9999 & 0.8756\\ 
$\Phi^+$ & 0.6666 & 0.6666 & 0.5864\\
$\pi^{\otimes 2}$ & 0.5 & 0.5000 & 0.4716\\
\hline
\end{tabular}
\caption{Results of $D_3$-Bose symmetry tests. The state $\rho$ is defined as $\outerproj{\psi}$ where $\ket{\psi} = \frac{1}{\sqrt{3}}(\ket{01}+\ket{10}+\ket{11})$.}
\label{tab:Dihedral_GBS}
\end{table}

\subsubsection{\texorpdfstring{$G$}{G}-symmetry}

As with $\mathbb{Z}_2$, moving to $G$-symmetry requires the addition of a prover. This alteration was already depicted in Figure~\ref{fig:Dihedral_GS_Circuit}. The prover is replaced for practical purposes with a parameterized circuit involving variational parameters, and the training process is shown in Figure~\ref{fig:Dihedral_GS_Training}. Table~\ref{tab:Dihedral_GS} shows the final results after training for various input states. The true fidelity is calculated using the semi-definite program given in~\eqref{eq:SDP-rootfid-GS}.

\begin{figure}
\begin{center}
\includegraphics[width=\linewidth]{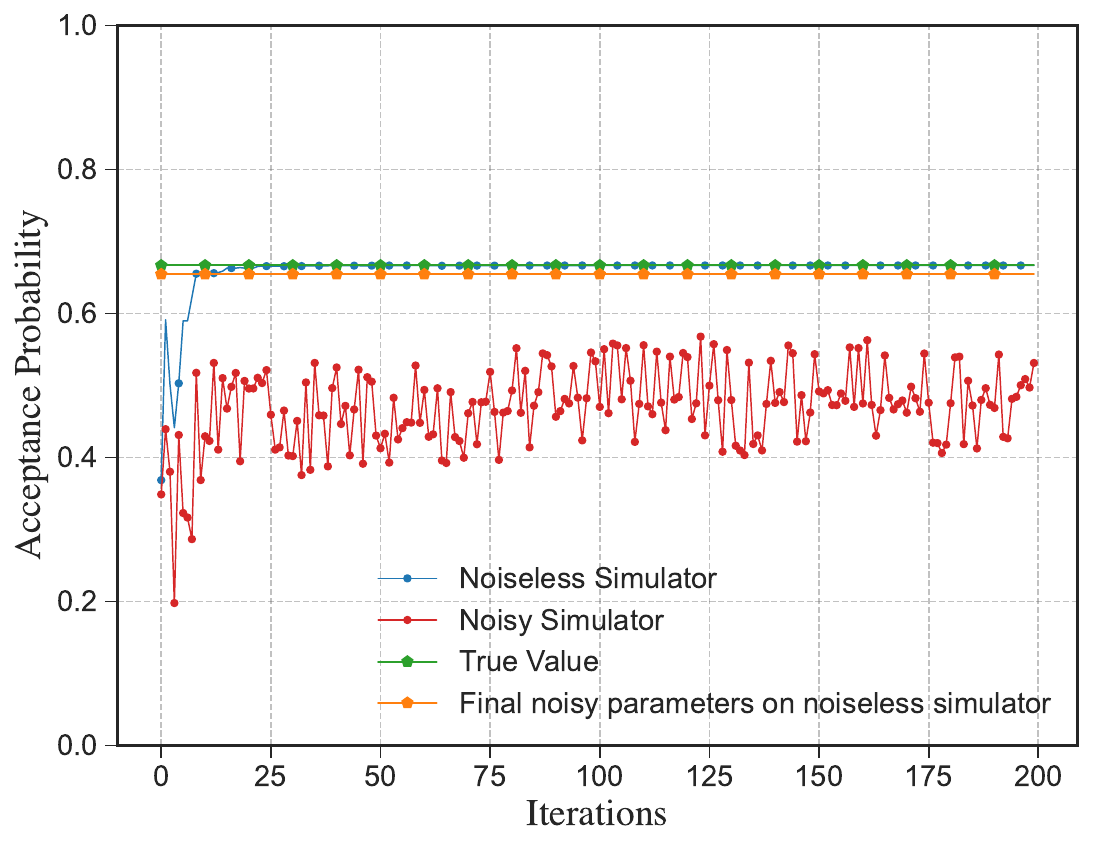}
\end{center}
\caption{Example of the training process for testing $D_3$-symmetry of $\Phi^+$. We see that the training exhibits a noise resilience.}
\label{fig:Dihedral_GS_Training}
\end{figure}

\begin{table}[h]
\centering
\vspace{.05in}
\begin{tabular}{
>{\centering\arraybackslash}p{0.06\textwidth} | >{\centering\arraybackslash}p{0.07\textwidth} | 
>{\centering\arraybackslash}p{0.07\textwidth} | 
>{\centering\arraybackslash}p{0.07\textwidth} |
>{\centering\arraybackslash}p{0.08\textwidth}}
\hline
\textrm{State} & 
\textrm{True Fidelity} &
\textrm{Noiseless} &
\textrm{Noisy} & 
\textrm{Noise Resilient}
\\
\hline\hline 
$\outerproj{00}$ & 1.0000 & 0.9999 & 0.9987 & 0.9999 \\
$\rho$ & 1.0000 & 0.9999 & 0.6564 & 0.9425 \\ 
$\Phi^+$ & 0.6666 & 0.6666 & 0.5330 & 0.6415 \\
$\pi^{\otimes 2} $ & 1.0000 & 0.9989 & 0.5189 & 0.8712\\ 
\hline
\end{tabular}
\caption{Results of $D_3$-symmetry tests. The state $\rho$ is defined as $\outerproj{\psi}$ where $\ket{\psi} = \frac{1}{\sqrt{3}}(\ket{01}+\ket{10}+\ket{11})$.}
\label{tab:Dihedral_GS}
\end{table}

\subsubsection{\texorpdfstring{$G$}{G}-Bose symmetric extendibility}

\begin{figure}
\begin{center}
\includegraphics[width=\linewidth]{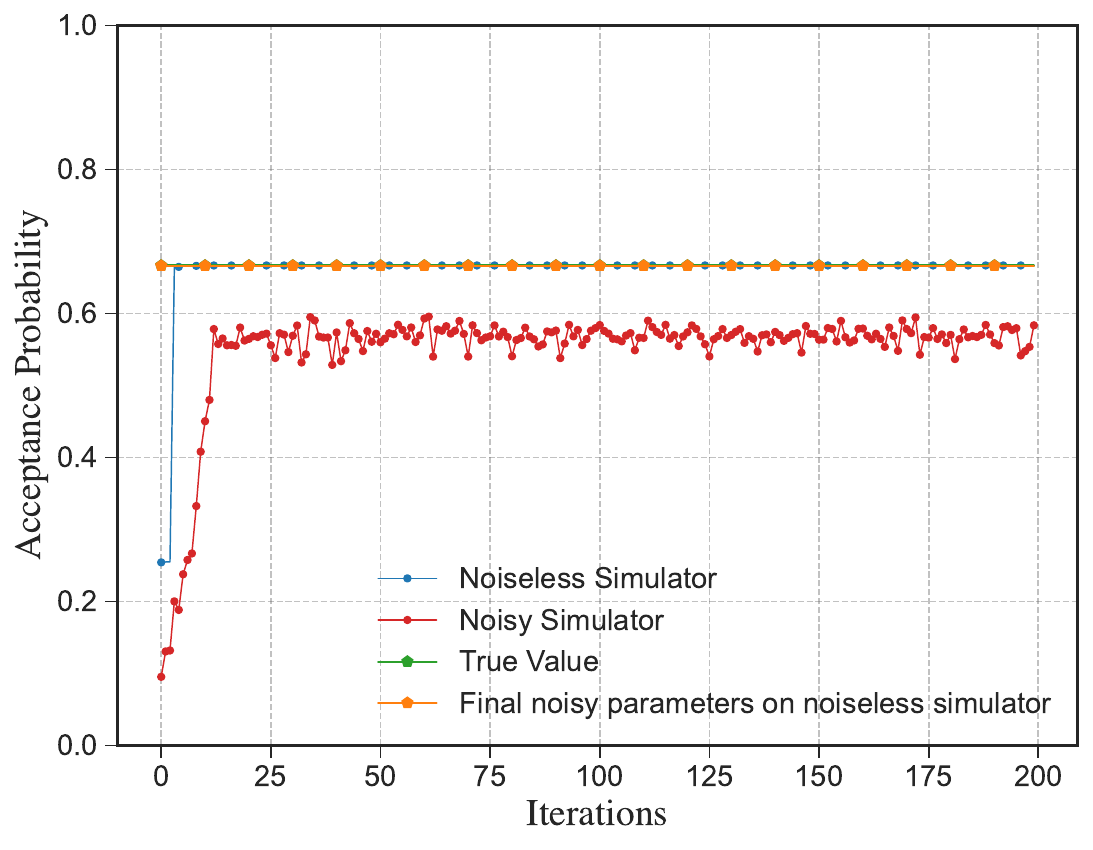}
\end{center}
\caption{Example of the training process for testing $D_3$-Bose symmetric extendibility of $\outerproj{1}$. We see that the training exhibits a noise resilience.}
\label{fig:Dihedral_GBSE_Training}
\end{figure}

A circuit that tests for $G$-Bose symmetric extendibility was originally shown in Figure~\ref{fig:Dihedral_GBSE_Circuit} as the example circuit construction. Now, we show how that construction behaves under a parameterized circuit substitution of the prover. Again, we give an example of the training behavior of the algorithm in Figure~\ref{fig:Dihedral_GBSE_Training}. We also provide Table~\ref{tab:Dihedral_GBSE}, which shows the final results after training for various input states. The true fidelity is calculated using the semi-definite program given in \eqref{eq:SDP-rootfid-GBSE}.

\begin{table}
\centering
\vspace{.05in}
\begin{tabular}{
>{\centering\arraybackslash}p{0.06\textwidth} | >{\centering\arraybackslash}p{0.07\textwidth} | 
>{\centering\arraybackslash}p{0.07\textwidth} | 
>{\centering\arraybackslash}p{0.05\textwidth} |
>{\centering\arraybackslash}p{0.08\textwidth}}
\hline
\textrm{State} & 
\textrm{True Fidelity} &
\textrm{Noiseless} &
\textrm{Noisy} & 
\textrm{Noise Resilient}
\\
\hline\hline 
$\outerproj{0}$ & 1.0000 & 1.0000 & 0.8758 & 0.9988 \\
$\outerproj{1} $ & 0.6670 & 0.6667 & 0.5834 & 0.6663 \\
$\pi$ & 1.0000 & 1.0000 & 0.8255 & 0.9995 \\
$\begin{bmatrix} \frac{1}{3} & \frac{1}{3}\\ \frac{1}{3} & \frac{2}{3} \end{bmatrix}$ & 1.0000 & 0.9999 & 0.6564 & 0.9425 \vspace{0.35cm}\\
\hline
\end{tabular}
\caption{Results of $D_3$-Bose symmetric extendibility tests.}
\label{tab:Dihedral_GBSE}
\end{table}

\subsubsection{\texorpdfstring{$G$}{G}-symmetric extendibility}

Finally, we address the circuit in Figure~\ref{fig:Dihedral_GSE_Circuit}, which gives a test for $G$-symmetric extendibility. This final circuit has the prover performing two actions at once---both finding the correct purification as in the case of $G$-symmetry and creating the correct extension as in $G$-Bose symmetric extendibility tests. Once again, the prover is replaced with a parameterized circuit, and an example of the training process is shown in Figure~\ref{fig:Dihedral_GSE_Training}. Table~\ref{tab:Dihedral_GSE} shows the final results after training for various input states. The true fidelity is calculated using the semi-definite program given in \eqref{eq:SDP-rootfid-GSE}.

\begin{figure}
\begin{center}
\includegraphics[width=\linewidth]{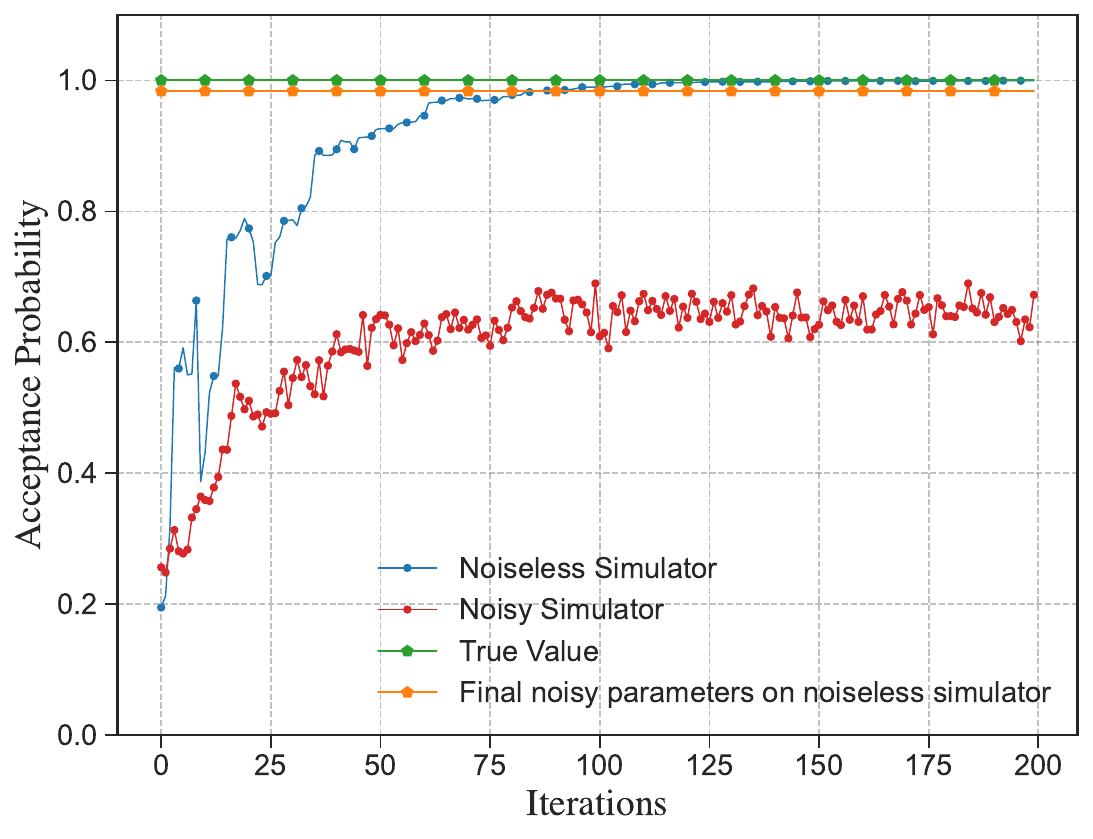}
\end{center}
\caption{Example of the training process for testing $D_3$-symmetric extendibility of $\outerproj{0}$. We see that the training exhibits a noise resilience.}
\label{fig:Dihedral_GSE_Training}
\end{figure}

\begin{table}
\centering
\vspace{.05in}
\begin{tabular}{
>{\centering\arraybackslash}p{0.04\textwidth} | >{\centering\arraybackslash}p{0.07\textwidth} | 
>{\centering\arraybackslash}p{0.08\textwidth} | 
>{\centering\arraybackslash}p{0.07\textwidth} |
>{\centering\arraybackslash}p{0.08\textwidth}}
\hline
\textrm{State} & 
\textrm{True Fidelity} &
\textrm{Noiseless} &
\textrm{Noisy} & 
\textrm{Noise Resilient}
\\
\hline\hline 
$\outerproj{0}$ & 1.0000 & 0.9998 & 0.6725 & 0.9835 \\
$\outerproj{1}$ & 0.6666 & 0.6641 & 0.4476 & 0.6497 \\
$\pi$ & 1.0000 & 0.9988 & 0.6901 & 0.9764 \\
$\rho$ & 0.9714 & 0.9662 & 0.5593 & 0.8789 \\
\hline
\end{tabular}
\caption{Results of $D_3$-symmetric extendibility tests. The state $\rho$ is defined as $\begin{bmatrix} 0.5 & -0.354i\\ 0.354i & 0.5 \end{bmatrix}$.}
\label{tab:Dihedral_GSE}
\end{table}

\subsection{Collective \texorpdfstring{$U$}{U} group}

\label{sec:collective-U}

Given an $n$-qudit state $\rho$, we wish to test if it is symmetric with respect to the following group:
\begin{equation}
    G_U \coloneqq \{U^{\otimes n}\}_{U \in \operatorname{SU}(d)}.
\end{equation}
This is an example of a continuous group symmetry; however, we will be able to draw upon the particular properties of this projector to realize each symmetry test nonetheless.

\subsubsection{\texorpdfstring{$G$}{G}-Bose symmetry}

A state that is $G_U$-Bose symmetric satisfies the condition given in \eqref{eq:Bose-symmetric-equiv-cond}, where 
\begin{equation}
    \Pi_U^{(n)} \coloneqq \int dU \ U^{\otimes n} ,
\end{equation}
with $dU$ being the Haar measure for the group $\operatorname{SU}(d)$.

In what follows, we focus on two-qubit states. A simple calculation shows that for $n=2$ and $d=2$, the singlet state $\ket{\Psi^-} \coloneqq \frac{1}{\sqrt{2}}\left(\ket{01} - \ket{10}\right)$, is the only $G_U$-Bose symmetric state. In other words, 
\begin{equation}
    \Pi^{(2)}_U = \outerproj{\Psi^-}.
\end{equation}
Thus, testing for $G_U$-Bose symmetry is equivalent to testing if the state is the singlet state.

To test a symmetry of this form, we rewrite the projector in terms of a set  $\{U_i\}_{i=1}^N$ of unitaries satisfying
\begin{equation}
    \Pi^{(2)}_U = \frac{1}{N}\sum\limits_{i=1}^N U_i.
    \label{eq:Bose-proj-2-design}
\end{equation}
While there exist multiple choices for the set  $\{U_i\}_{i=1}^N$, we pick a set that is compatible with all of the symmetry tests that we perform in the forthcoming subsections. Our choice $\{U_i\}_{i=1}^N$ is given in \cite[Appendix~A]{bennett96mixed} and is composed of products of bilateral rotations $B_{x}$, $B_{y}$, and $B_{z}$, where
\begin{equation}
    B_a \coloneqq R_a(-\pi/2) \otimes R_a(-\pi/2)\, ,
\end{equation}
and $R_a$ is the following rotation gate about the $a$ axis:
\begin{align}
    R_a(\theta) & \coloneqq e^{-i \theta \sigma_a /2} \\
    & = \cos(\theta/2)\mathbb{I} - i\sin(\theta/2)\sigma_a.
\end{align}
(Note the different convention that we take here, as compared to \cite{bennett96mixed}, when defining bilateral rotations.)
Specifically, the set $\{U_i\}_i$ is given by
\begin{multline}
\label{eq:CU-2-design}
    \{U_i\}_i = \{\mathbb{I},\ B_xB_x,\ B_yB_y,\ B_zB_z,\ B_xB_y,\ B_yB_z,\\  B_zB_x,\    B_yB_x,\ B_xB_yB_xB_y,\ \\
    B_yB_zB_yB_z, \ B_zB_xB_zB_x,\ B_yB_xB_yB_x \}.
\end{multline}

The set $\{U_i\}_i$ forms a group isomorphic to the alternating group $A_4$, which is defined as the set of even permutations on four objects. Furthermore, $A_4$ can be written as a product of a Klein group on four objects $K_4 = \{e, a=(12)(34), b=(13)(24), c=(14)(23)\}$ and the cyclic group $C_3 = \{e, g=(123), h=(132)\}$. In other words,
\begin{equation}
    A_4 = K_4 \times C_3.
\end{equation}

The Klein group $K_4$ can be mapped as $\{e \rightarrow \mathbb{I}, a \rightarrow B_xB_x, b \rightarrow B_yB_y, c \rightarrow B_zB_z\}$. Similarly, the cyclic group can be mapped as $\{e \rightarrow \mathbb{I}, g \rightarrow B_xB_y, h \rightarrow B_yB_x\}$. We use this to design our control register and corresponding mapping there. Since we have 12 elements, the $\ket{+}_C$ state is a uniform superposition of 12 elements. However, the aforementioned decomposition allows us to split the control register into two sets, one controlling the $K_4$ group and another controlling the $C_3$ group. We use a unary encoding for both subgroups, leading to a five-qubit control register. The specific mapping and group assignment are as follows: 
\begin{center}
\begin{tabular}{
>{\centering\arraybackslash}p{0.12\textwidth} | >{\centering\arraybackslash}p{0.12\textwidth} | 
>{\centering\arraybackslash}p{0.12\textwidth}}
 \hline
  Control State & Group Element & Unitary Representation \\ 
  \hline
 00 000 &  $e$ & $\mathbb{I}$ \\ 
 00 100 &  $c$ & $B_zB_z$ \\ 
 00 010 &  $b$ & $B_yB_y$ \\ 
 00 001 &  $a$ & $B_xB_x$ \\
 
 01 000 &  $g$ & $B_xB_y$ \\ 
 01 100 & $gc$ & $B_yB_z$ \\ 
 01 010 & $gb$ & $B_zB_x$ \\ 
 01 001 & $ga$ & $B_yB_xB_yB_x$ \\
 
 10 000 &  $h$ & $B_yB_x$ \\ 
 10 100 & $hc$ & $B_yB_zB_yB_z$ \\ 
 10 010 & $hb$ & $B_xB_yB_xB_y$ \\ 
 10 001 & $ha$ & $B_zB_xB_zB_x$ \\
 \hline
\end{tabular}
\end{center}

To generate an equal superposition of the 12 basis elements, we use the unitary $U_W$ depicted in Figure~\ref{fig:CollectiveU_Superposition}. With this construction settled, we can now test for symmetry with respect to this collective $U$ group.

\begin{figure}
\begin{center}
\includegraphics[width=0.65\linewidth]{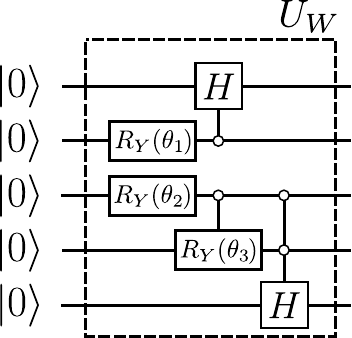}
\end{center}
\caption{Unitary $U_W$, with $\theta_1 = \theta_3 = 2 \arctan \left(\frac{1}{\sqrt{2}}\right)$ and $\theta_2 = \pi/3$, generates the equal superposition of 12 elements given. The circuit acting on the top two qubits generates the state $(\ket{00}+\ket{01}+\ket{10})/\sqrt{3}$, and the circuit acting on the bottom three qubits generates the state $(\ket{000}+\ket{001}+\ket{010}+\ket{100})/\sqrt{4}$.}
\label{fig:CollectiveU_Superposition}
\end{figure}

\begin{figure*}[t!]
\begin{center}
\includegraphics[width=\linewidth]{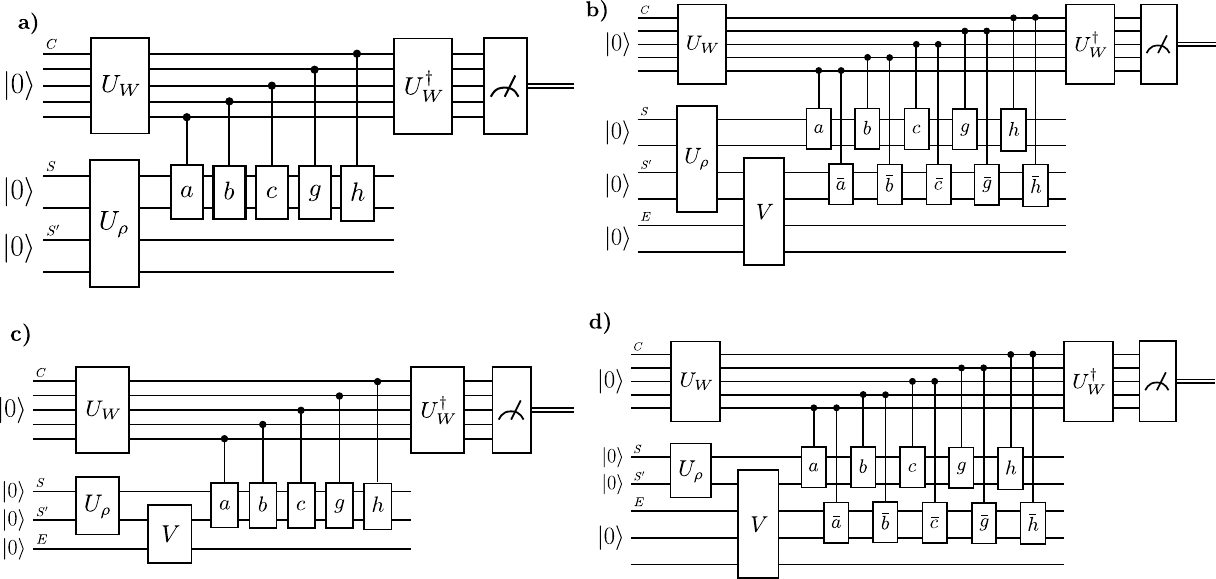}
\end{center}
\caption{Symmetry tests for the collective-$U$ group: a) $G$-Bose symmetry, b) $G$-symmetry, c) $G$-Bose symmetric extendible, and d) $G$-symmetric extendible. }
\label{fig:CU_Circuits}
\end{figure*}

Figure~\ref{fig:CU_Circuits}a) depicts the circuit that tests for $G$-Bose symmetry. Table~\ref{tab:CollectiveU_GBS} shows the results for various input states. The true fidelity value is calculated using \eqref{eq:acc-prob-bose-test}, where $\Pi^G_S$ is defined in \eqref{eq:group_proj_GBS}.

\begin{table}[h]
\centering
\vspace{.05in}
\begin{tabular}{
>{\centering\arraybackslash}p{0.08\textwidth} | >{\centering\arraybackslash}p{0.08\textwidth} | 
>{\centering\arraybackslash}p{0.08\textwidth} | 
>{\centering\arraybackslash}p{0.08\textwidth} 
}
\hline
\textrm{State} & 
\textrm{True Fidelity} &
\textrm{Noiseless} &
\textrm{Noisy}\\
\hline\hline 
$\outerproj{00}$ & 0  & 0.0000 & 0.0459 \\
$\rho$ & 0.6667 & 0.6667 & 0.2661\\ 
$\Psi^+$ & 0 & 0.0000 & 0.0389\\ 
$\Psi^- $ & 1.0 & 1.0000 & 0.3517\\ 
\hline
\end{tabular}
\caption{Results of collective $U$-Bose symmetry tests. The state $\rho$ is defined as $\outerproj{\psi}$ where $\ket{\psi} = \frac{1}{\sqrt{3}}\left(\ket{00} - \ket{01} + \ket{10}\right)$.}
\label{tab:CollectiveU_GBS}
\end{table}

\subsubsection{\texorpdfstring{$G$}{G}-symmetry}

An $n$-qudit state $\rho$ that is $G_U$-symmetric satisfies the following condition:
\begin{equation}
    \rho = \int dU \ U^{\otimes n} \rho (U^\dagger)^{\otimes n} ,
\end{equation}
where $dU$ is the Haar measure for the group $\operatorname{SU}(d)$. States that satisfy this condition for $n=2$ are called Werner states \cite{Werner89}, i.e.,
\begin{equation}
    \rho = \int dU\ (U \otimes U)\,\rho\,(U \otimes U)^\dagger .
\end{equation}
As shown in \cite{bennett96mixed}, for $n=2$ and $d=2$, the continuum of rotations in the symmetry test can be replaced by a discrete sum (a two-design), as follows:
\begin{equation}
    \bar{\rho} = \frac{1}{N}\sum\limits_{i=1}^{N} U_i \rho U^\dagger_i,
\end{equation}
where $\{U_i\}_{i=1}^N$ is the set defined in \eqref{eq:CU-2-design}. A circuit that tests for $G$-symmetry is shown in Figure~\ref{fig:CU_Circuits}b). It involves variational parameters, and an example of the training process is shown in Figure~\ref{fig:CollectiveU_GS_Training}. Note that, as this construction requires many qubits, only noiseless simulations results could be obtained. These results may be easily extended as access to higher-qubit machines becomes more readily available, allowing for noisy simulations of more complex systems. Table~\ref{tab:CollectiveU_GS} shows the final results after training for various input states. The true fidelity is calculated using the semi-definite program given in \eqref{eq:SDP-rootfid-GS}. 

\begin{figure}
\begin{center}
\includegraphics[width=\linewidth]{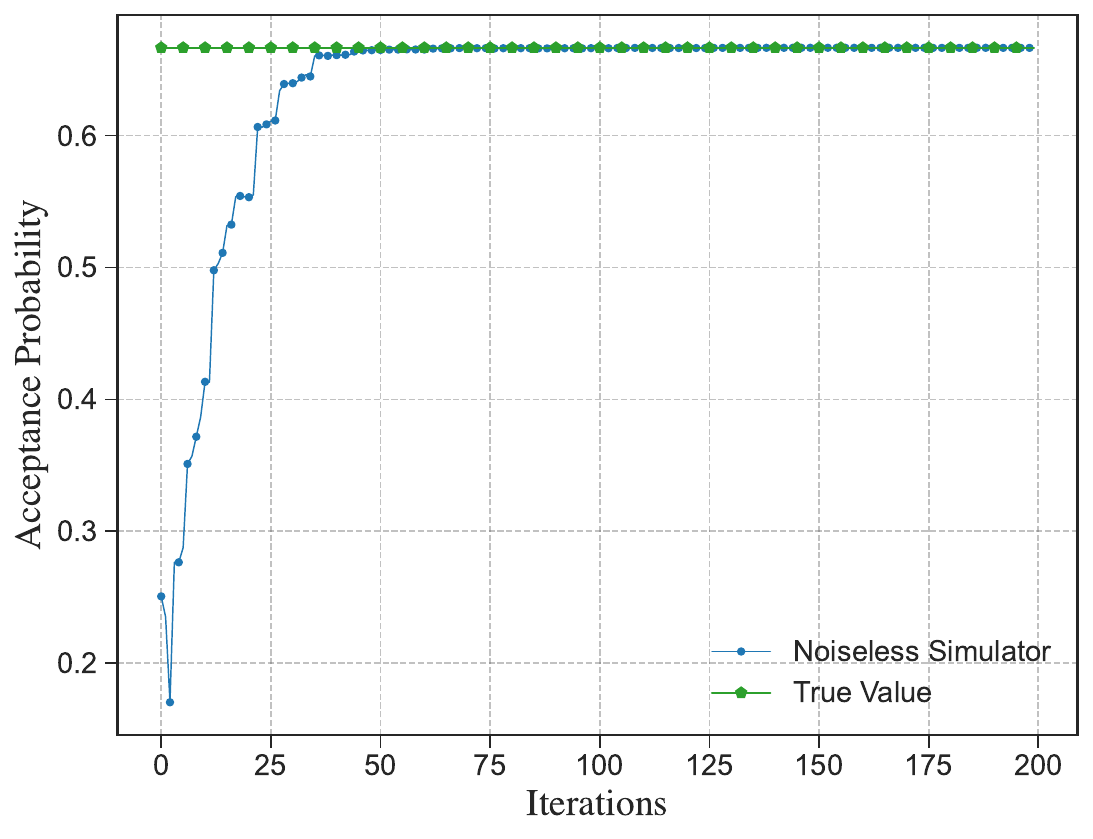}
\end{center}
\caption{Example of the training process for testing collective $U$-symmetry of $\rho = \outerproj{\psi}$ where $\ket{\psi} = \frac{1}{\sqrt{3}}(\ket{00} - \ket{01} + \ket{10})$.}
\label{fig:CollectiveU_GS_Training}
\end{figure}

\begin{table}[h]
\centering
\vspace{.05in}
\begin{tabular}{
>{\centering\arraybackslash}p{0.08\textwidth} | >{\centering\arraybackslash}p{0.08\textwidth} | 
>{\centering\arraybackslash}p{0.08\textwidth}}
\hline
\textrm{State} & 
\textrm{True Fidelity} &
\textrm{Noiseless}
\\
\hline\hline 
$\outerproj{10}$ & 0.5000 & 0.4997 \\
$\rho$ & 0.6667 & 0.6666\\
$\Psi^+$ & 0.3333 & 0.3332 \\
$\pi^{\otimes 2}$ & 1.0000 & 0.9988\\
\hline
\end{tabular}
\caption{Results of collective $U$-symmetry tests. The state $\rho$ is defined as $\outerproj{\psi}$ where $\ket{\psi} = \frac{1}{\sqrt{3}}(\ket{00} - \ket{01} + \ket{10})$.}
\label{tab:CollectiveU_GS}
\end{table}

We note here that the $G_U$-symmetry test would be unaffected by redefining the integral over all unitaries $U \in \operatorname{U}(2)$ without the restriction to $\operatorname{SU}(2)$. However, the projector for the $G_U$-Bose symmetry test would be as follows in that case:
\begin{equation}
    \Pi_U = \int_{U \in \operatorname{U}(2)} dU \ U \otimes U = 0, 
\end{equation}
making the test trivial. Thus, in the previous section, we chose to restrict the group to $\operatorname{SU}(2)$ unitaries.

\subsubsection{\texorpdfstring{$G$}{G}-Bose symmetric extendibility}

A circuit that tests for $G$-Bose symmetric extendibility is shown in Figure~\ref{fig:CU_Circuits}c). It involves variational parameters, and an example of the training process is shown in Figure~\ref{fig:CollectiveU_GBSE_Training}. Table~\ref{tab:CollectiveU_GBSE} shows the final results after training for various input states. The true fidelity is calculated using the semi-definite program given in \eqref{eq:SDP-rootfid-GBSE}.

\begin{figure}[h!]
\begin{center}
\includegraphics[width=\linewidth]{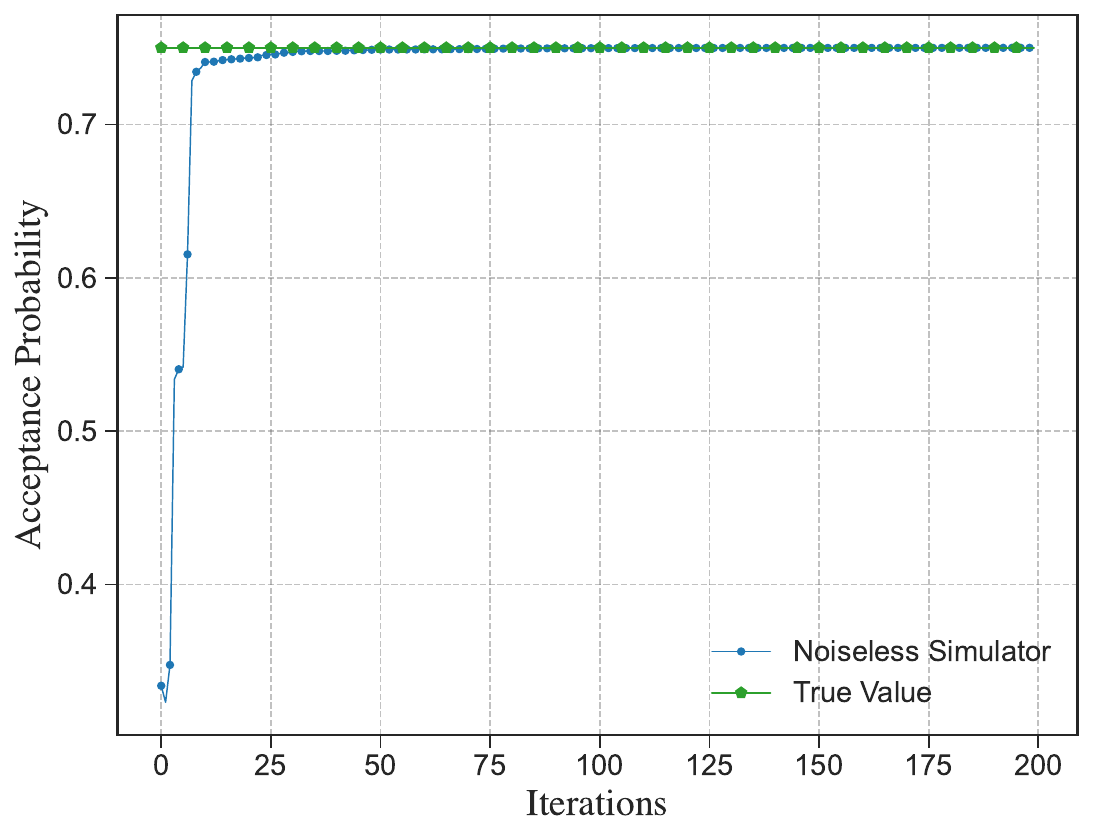}
\end{center}
\caption{Example of the training process for testing collective $U$-Bose symmetric extendibility of the state $\begin{bmatrix} 0.93 & 0\\ 0 & 0.07 \end{bmatrix}$.}
\label{fig:CollectiveU_GBSE_Training}
\end{figure}

\begin{table}[h]
\centering
\vspace{.05in}
\begin{tabular}{
>{\centering\arraybackslash}p{0.12\textwidth} | >{\centering\arraybackslash}p{0.08\textwidth} | 
>{\centering\arraybackslash}p{0.08\textwidth}}
\hline
\textrm{State} & 
\textrm{True Fidelity} &
\textrm{Noiseless}
\\
\hline\hline 
$\outerproj{1}$ & 0.5000 & 0.5000\\
$\pi$ & 1.0000 & 0.9998 \\ 
$\begin{bmatrix} 0.93 & 0\\ 0 & 0.07 \end{bmatrix}$ & 0.7500 & 0.7499 \vspace{0.35cm}\\
\hline
\end{tabular}
\caption{Results of collective $U$-BSE tests.}
\label{tab:CollectiveU_GBSE}
\end{table}

\subsubsection{\texorpdfstring{$G$}{G}-symmetric extendibility}

\begin{table}
\centering
\vspace{.05in}
\begin{tabular}{
>{\centering\arraybackslash}p{0.12\textwidth} | >{\centering\arraybackslash}p{0.12\textwidth} | 
>{\centering\arraybackslash}p{0.08\textwidth}}
\hline
\textrm{State} & 
\textrm{True Fidelity} &
\textrm{Noiseless}
\\
\hline\hline 
$\outerproj{0}$ & 0.5000 & 0.4995 \\
$\pi$ & 1.0000 & 0.9996 \\
$\begin{bmatrix} 0.95 & 0\\ 0 & 0.05 \end{bmatrix}$ & 0.7169 & 0.7095 \vspace{0.35cm}\\
\hline
\end{tabular}
\caption{Results of collective $U$-symmetric extendibility tests.}
\label{tab:CollectiveU_GSE}
\end{table}

A circuit that tests for $G$-symmetric extendibility is shown in Figure~\ref{fig:CU_Circuits}d). It involves variational parameters, and an example of the training process is shown in Figure~\ref{fig:CollectiveU_GSE_Training}. Table~\ref{tab:CollectiveU_GSE} shows the final results after training for various input states. The true fidelity is calculated using the semi-definite program given in \eqref{eq:SDP-rootfid-GSE}.

These group symmetry tests have applications in the identification and verification of Werner states, as discussed above. Current limitations include access to higher qubit machines, but also the noisiness of these machines. Our VQA results converge well in the noiseless case, but it is likely that noise will only become a bigger problem as the circuit size scales up, unless adequately addressed.

\begin{figure}
\begin{center}
\includegraphics[width=\linewidth]{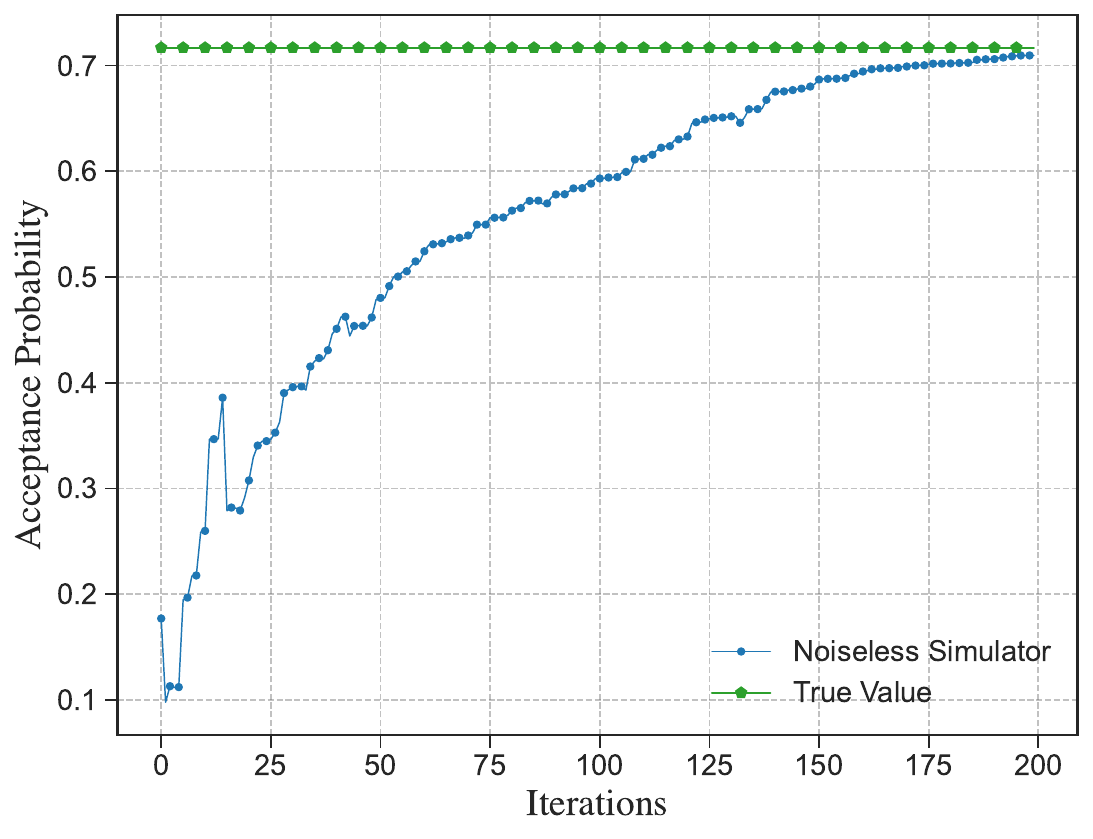}
\end{center}
\caption{Example of the training process for testing collective $U$-symmetric extendibility of the state $\begin{bmatrix} 0.95 & 0\\ 0 & 0.05 \end{bmatrix}$.}
\label{fig:CollectiveU_GSE_Training}
\end{figure}

\subsection{Collective phase group}

\label{sec:collective-z}

Given an $n$-qubit state $\rho$, we wish to test if the state is symmetric with respect to the following collective phase group:
\begin{equation}
    G_z \coloneqq \{R_z(\phi)^{\otimes n}\}_{\phi \in [0, 4\pi)},
\end{equation}
where we recall that $R_z(\phi) \coloneqq \exp(-i\phi\sigma_z/2)$.
The interval for $\phi$ is $[0, 4\pi)$ to ensure that $G_z$ is a group. This is a consequence of $\operatorname{SU}(2)$ double covering $\operatorname{SO}(3)$, implying that $R_z(4\pi) = \mathbb{I}$. Additionally, the Haar measure for the group of unitaries $\{R_z(\phi)\}_{\phi \in [0, 4\pi)}$ is given by
\begin{equation}
    dU = \frac{d\phi}{4\pi}.
\end{equation}

\subsubsection{\texorpdfstring{$G$}{G}-Bose symmetry}

A state that is $G_z$-Bose symmetric satisfies the condition given in \eqref{eq:Bose-symmetric-equiv-cond}, where 
\begin{equation}
    \Pi_z^{(n)} \coloneqq \frac{1}{4\pi}\intCZ R_z(\phi)^{\otimes n} \,d\phi . 
\end{equation}
Expressing $R_z(\phi)$ in the computational basis, 
\begin{equation}
    R_z(\phi) = \operatorname{Diag}\left\{\exp\!\left(-\frac{i\phi}{2} \right), \exp\!\left(\frac{i\phi}{2} \right)\right\}.
\end{equation}
Similarly, expressing $R_z(\phi)^{\otimes 2}$ in the computational basis, 
\begin{equation}
    R_z(\phi)^{\otimes 2} = \operatorname{Diag}\left\{\exp\left(-i\phi\right), 1, 1, \exp\left(i\phi\right)\right\}.
\end{equation}
Generalizing to the case of $n$ qubits, observe that the number of zeros in a bit-string $x$ is $n-H(x)$ and the number of ones is $H(x)$, where $H(x)$ is the Hamming weight of $x$. For example, $H(6) = 2$ since $6_{10} \equiv 110_2$. Each zero contributes a phase of $-\phi/2$ for a total of $-(n-H(x))\phi/2$, and each one contributes a phase of $\phi/2$,  for a total of $H(x) \phi / 2$. Then the overall total for the bit-string $x$ is 
\begin{equation}
-(n-H(x))\phi/2 + H(x) \phi / 2 = (2H(x) - n)\phi/2.   
\end{equation}
This implies that 
\begin{equation}
    R_z(\phi)^{\otimes n} = \operatorname{Diag}\left\{ \exp \left[ \left( \frac{2H(x)-n}{2} \right) i\phi \right]_{x=0}^{2^n-1} \right\},
\end{equation}
where $H(x)$ is the Hamming weight of $x$ written in binary.

Performing the integral, we note that for $a \in \mathbb{Z} \setminus \{ 0 \}$, 
\begin{equation}
    \intCZ \exp \!\left(\frac{a}{2} i\phi \right) \,d\phi = 0.
\end{equation}
Thus, only terms satisfying $H(x) = n/2$ survive the integral. Observe then that $\Pi_{z}^{(n)} = 0$ for all odd $n$. Thus, it follows that
\begin{equation}
\Pi_{z}^{(n)} =
    \begin{cases}
        P_{k} & \text{if } n = 2k\\
        0 & \text{otherwise},
    \end{cases}
\end{equation}
where $P_k$ is defined as the projector onto the subspace of computational basis elements with Hamming weight~$k$. As an example, for $n=2$, 
\begin{equation}
    \Pi_z^{(2)} = P_1 = \outerproj{01} + \outerproj{10}.
\end{equation}

To test a symmetry of this form, we rewrite the projector in terms of unitaries. We construct a set of unitaries $U_y$ such that 
\begin{equation}
    \Pi_z^{(n)} = \frac{1}{n+1} \sum_{y=0}^{n} U_y.
\end{equation}
We use a construction similar to the form given in  \cite[Eq.~(2.59)]{tomamichel2015quantum}. Define a unitary representation $\{U_y\}_{y=0}^{n}$ as
\begin{equation}
    \label{eq:CZ-U-rep}
    U_y \coloneqq  \sum\limits_{x=0}^n \exp \left[ \frac{\pi i}{n+1}\left(2y-n\right)\left(2x-n\right) \right] P_x.
\end{equation}
Observe that $U_y^\dagger U_y = \mathbb{I}$. Furthermore, we see that 
\begin{equation}
    \sum\limits_{y=0}^n U_y = \sum\limits_{x=0}^n \sum\limits_{y=0}^n \exp \left[ \frac{\pi i}{n+1}\left(2y-n\right)\left(2x-n\right) \right] P_x.
\end{equation}
Consider that for integer $c \neq 0$,
\begin{align}
    & \sum\limits_{y=0}^n \exp \!\left( \frac{\pi i}{n+1} c \left(2y-n \right) \right) \nonumber \\
    &= \exp \!\left( \frac{-\pi i cn}{n+1} \right)\frac{1-\exp(2\pi i c)}{1 - \exp\left( \frac{2\pi i c}{n+1} \right)}   \\
    &= 0.
\label{eq:CZ-geom-series}
\end{align}
Thus, only terms satisfying $2x = n$ survive the summation. Therefore,
\begin{align}
    \frac{1}{n+1} \sum\limits_{y=0}^n U_y &= \sum\limits_{x=0}^n \delta_{2x, n} P_x \\
    &= \begin{cases}
            P_{k} & \text{if } n = 2k\\
            0 & \text{otherwise}
    \end{cases}\\
    &= \Pi_z^{(n)}.
\end{align}

Thus, testing $G$-Bose symmetry with respect to $G_z = \{R_z(\phi)^{\otimes n}\}_{\phi \in [0, 4\pi)}$ is equivalent to testing $G$-Bose symmetry with respect to $\{U_y\}_{y=0}^n$. To summarize, testing if a $n$-qubit state is $G_z$-Bose symmetric is equivalent to testing if it belongs to the subspace of Hamming weight $n = 2k$. As an aside, we note that a generalization of our method allows for performing a projection onto constant-Hamming-weight subspaces, which is useful in tasks like entanglement concentration \cite{Wbook17}. See also \cite{KM01} for alternative circuit constructions for performing measurements of Hamming weight.

In what follows, we test the symmetry for an example, with $n=2$. 
From the definition, we see that 
\begin{align}
    U_0 &= \exp \!\left( -\frac{2\pi i}{3} \right) P_0 + P_1 + \exp \!\left( \frac{2\pi i}{3}\right) P_2, \\
    U_1 &= \mathbb{I} ,\\
    U_2 &= \exp \!\left( \frac{2\pi i}{3}\right) P_0 + P_1 + \exp \!\left(-\frac{2\pi i}{3}\right) P_2 \nonumber \\
    &= U_0^2 .
\end{align}
Thus, the set of unitaries forms a unitary representation of the cyclic group $C_3$. The group table can be seen in Appendix~\ref{app:cyclic-c-3}, where $\{\ket{00} \rightarrow U_1, \ket{01} \rightarrow U_0, \ket{11} \rightarrow U_2\}$.  Expanding terms, we see that 
\begin{equation}
\label{eq:CZ-U_0}
    U_0 = \left(R_z\!\left(\frac{2\pi}{3}\right)\right)^{\otimes 2}.
\end{equation}
Furthermore, since $U_2 = U^2_0$, 
\begin{equation}
    U_2 = \left(R_z\!\left(-\frac{2\pi}{3}\right)\right)^{\otimes 2}.
\end{equation}
Since we have three elements, the $\ket{+}_C$ state is a uniform superposition of three elements. We use two qubits and the unitary $U_3$ used to generate the following superposition, as shown in Figure~\ref{fig:U3_Superposition}:
\begin{equation}
    \label{eq:U3_Superposition}
    U_3\ket{00} = \frac{1}{\sqrt{3}} (\ket{00} + \ket{01} + \ket{11}).
\end{equation}

\begin{figure}
\begin{center}
\includegraphics[width=0.65\linewidth]{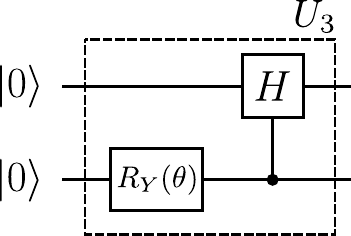}
\end{center}
\caption{Unitary $U_3$, with $\theta = 2 \arctan \left(\sqrt{2}\right)$, generates the equal superposition of three elements from \eqref{eq:U3_Superposition}.}
\label{fig:U3_Superposition}
\end{figure}

Figure~\ref{fig:CZ_Circuits}a) depicts the circuit that tests for $G$-Bose symmetry. Table~\ref{tab:CollectiveZ_GBS} shows the results for various input states. The true fidelity value is calculated using \eqref{eq:acc-prob-bose-test}, where $\Pi^G_S$ is defined in \eqref{eq:group_proj_GBS}.

\begin{table}[h]
\centering
\vspace{.05in}
\begin{tabular}{
>{\centering\arraybackslash}p{0.12\textwidth} | 
>{\centering\arraybackslash}p{0.08\textwidth} | 
>{\centering\arraybackslash}p{0.08\textwidth} | 
>{\centering\arraybackslash}p{0.08\textwidth} 
}
\hline
\textrm{State} & 
\textrm{True Fidelity} &
\textrm{Noiseless} &
\textrm{Noisy}\\
\hline\hline 
$\outerproj{00}$ & 0.0 & 0.0000 & 0.0220 \\
$\rho$ & 1.0 & 1.0000 & 0.9170\\
$\outerproj{0} \otimes \outerproj{+}$ & 0.5 & 0.5000 & 0.4877\\ 
$\pi^{\otimes 2} $ & 0.5 & 0.5000 & 0.4661\\ 
\hline
\end{tabular}
\caption{Results of collective-phase-Bose symmetry tests. The state $\rho$ is defined as $\outerproj{\psi}$ where $\ket{\psi} = \frac{1}{\sqrt{2}}(\ket{01} + \ket{10})$.}
\label{tab:CollectiveZ_GBS}
\end{table}

\begin{figure*}
\begin{center}
\includegraphics[width=\linewidth]{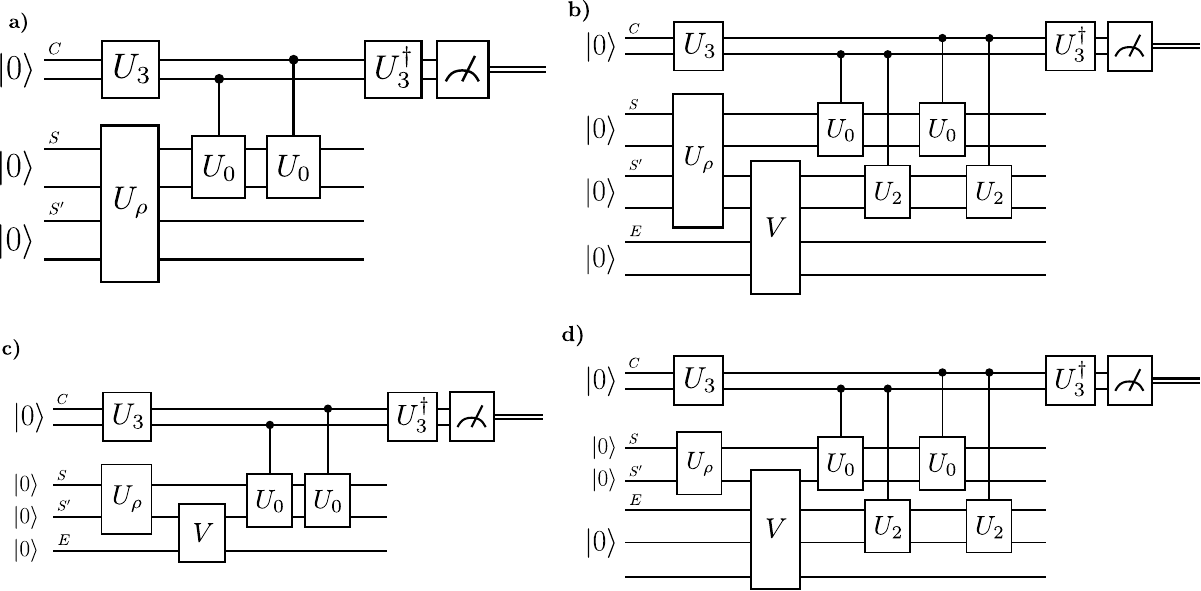}
\end{center}
\caption{Symmetry tests for the collective phase group: a) $G$-Bose symmetry, b) $G$-symmetry, c) $G$-Bose symmetric extendibility, and d) G-symmetric extendibility. The unitary $U_0$ is defined in \eqref{eq:CZ-U_0}. Note that $U_2 = U_0^\dagger$.}
\label{fig:CZ_Circuits}
\end{figure*}

\subsubsection{\texorpdfstring{$G$}{G}-symmetry}

A state that is $G_z$-symmetric satisfies the following condition:
\begin{equation}
    \rho = \mathcal{C}_z^{(n)}(\rho) ,
\end{equation}
where the collective dephasing channel $\mathcal{C}_z^{(n)}$ is defined as
\begin{equation}
    \mathcal{C}_z^{(n)}(\omega) \coloneqq \frac{1}{4\pi} \int_{0}^{4\pi} d\phi \ R_z(\phi)^{\otimes n} \omega R^\dagger_z(\phi)^{\otimes n}.
\end{equation}
Using the fact that
\begin{equation}
    R_z(\phi) = \operatorname{exp}(-i\phi \sigma_z / 2 ),
\end{equation}
we see that 
\begin{equation}
    R_z(\phi) \ket{a}\!\bra{b} R^\dagger_z(\phi) = e^{i\phi(a-b)} \ket{a}\!\bra{b}, 
\end{equation}
for $a, b \in \{0, 1\}$. Thus, for a general $n$-qubit state $\rho$, expanded in the computational basis as
\begin{equation}
    \rho = \sum_{x_1, \ldots, x_n, y_1, \ldots, y_n} \rho_{x_1, \ldots, x_n, y_1, \ldots, y_n} \ket{x_1\cdots x_n}\!\bra{y_1\cdots y_n},
\end{equation}
it follows that 
\begin{multline}
    \mathcal{C}_z^{(n)}(\rho) =  \sum_{x_1, \ldots, x_n, y_1, \ldots, y_n}
    \delta\!\left(\sum_i x_i , \sum_j y_j \right) \times \\ 
    \rho_{x_1,\ldots, x_n, y_1,\ldots, y_n} \ket{x_1\cdots x_n}\!\bra{y_1\cdots y_n}.
\end{multline}
Since $\sum_i x_i = H(x)$, it follows that
\begin{equation}
    \mathcal{C}_z^{(n)}(\rho) = \sum\limits_{k=0}^{n} P_k \rho P_k,
\end{equation}
where, as before, $P_k$ is the projector onto the subspace of Hamming weight $k$. For the case of $n=2$, we get the following projectors
\begin{align}
    P_0 &= \outerprod{00}{00} , \\
    P_1 &= \outerprod{01}{01} + \outerprod{10}{10} , \\
    P_2 &= \outerprod{11}{11}.
\end{align}

To test a symmetry of this form, we can rewrite the channel in terms of a set  $\{U_y\}_y$ of unitaries satisfying
\begin{equation}
    \mathcal{C}_z^{(n)}(\rho) = \frac{1}{n+1} \sum\limits_{y=0}^{n} U_y \rho U^\dagger_y.
\end{equation}
We now prove that the unitaries $\{U_y\}_{y=0}^n$ from \eqref{eq:CZ-U-rep} satisfy this condition:
\begin{align}
    &\frac{1}{n+1} \sum\limits_{y=0}^n U_y \rho U^\dagger_y \nonumber \\
    &= \frac{1}{n+1} \sum\limits_{\substack{x, x', \\ y=0}}^n \exp \!\left[ \frac{\pi i}{n+1}\left(2y-n\right)2 \left(x-x'\right) \right] P_x \rho P_{x'} \nonumber \\
    &= \frac{1}{n+1} \sum\limits_{x, x'=0}^n (n+1) \delta_{x, x'} P_x \rho P_{x'}  \\
    &= \sum\limits_{x=0}^{n} P_x \rho P_x,
\end{align}
where the third equality follows from the reasoning in~\eqref{eq:CZ-geom-series}.

Thus, similar to the $G$-Bose symmetry tests, testing $G$-symmetry with respect to $G_z = \{R_z(\phi)^{\otimes n}\}_{\phi \in [0, 4\pi)}$ is equivalent to testing $G$-symmetry with respect to $\{U_y\}_{y=0}^n$. To summarize, testing if an $n$-qubit state is $G_z$-symmetric is equivalent to testing if it belongs to a subspace of fixed Hamming weight. In this work, we test the symmetry for $n=2$. 

A circuit that tests for $G$-symmetry is shown in Figure~\ref{fig:CZ_Circuits}b). It involves variational parameters, and an example of the training process is shown in Figure~\ref{fig:CollectiveZ_GS_Training}. Table~\ref{tab:CollectiveZ_GS} shows the final results after training for various input states. The true fidelity is calculated using the semi-definite program given in \eqref{eq:SDP-rootfid-GS}.

\begin{figure}
\begin{center}
\includegraphics[width=\linewidth]{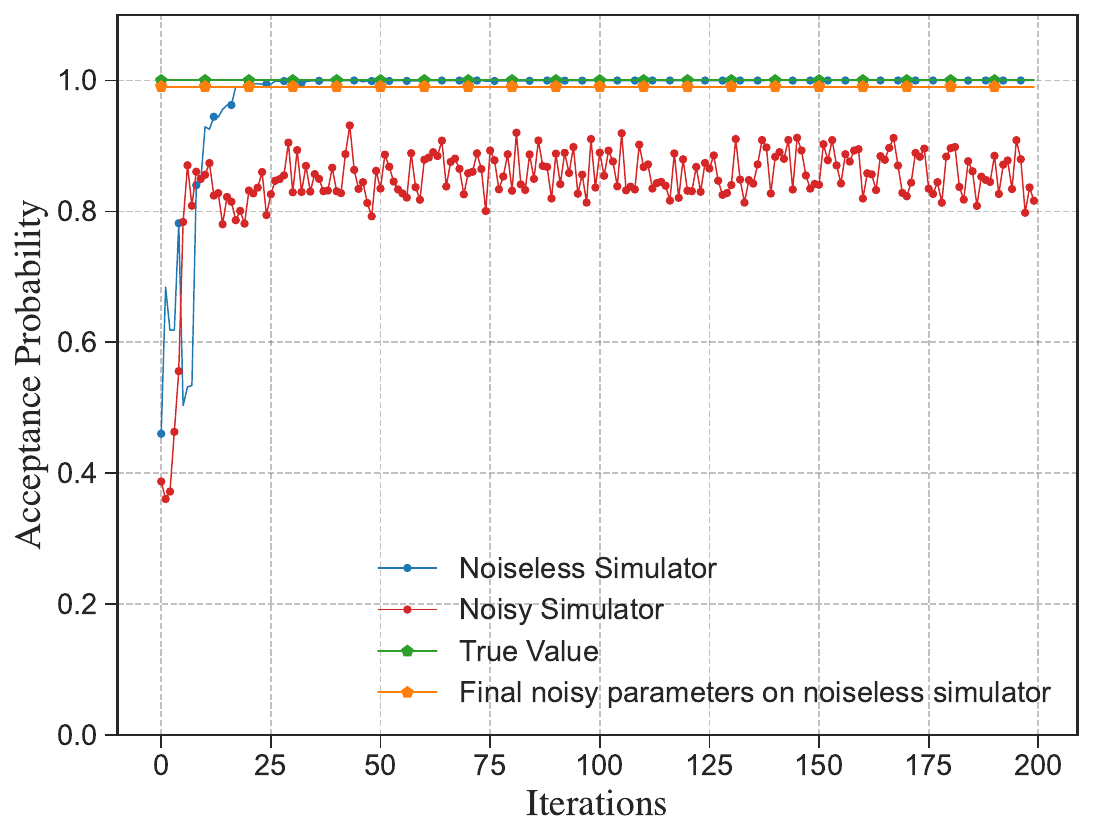}
\end{center}
\caption{Example of the training process for testing collective-phase-symmetry of $\rho = \outerproj{\Psi^+}$, where $\ket{\Psi^+} = \frac{1}{\sqrt{2}}(\ket{01} + \ket{10})$.}
\label{fig:CollectiveZ_GS_Training}
\end{figure}

\begin{table}[h]
\centering
\vspace{.05in}
\begin{tabular}{
>{\centering\arraybackslash}p{0.06\textwidth} | 
>{\centering\arraybackslash}p{0.07\textwidth} | 
>{\centering\arraybackslash}p{0.07\textwidth} | 
>{\centering\arraybackslash}p{0.06\textwidth} |
>{\centering\arraybackslash}p{0.08\textwidth}}
\hline
\textrm{State} & 
\textrm{True Fidelity} &
\textrm{Noiseless} &
\textrm{Noisy} & 
\textrm{Noise Resilient}
\\
\hline\hline 
$\outerproj{00}$ & 1.0000 & 0.9999 & 0.8380 & 0.9928 \\
$\rho$ & 1.0000 & 1.0000 & 0.8162 & 0.9906 \\
$\tau$ & 0.5001 & 0.5000 & 0.4630 & 0.4990 \\
$\pi^{\otimes 2}$ & 1.0000 & 0.9998 & 0.8417 & 0.9934 \\  
\hline
\end{tabular}
\caption{Results of collective-phase-symmetry tests. The state $\rho$ is defined as $\outerproj{\Psi^+}$ where $\ket{\Psi^+} = \frac{1}{\sqrt{2}}(\ket{01} + \ket{10})$. The state $\tau$ is defined as $\outerproj{\Phi^+}$ where $\ket{\Phi^+} = \frac{1}{\sqrt{2}}(\ket{00}) + \ket{11})$.}
\label{tab:CollectiveZ_GS}
\end{table}

\subsubsection{\texorpdfstring{$G$}{G}-Bose symmetric extendibility}

A circuit that tests for $G$-Bose symmetric extendibility is shown in Figure~\ref{fig:CZ_Circuits}c). It involves variational parameters, and an example of the training process is shown in Figure~\ref{fig:CollectiveZ_GBSE_Training}. Table~\ref{tab:CollectiveZ_GBSE} shows the final results after training for various input states. The true fidelity is calculated using the semi-definite program given in \eqref{eq:SDP-rootfid-GBSE}.

\begin{figure}
\begin{center}
\includegraphics[width=\linewidth]{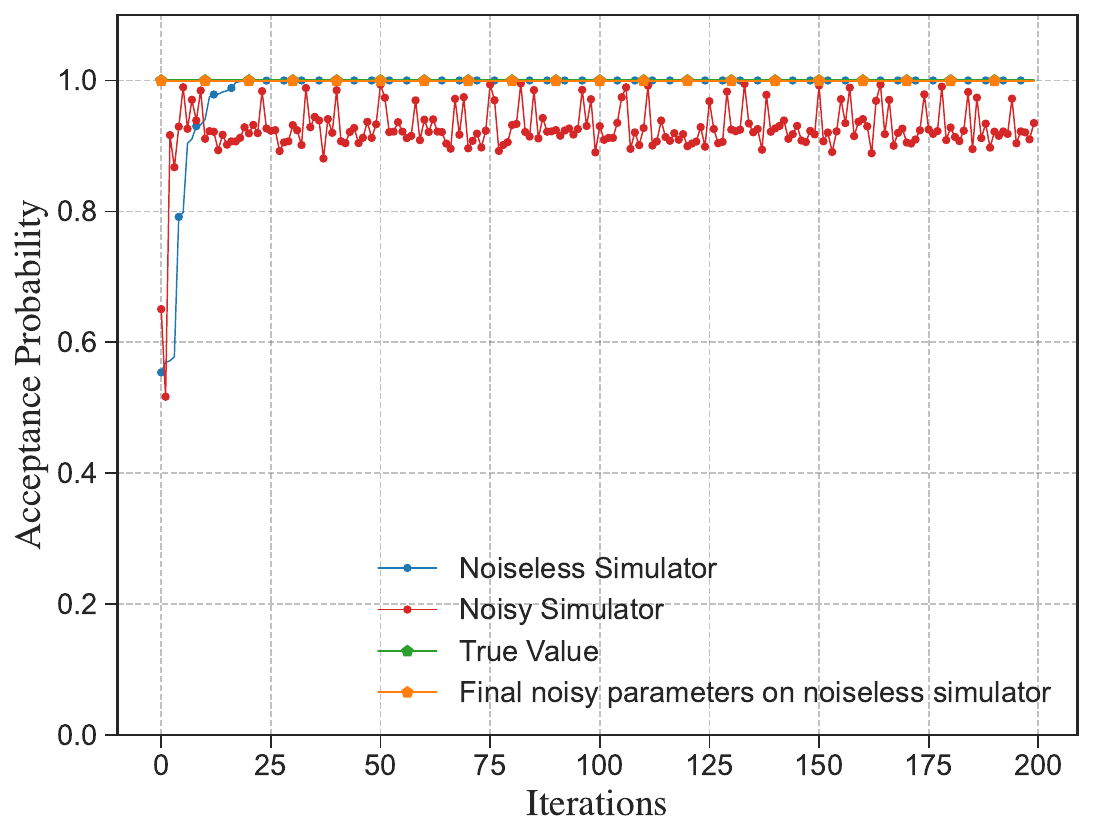}
\end{center}
\caption{Example of the training process for testing collective-phase-Bose symmetric extendibility of $\sfrac{3}{4}\outerproj{0} + \sfrac{1}{4} \outerproj{1}$. We see that the training exhibits a noise resilience.}
\label{fig:CollectiveZ_GBSE_Training}
\end{figure}

\begin{table}[h]
\centering
\vspace{.05in}
\begin{tabular}{
>{\centering\arraybackslash}p{0.05\textwidth} | 
>{\centering\arraybackslash}p{0.07\textwidth} | 
>{\centering\arraybackslash}p{0.07\textwidth} | 
>{\centering\arraybackslash}p{0.07\textwidth} |
>{\centering\arraybackslash}p{0.08\textwidth}}
\hline
\textrm{State} & 
\textrm{True Fidelity} &
\textrm{Noiseless} &
\textrm{Noisy} & 
\textrm{Noise Resilient}
\\
\hline\hline 
$\outerproj{0}$ & 1.0000 & 1.0000 & 0.9783 & 0.9980 \\
$\sigma$ & 1.0000 & 1.0000 & 0.9349 & 0.9993 \\
$\outerproj{-}$ & 0.5002 & 0.5000 & 0.4464 & 0.5000 \\
$\rho$ & 0.9330 & 0.9330 & 0.9208 & 0.9328 \\
\hline
\end{tabular}
\caption{Results of collective-phase-Bose symmetric extendibility tests. The state $\sigma$ is defined as $\sfrac{3}{4}\outerproj{0} + \sfrac{1}{4} \outerproj{1}$. The state $\rho$ is defined as  $\begin{bmatrix} 0.93 & 0.25\\ 0.25 & 0.07 \end{bmatrix}$.}
\label{tab:CollectiveZ_GBSE}
\end{table}
 
\begin{table}[h]
\centering
\vspace{.05in}
\begin{tabular}{
>{\centering\arraybackslash}p{0.05\textwidth} | 
>{\centering\arraybackslash}p{0.07\textwidth} | 
>{\centering\arraybackslash}p{0.07\textwidth} | 
>{\centering\arraybackslash}p{0.07\textwidth} |
>{\centering\arraybackslash}p{0.08\textwidth}}
\hline
\textrm{State} & 
\textrm{True Fidelity} &
\textrm{Noiseless} &
\textrm{Noisy} & 
\textrm{Noise Resilient}
\\
\hline\hline 
$\outerproj{0}$ & 1.0000 & 0.9960 & 0.8632 & 0.9988 \\
$\outerproj{+}$ & 0.5000 & 0.5000 & 0.4580 & 0.4997  \\
$\rho$ & 0.7500 & 0.7494 & 0.6577 & 0.7484\\
\hline
\end{tabular}
\caption{Results of collective-phase-symmetric extendibility tests. The state $\rho$ is defined as $\begin{bmatrix} 0.75 & 0.43\\ 0.43 & 0.25 \end{bmatrix}$.}
\label{tab:CollectiveZ_GSE}
\end{table}

\subsubsection{\texorpdfstring{$G$}{G}-symmetric extendibility}

A circuit that tests for $G$-symmetric extendibility is shown in Figure~\ref{fig:CZ_Circuits}d). It involves variational parameters, and an example of the training process is shown in Figure~\ref{fig:CollectiveZ_GSE_Training}. Table~\ref{tab:CollectiveZ_GSE} shows the final results after training for various input states. The true fidelity is calculated using the semi-definite program given in \eqref{eq:SDP-rootfid-GSE}.

\begin{figure}[h!]
\begin{center}
\includegraphics[width=\linewidth]{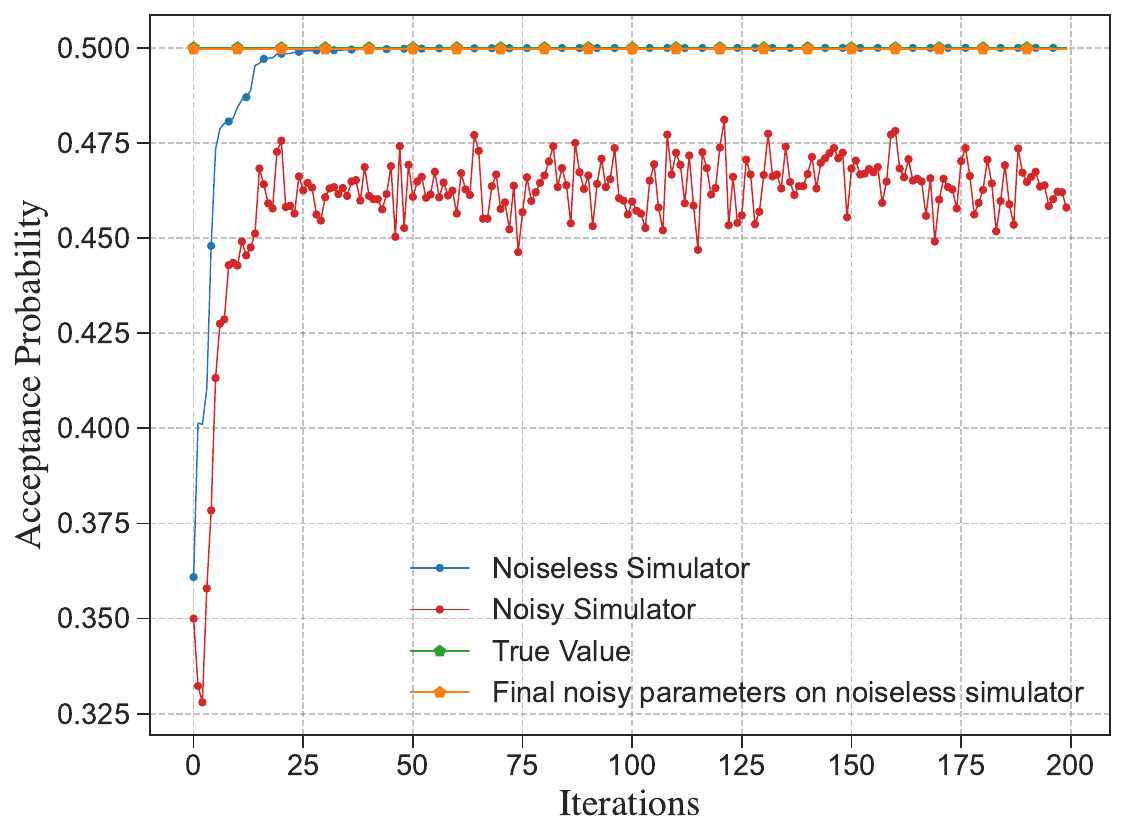}
\end{center}
\caption{Example of the training process for testing collective-phase-symmetric extendibility of $\outerproj{+}$.}
\label{fig:CollectiveZ_GSE_Training}
\end{figure}

\subsection{\texorpdfstring{$k$}{k}-Extendibility and \texorpdfstring{$k$}{k}-Bose extendibility}

As seen in Examples~\ref{ex:k-ext} and \ref{ex:k-bose-ext}, $k$-extendibility and $k$-Bose extendibility are special cases of $G$-symmetric extendibility and $G$-Bose symmetric extendibility, respectively. In this section, we look at the cases of two and three extending subsystems. 

As seen in \eqref{eq:ident-k-to-g-1}--\eqref{eq:ident-k-to-g-4}, $U_{RS}(g) = \mathbb{I}_A \otimes W_{B_1 \cdots B_k}(\pi)$, where $W_{B_1 \cdots B_k}(\pi)$ is a unitary representation of the symmetric group $S_k$. Thus, given a unitary representation of $S_k$, we can test for the required symmetries.

The $S_2$ group has two elements, and the group table is given by
\begin{center}
\begin{tabular}{
>{\centering\arraybackslash}p{0.14\textwidth} | >{\centering\arraybackslash}p{0.03\textwidth} | 
>{\centering\arraybackslash}p{0.03\textwidth}}
\hline
    Group element & $e$ & $a$ \\ 
\hline
    $e$ & $e$ & $a$ \\
    $a$ & $a$ & $e$ \\
\hline
\end{tabular}
\end{center}

The standard representation of $S_2$ translates easily to a two-qubit unitary representation with $\{e \rightarrow \mathbb{I}, a \rightarrow F\}$, where $F$ is the SWAP gate. In fact, throughout this section, we will consider unitary representations corresponding to system permutations in a direct correspondence with the standard representations of $S_k$. Using this definition, let $U_{RS}(e) = \mathbb{I}_{A} \otimes \mathbb{I}_{B_1B_2}$ and $U_{RS}(a) = \mathbb{I}_{A} \otimes F_{B_1B_2}$. Since we have two elements, the $\ket{+}_C$ state is a uniform superposition of two elements. We thus use one qubit and the Hadamard gate to generate the necessary state:
\begin{equation}
    H\ket{0} = \frac{1}{\sqrt{2}}\left(\ket{0} + \ket{1}\right).
\end{equation}
The control register states need to be mapped to group elements; for this, we employ the mapping $\{\ket{0} \rightarrow e, \ket{1} \rightarrow a\}$ for our circuit constructions. 

Similarly, the $S_3$ group has six elements and the group table is given by
\begin{center}
\begin{tabular}{
>{\centering\arraybackslash}p{0.14\textwidth} | >{\centering\arraybackslash}p{0.02\textwidth} |
>{\centering\arraybackslash}p{0.02\textwidth} |
>{\centering\arraybackslash}p{0.02\textwidth} |
>{\centering\arraybackslash}p{0.02\textwidth} |
>{\centering\arraybackslash}p{0.02\textwidth} |
>{\centering\arraybackslash}p{0.02\textwidth}}
\hline
    Group element & $e$ & $a$ & $b$ & $c$ & $d$& $f$\\ 
\hline
   $e$ & $e$ & $a$ & $b$ & $c$ & $d$& $f$\\
    $a$ & $a$ & $e$ & $d$& $f$ & $b$ & $c$\\
    $b$ & $b$ & $f$ & $e$ & $d$& $c$ & $a$\\
    $c$ & $c$ & $d$& $f$ & $e$ & $a$ & $b$\\
    $d$ & $d$& $c$ & $a$ & $b$ & $f$ & $e$ \\
    $f$ & $f$ & $b$ & $c$ & $a$ & $e$ & $d$\\
    
\hline
\end{tabular}
\end{center}

The $S_3$ group has a three-qubit unitary representation $\{e \rightarrow \mathbb{I}, a \rightarrow F_{23}, b \rightarrow F_{13}, c \rightarrow F_{12}, d \rightarrow F_{12}F_{23}, f \rightarrow F_{13}F_{23}\}$, where $F_{ij}$ is the SWAP gate between qubits $i$ and $j$. Since we have six elements, the $\ket{+}_C$ state is a uniform superposition of six elements. We use three qubits and the same unitary $U_d$ used to generate the superposition for the triangular dihedral group, as shown in Figure~\ref{fig:Dihedral_Superposition}, to generate an equal superposition of six elements,
\begin{multline}
    \label{eq:S3_Superposition}
    U_d\ket{000} = \frac{1}{\sqrt{6}} (\ket{000} + \ket{001} + \ket{010} + \\
    \ket{011} + \ket{100} + \ket{101}).
\end{multline}
The control register states need to be mapped to group elements, and we do so via the mapping $\{\ket{000} \rightarrow e, \ket{001} \rightarrow a, \ket{010} \rightarrow b, \ket{011} \rightarrow f, \ket{100} \rightarrow c, \ket{101} \rightarrow d\}$.

\begin{figure*}[t!]
\begin{center}
\includegraphics[width=\linewidth]{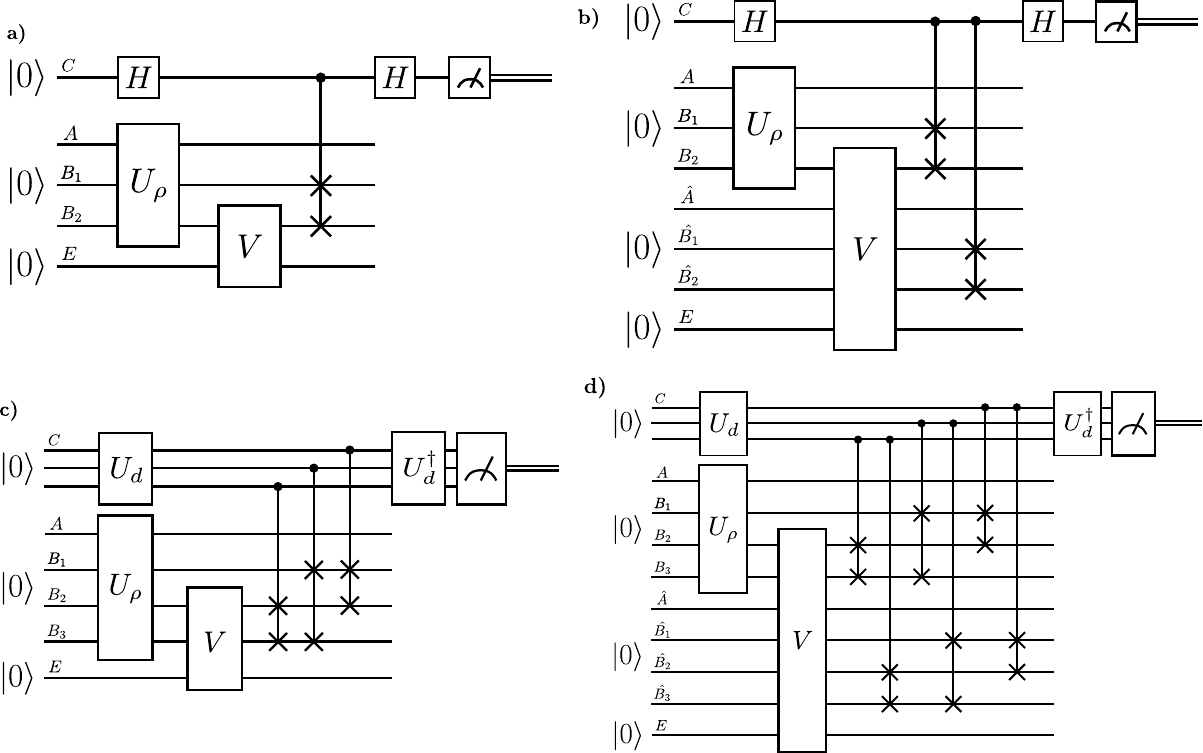}
\end{center}
\caption{Tests for extendibility: a) two-Bose extendibility, b) two-extendibility, c) three-Bose extendibility, and d) three-extendibility.}
\label{fig:Ext_Circuits}
\end{figure*}

\subsubsection{Two-Bose extendibility}

A circuit that tests for two-Bose extendibility is shown in Figure~\ref{fig:Ext_Circuits}a). It involves variational parameters, and an example of the training process is shown in Figure~\ref{fig:S2_GBSE_Training}. Table~\ref{tab:S2_GBSE} shows the final results after training for various input states. The true fidelity is calculated using the semi-definite program given in \eqref{eq:SDP-rootfid-GBSE}.

\begin{figure}[h!]
\begin{center}
\includegraphics[width=\linewidth]{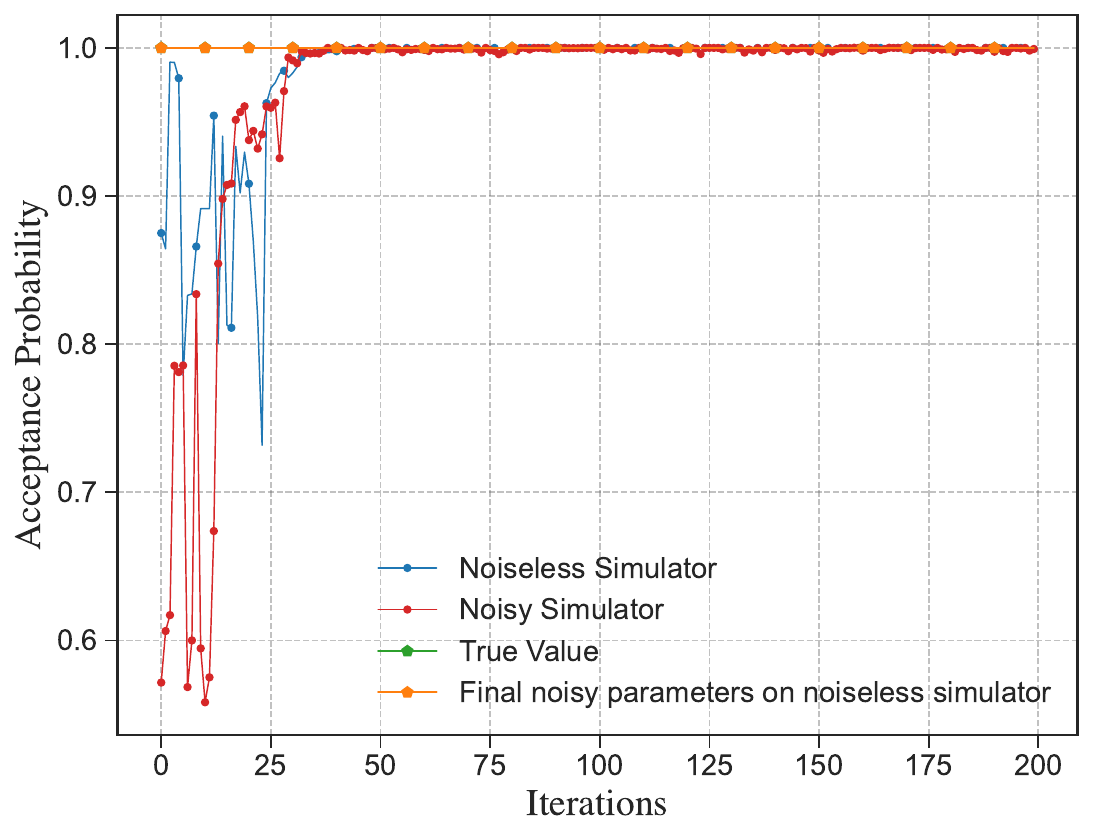}
\end{center}
\caption{Example of the training process for testing two-Bose extendibility of $\rho = \sfrac{3}{4}\outerprod{00}{00} + \sfrac{1}{4}\outerprod{11}{11}$. We see that the training exhibits a noise resilience.}
\label{fig:S2_GBSE_Training}
\end{figure}

\begin{table}[h]
\centering
\vspace{.05in}
\begin{tabular}{
>{\centering\arraybackslash}p{0.06\textwidth} | >{\centering\arraybackslash}p{0.07\textwidth} | 
>{\centering\arraybackslash}p{0.07\textwidth} | 
>{\centering\arraybackslash}p{0.05\textwidth} |
>{\centering\arraybackslash}p{0.08\textwidth}}
\hline
\textrm{State} & 
\textrm{True Fidelity} &
\textrm{Noiseless} &
\textrm{Noisy} & 
\textrm{Noise Resilient}
\\
\hline\hline 
$\outerproj{00}$ & 1.0000 & 1.0000 & 0.9544 & 0.9995 \\
$\rho$ & 1.0000 & 1.0000 & 0.9584 & 0.9995\\
$\Psi^+$ & 0.7500 & 0.7500 & 0.7256 & 0.7500 \\
\hline
\end{tabular}
\caption{Results of $S_2$-Bose symmetric extendibility tests. The state $\rho$ is defined as $\sfrac{3}{4}\outerprod{00}{00} + \sfrac{1}{4}\outerprod{11}{11}$.}
\label{tab:S2_GBSE}
\end{table}

\subsubsection{Two-Extendibility}

Similar to the non-extended cases, it is simpler to test if a state exhibits $G$-BSE---or, in this case, if the state is $k$-Bose-symmetric extendible---than to test if it is symmetric extendible. This is reflected in Figure~\ref{fig:Ext_Circuits}b), which shows a test for 2-BSE. The circuit involves variational parameters, and an example of the training process is shown in Figure~\ref{fig:S2_GSE_Training}. Table~\ref{tab:S2_GSE} shows the final results after training for various input states. The true fidelity is calculated using the semi-definite program given in~\eqref{eq:SDP-rootfid-GSE}.

\begin{figure}[h]
\begin{center}
\includegraphics[width=\linewidth]{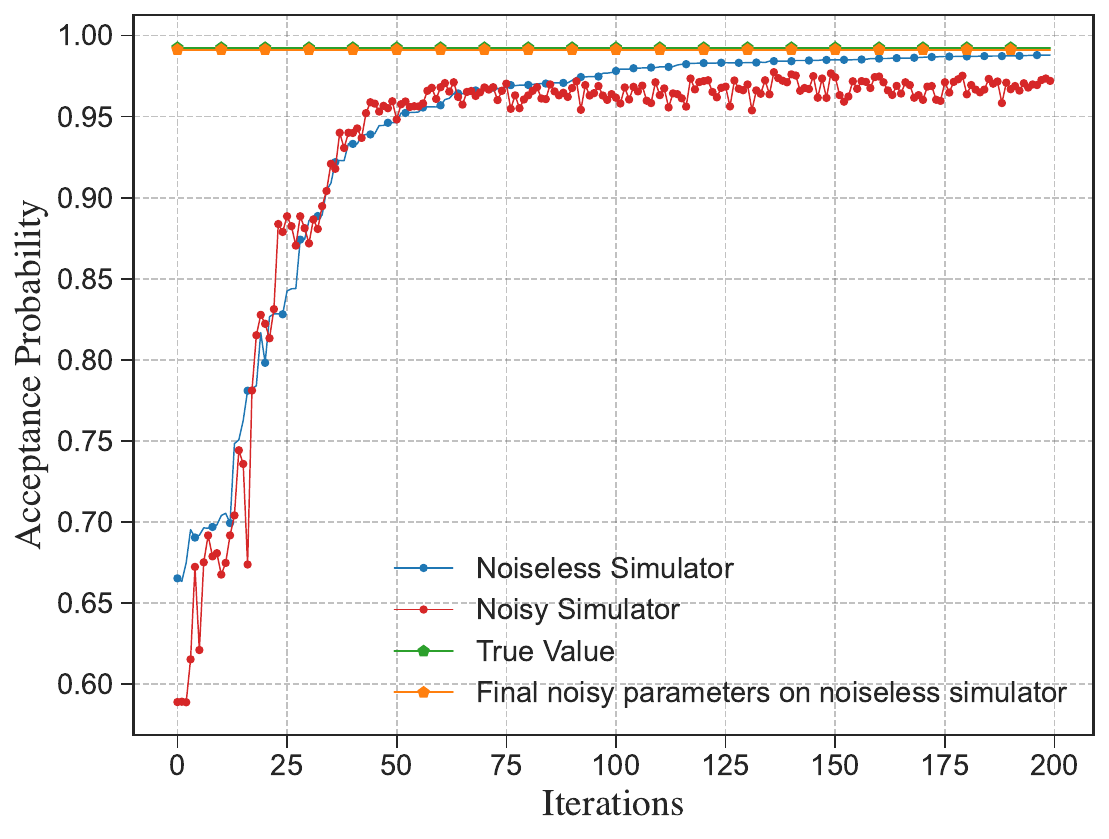}
\end{center}
\caption{Example of the training process for testing two-extendibility of $\rho = \outerproj{\psi}$, where $\ket{\psi} = \frac{1}{\sqrt{2}}\ket{11} + \frac{1}{\sqrt{6}}(\ket{00} + \ket{01} + \ket{10})$. We see that the training exhibits a noise resilience.}
\label{fig:S2_GSE_Training}
\end{figure}

\begin{table}[h]
\centering
\vspace{.05in}
\begin{tabular}{
>{\centering\arraybackslash}p{0.06\textwidth} | 
>{\centering\arraybackslash}p{0.07\textwidth} | 
>{\centering\arraybackslash}p{0.07\textwidth} | 
>{\centering\arraybackslash}p{0.06\textwidth} |
>{\centering\arraybackslash}p{0.08\textwidth}}
\hline
\textrm{State} & 
\textrm{True Fidelity} &
\textrm{Noiseless} &
\textrm{Noisy} & 
\textrm{Noise Resilient}
\\
\hline\hline 
$\outerproj{00}$ & 1.0000 & 0.9991 & 0.9267 & 0.9960 \\
$\rho$ & 0.9925 & 0.9901 & 0.9720 & 0.9913 \\
$\Psi^+$ & 0.7506 & 0.7498 & 0.6959 & 0.7480 \\
\hline
\end{tabular}
\caption{Results of $S_2$-symmetric extendibility tests. The state $\rho$ is defined as $\outerproj{\psi}$ where $\ket{\psi} = \frac{1}{\sqrt{2}}\ket{11} + \frac{1}{\sqrt{6}}(\ket{00} + \ket{01} + \ket{10})$. The reduced state of $\rho$ has eigenvalues $\frac{1}{6}\left(3+\sqrt{5+2\sqrt{3}}\right)\approx 0.985$ and $\frac{1}{6}\left(3-\sqrt{5+2\sqrt{3}}\right)\approx 0.015$. It is thus not so entangled, and we expect its two-extendible fidelity to be close to one.}
\label{tab:S2_GSE}
\end{table}

\subsubsection{Three-Bose Extendibility}

A circuit that tests for three-Bose extendibility is shown in Figure~\ref{fig:Ext_Circuits}c). It involves variational parameters, and an example of the training process is shown in Figure~\ref{fig:S3_GBSE_Training}. Table~\ref{tab:S3_GBSE} shows the final results after training for various input states. The true fidelity is calculated using the semi-definite program given in~\eqref{eq:SDP-rootfid-GBSE}.

\begin{figure}[h!]
\begin{center}
\includegraphics[width=\linewidth]{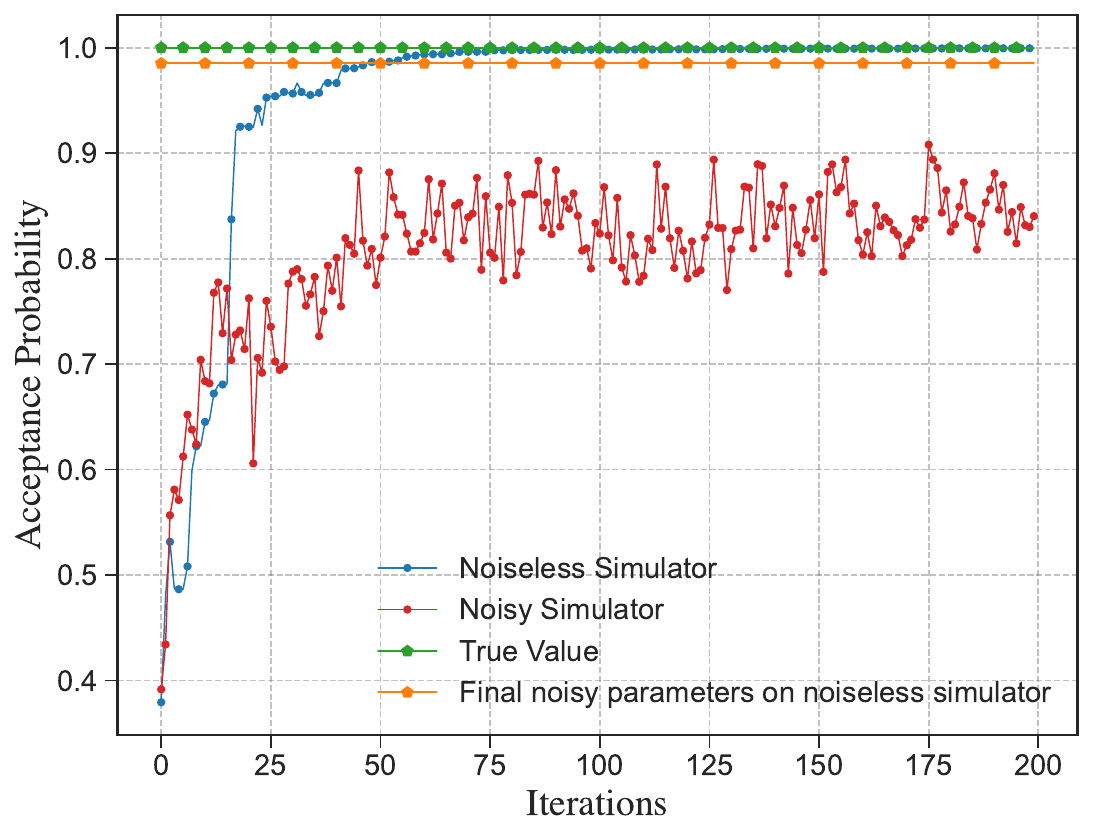}
\end{center}
\caption{Example of the training process for testing three-Bose extendibility of $\rho = \sfrac{3}{4}\outerprod{00}{00} + \sfrac{1}{4}\outerprod{11}{11}$. We see that the training exhibits a noise resilience.}
\label{fig:S3_GBSE_Training}
\end{figure}

\begin{table}[h]
\centering
\vspace{.05in}
\begin{tabular}{
>{\centering\arraybackslash}p{0.06\textwidth} | >{\centering\arraybackslash}p{0.07\textwidth} | 
>{\centering\arraybackslash}p{0.08\textwidth} | 
>{\centering\arraybackslash}p{0.06\textwidth} |
>{\centering\arraybackslash}p{0.08\textwidth}}
\hline
\textrm{State} & 
\textrm{True Fidelity} &
\textrm{Noiseless} &
\textrm{Noisy} & 
\textrm{Noise Resilient}
\\
\hline\hline 
$\outerproj{00}$ & 1.0000 & 0.9999 & 0.8644 & 0.9982 \\
$\rho$ & 1.0000 & 0.9994 & 0.8403 & 0.9851 \\
$\Psi^+$ & 0.6675 & 0.6667 & 0.5666 & 0.6666 \\
\hline
\end{tabular}
\caption{Results of $S_3$-Bose symmetric extendibility tests. The state $\rho$ is defined as $\sfrac{3}{4}\outerprod{00}{00} + \sfrac{1}{4}\outerprod{11}{11}$.}
\label{tab:S3_GBSE}
\end{table}

\subsubsection{Three-Extendibility}

A circuit that tests for three-extendibility is shown in Figure~\ref{fig:Ext_Circuits}d). It involves variational parameters, and an example of the training process is shown in Figure~\ref{fig:S3_GSE_Training}. Table~\ref{tab:S3_GSE} shows the final results after training for various input states. The true fidelity is calculated using the semi-definite program given in \eqref{eq:SDP-rootfid-GSE}.

\begin{figure}
\begin{center}
\includegraphics[width=\linewidth]{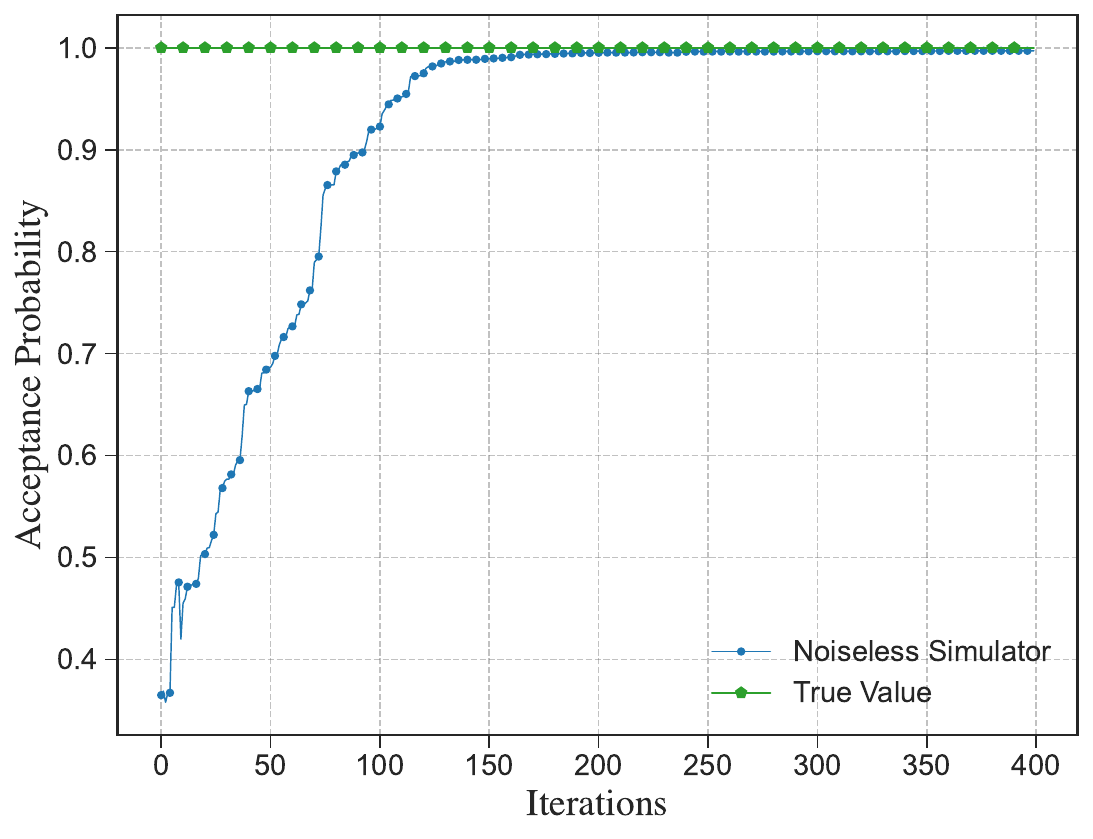}
\end{center}
\caption{Example of the training process for testing three-extendibility of $\outerproj{00}$.}
\label{fig:S3_GSE_Training}
\end{figure}

\begin{table}
\centering
\vspace{.05in}
\begin{tabular}{
>{\centering\arraybackslash}p{0.08\textwidth} | >{\centering\arraybackslash}p{0.08\textwidth} | 
>{\centering\arraybackslash}p{0.08\textwidth}}
\hline
\textrm{State} & 
\textrm{True Fidelity} &
\textrm{Noiseless}
\\
\hline\hline 
$\outerproj{00}$ & 1.0000 & 0.9970 \\
$\rho$ & 1.0000 & 0.9988 \\
$\Psi^+$ & 0.6670 & 0.6650 \\
\hline
\end{tabular}
\caption{Results of $S_3$-symmetric extendibility tests. Here, $\rho = \sfrac{3}{4}\outerprod{00}{00} + \sfrac{1}{4}\outerprod{11}{11}$.}
\label{tab:S3_GSE}
\end{table}

For all of the above cases, we see that results achieved via parameterized circuit substitutions for the prover demonstrate noise resilience, and thus give some confidence for practical applications. In this final case, we have shown explicitly how our algorithm allows for tests of $k$-extendibility and related quantities. While only small systems are considered here, this is a limitation of current hardware more so than of the algorithm itself. Indeed, it would be interesting to observe the performance of this algorithm on higher fidelity machines with more qubits, which could possibly be achievable in the near future.

\section{Resource theories}

\label{sec:res-theories}

In this section, we prove that the various maximum symmetric fidelities proposed in Section~\ref{sec:tests-o-sym} are proper resource-theoretic monotones, in the sense reviewed in~\cite{CG18}. Thus, they are indeed measures of symmetry as claimed.

To begin with, let us recall the basics of a resource theory (see \cite[Definition~1]{CG18}). Intuitively, one can delineate a resource theory by first specifying a restricted set of free channels, which are understood as allowed operations. In a resource theory, one of the basic questions is to determine whether it is possible to transition from a source state to a target state by means of only these free channels. Furthermore, once the set of free channels is fixed, the free states are also set, because a free state can be understood as a particular kind of free channel in which the input system to the channel is a trivial system.

More formally, let $\mathcal{F}$ be a mapping that assigns a unique set of quantum channels to any arbitrary input and output systems $A$ and $B$, respectively. We require that $\mathcal{F}$ include the identity channel ($\mathcal{F}(A\rightarrow A) = \operatorname{id}_A$) and  that, for any three physical systems $A$, $B$, and $C$, any two maps $\mathcal{N}_{A\rightarrow B} \in \mathcal{F}(A \rightarrow B)$ and $\mathcal{M}_{B\rightarrow C}\in\mathcal{F}(B \rightarrow C)$ have the transitive property
\begin{equation}\label{eq:trans-prop}
 \mathcal{M}_{B\rightarrow C}\circ\mathcal{N}_{A\rightarrow B} \in\mathcal{F}(A\rightarrow C)\, .
\end{equation}
If $\mathcal{F}$ obeys above criteria, then the mapping $\mathcal{F}$ defines the resource theory. 
The set $\mathcal{F}(\mathbb{C}\rightarrow A)$ defines the set of free states---that is, channels from the trivial space ($\mathbb{C}$) to system $A$ are quantum states. The set $\mathcal{F}(A \rightarrow B)$ defines the set of free channels from system $A$ to system~$B$. 

\subsection{Resource theory of asymmetry}
\label{sec:rt-asym}
The resource theory of asymmetry is well established by now \cite{MS13}, but to the best of our knowledge, the resource theory of Bose asymmetry has not been defined yet. Let us begin by recalling the resource theory of asymmetry. Afterwards, we establish the resource theory of Bose asymmetry as well as two other generalizations involving unextendibility, which are in turn generalizations of the resource theory of unextendibility proposed in \cite{KDWW19,KDWW21}.

Let $G$ be a group, and let $\{U_{A}(g)\}_{g\in G}$ and $\{V_{B}(g)\}_{g\in G}$ denote projective unitary representations of $G$. A channel $\mathcal{N}
_{A\rightarrow B}$ is a free channel in the resource theory of asymmetry if the following $G$-covariance symmetry condition holds
\begin{equation}
\label{eq:G-cov-channel-def}
\mathcal{N}_{A\rightarrow  B} \circ \mathcal{U}_{A}(g)=\mathcal{V}_{B} (g) \circ\mathcal{N}_{A\rightarrow B}\qquad\forall g\in G,
\end{equation}
where the unitary channels $\mathcal{U}_{A}(g)$ and $\mathcal{V}_{B}(g)$ are respectively defined from $U_{A}(g)$ and $V_{B}(g)$ as in \eqref{eq:uandv}. It then follows that a state $\sigma_{A}$ is free in this resource theory if it is $G$-symmetric, i.e., 
\begin{equation}
\sigma_{A}=\mathcal{U}_{A}(g)(\sigma_{A})\qquad\forall g\in G,
\end{equation}
with a similar definition for the $B$ system; furthermore, the free channels take free states to free states \cite{MS13}, in the sense that $\mathcal{N}_{A\rightarrow B}(\sigma_{A})$ is a free state if $\mathcal{N}_{A\rightarrow B}$ is a free channel and $\sigma_{A}$ is a free state.

For $\mathcal{N}_{A\rightarrow B}$ a free channel satisfying \eqref{eq:G-cov-channel-def}, the maximum $G$-symmetric fidelity is a resource monotone in the following sense:
\begin{equation}
\max_{\sigma_{A}\in\operatorname*{Sym}_{G}}F(\rho_{A},\sigma_{A}) \leq \max_{\sigma_{B}\in\operatorname*{Sym}_{G}}F(\mathcal{N}_{A\rightarrow B} (\rho_{A}),\sigma_{B}).
\end{equation}
This follows from the facts that the fidelity does not decrease under the action of a quantum channel and the free channels take free states to free states.

\subsection{Resource theory of Bose asymmetry}

Now we define the resource theory of Bose asymmetry and prove that the acceptance probability $\operatorname{Tr}[\Pi_{A}^{G}\rho_{A}]$ of Algorithm~\ref{alg:simple} is a resource monotone in this resource theory. This demonstrates that $\operatorname{Tr}[\Pi_{A}^{G}\rho_{A}]$ is a legitimate quantifier of Bose symmetry of a state.

Following the same notation as in Section~\ref{sec:rt-asym}, recall that a state $\sigma_{A}$ is Bose symmetric if the following condition holds
\begin{equation}
\sigma_{A}=\Pi_{A}^{G}\sigma_{A}\Pi_{A}^{G},
\end{equation}
where $\Pi_{A}^{G}$ is given by \eqref{eq:simple-case-g-proj}. Similarly, a state $\tau_{B}$ is Bose symmetric if it obeys the same conditions but for the projector $\Pi_{B}^{G}$ specified by $\{V_B(g)\}_{g\in G}$. These are the free states in the resource theory of Bose asymmetry.

To define the resource theory, we need to specify the free channels.

\begin{definition}
[Bose symmetric channel]We define a channel $\mathcal{N}_{A\rightarrow B}$ to be a Bose symmetric channel (i.e., free channel) if the following condition holds
\begin{equation}
\left(  \mathcal{N}_{A\rightarrow B}\right)  ^{\dag}(\Pi_{B}^{G})\geq\Pi_{A}^{G}, \label{eq:Bose-sym-channel}
\end{equation}
where $\left(  \mathcal{N}_{A\rightarrow B}\right)  ^{\dag}$ is the Hilbert--Schmidt adjoint of $\mathcal{N}_{A\rightarrow B}$
\cite{Wbook17,KW20book}.
\end{definition}

\begin{proposition}
Bose symmetric channels include the identity channel and they obey the transitive property in \eqref{eq:trans-prop}. Additionally, Bose symmetric states are a special case of Bose symmetric channels when the input space is trivial.
\end{proposition}

\begin{proof}
When the input and output systems are the same, as well as the unitary representations, it follows that $\Pi_{B}^{G}=\Pi_{A}^{G}$. Since the identity channel is its own adjoint, we then conclude that \eqref{eq:Bose-sym-channel} holds for the identity channel.

Suppose that $\mathcal{N}_{A\rightarrow B}$ is a quantum channel that obeys the condition in \eqref{eq:Bose-sym-channel}. Let $\{W_{C}(g)\}_{g\in G}$ be a projective unitary representation of $G$, and suppose that $\mathcal{M}_{B\rightarrow C}$ is a Bose symmetric channel satisfying
\begin{equation}
\left(  \mathcal{M}_{B\rightarrow C}\right)  ^{\dag}(\Pi_{C}^{G})\geq\Pi_{B}^{G},
\end{equation}
where $\Pi_{C}^{G}\coloneqq \frac{1}{\left\vert G \right\vert }\sum_{g\in G}W_{C}(g)$. Consider that
\begin{align}
&  \left(  \mathcal{M}_{B\rightarrow C}\circ\mathcal{N}_{A \rightarrow B}\right)  ^{\dag}(\Pi_{C}^{G})\nonumber\\
&  =(\mathcal{N}_{A\rightarrow B})^{\dag}\left[  \left(  \mathcal{M}_{B\rightarrow C} \right)  ^{\dag} (\Pi_{C}^{G})\right] \\
&  \geq(\mathcal{N}_{A\rightarrow B})^{\dag}\left[  \Pi_{B}^{G}\right] \\
&  \geq\Pi_{A}^{G}.
\end{align}
The first equality follows by exploiting the identity $\left(  \mathcal{M}_{B\rightarrow C} \circ \mathcal{N}_{A\rightarrow B}\right)  ^{\dag}=(\mathcal{N}_{A\rightarrow B})^{\dag} \circ \left(  \mathcal{M}_{B\rightarrow C}\right)  ^{\dag}$ for adjoints. The first inequality follows from the assumption that $\mathcal{M}_{B\rightarrow C}$ is a Bose symmetric channel and from the fact that $\mathcal{N}_{A\rightarrow B}$ is completely positive, so that $(\mathcal{N}_{A\rightarrow B})^{\dag}$ is also. We thus conclude that $\mathcal{M}_{B\rightarrow C} \circ \mathcal{N}_{A\rightarrow B}$ is a Bose symmetric channel, so that the transitive property in \eqref{eq:trans-prop} holds.

Finally, suppose that the input system $A$ of a Bose symmetric channel $\mathcal{N}_{A\rightarrow B}$ is trivial. Then each group element $g$ is trivially represented by the number one. It follows that $\Pi_{A}^{G}=1$. Then the channel $\mathcal{N}_{A\rightarrow B}$ is really just a state $\omega_{B}$ \cite{Wbook17} with a spectral decomposition $\omega_{B}=\sum_{x}p(x)|x\rangle\!\langle x|_{B}$; furthermore, the associated Kraus operators are given by $\{\sqrt{p(x)}|x\rangle_{B}\}_{x}$. Then the condition
\begin{equation}
\left(  \mathcal{N}_{A\rightarrow B}\right)  ^{\dag}(\Pi_{B}^{G})\geq\Pi_{A}^{G}
\end{equation}
reduces to
\begin{equation}
\sum_{x}p(x)\langle x|_{B}\Pi_{B}^{G}|x\rangle_{B}\geq1,
\end{equation}
which is the same as
\begin{equation}
\operatorname{Tr}[\Pi_{B}^{G}\omega_{B}]\geq1.
\end{equation}
Since $\omega_{B}$ is a state and $\Pi_{B}^{G}$ is a projection, it follows that $\operatorname{Tr}[\Pi_{B}^{G}\omega_{B}]\leq1$. Combining these inequalities, we conclude that $\operatorname{Tr}[\Pi_{B}^{G}\omega_{B}]=1$. Finally, we apply \eqref{eq:Bose-symmetric-equiv-cond} to conclude that $\omega_{B}$ is a Bose symmetric state.
\end{proof}

\begin{theorem}
\label{thm:G-Bose-sym-states-preserved}
Suppose that a quantum channel $\mathcal{N}_{A\rightarrow B}$ obeys the condition in \eqref{eq:Bose-sym-channel}. Let $\sigma_{A}$ be a Bose symmetric state. Then $\mathcal{N}_{A\rightarrow B}(\sigma_{A})$ is a Bose symmetric state.
\end{theorem}

\begin{proof}
Recall from \eqref{eq:Bose-symmetric-equiv-cond}\ that a state $\sigma_{A}$ is Bose symmetric if and only if $\operatorname{Tr}[\Pi_{A}^{G}\sigma_{A}]=1$. Then consider that
\begin{align}
1  &  \geq\operatorname{Tr}[\Pi_{B}^{G}\mathcal{N}_{A\rightarrow B}(\sigma
_{A})]\\
&  =\operatorname{Tr}[\left(  \mathcal{N}_{A\rightarrow B}\right)  ^{\dag} (\Pi_{B}^{G})\sigma_{A}]\\
&  \geq\operatorname{Tr}[\Pi_{A}^{G}\sigma_{A}]\\
&  =1.
\end{align}
It follows that $\operatorname{Tr}[\Pi_{B}^{G}\mathcal{N}_{A\rightarrow B}(\sigma_{A})]=1$, and, by applying \eqref{eq:Bose-symmetric-equiv-cond} again, that $\mathcal{N}_{A\rightarrow B}(\sigma_{A})$ is Bose symmetric.
\end{proof}

\medskip
By essentially the same proof, it follows that the measure $\operatorname{Tr}[\Pi_{A}^{G}\rho_{A}]$ from \eqref{eq:acc-prob-bose-test}\ is non-decreasing under the action of a Bose symmetric channel $\mathcal{N}_{A\rightarrow B}$. Thus, the acceptance probability $\operatorname{Tr}[\Pi_{A}^{G}\rho_{A}]$ of a Bose symmetry test is a resource monotone in the resource theory of Bose asymmetry.

\begin{theorem}
\label{thm:res-mono-G-Bose-sym}Let $\rho_{A}$ be a state, and let $\mathcal{N}_{A\rightarrow B}$ be a Bose symmetric channel. Then $\operatorname{Tr}[\Pi_{A}^{G}\rho_{A}]$ is a resource monotone in the following sense:
\begin{equation}
\operatorname{Tr}[\Pi_{B}^{G}\mathcal{N}_{A\rightarrow B}(\rho_{A})]\geq\operatorname{Tr}[\Pi_{A}^{G}\rho_{A}].
\end{equation}

\end{theorem}

\begin{proof}
Consider that
\begin{align}
\operatorname{Tr}[\Pi_{B}^{G}\mathcal{N}_{A\rightarrow B}(\rho_{A})]  &
=\operatorname{Tr}[\left(  \mathcal{N}_{A\rightarrow B}\right)  ^{\dag}
(\Pi_{B}^{G})\rho_{A}]\\
&  \geq\operatorname{Tr}[\Pi_{A}^{G}\rho_{A}],
\end{align}
which follows from \eqref{eq:Bose-sym-channel}.

Alternatively, this follows from Theorem~\ref{thm:acc-prob-g-Bose-sym}, Theorem~\ref{thm:G-Bose-sym-states-preserved}, and the data-processing inequality for fidelity under quantum channels.
\end{proof}

\medskip
Throughout this section, we have adopted the perspective that Bose symmetric channels are defined by the condition in \eqref{eq:Bose-sym-channel}. It then follows as a consequence that $\operatorname{Tr}[\Pi_{A}^{G}\rho_{A}]$ is a resource monotone. We can adopt a different perspective and conclude consistency between them. Let us instead suppose that $\operatorname{Tr}[\Pi_{A}^{G}\rho_{A}]$ is non-decreasing under the action of a free channel $\mathcal{N}_{A\rightarrow B}$. That is, suppose that the following inequality holds for every state~$\rho_{A}$:
\begin{equation}
\operatorname{Tr}[\Pi_{B}^{G}\mathcal{N}_{A\rightarrow B}(\rho_{A})]\geq\operatorname{Tr}[\Pi_{A}^{G}\rho_{A}].
\end{equation}
Then by rewriting this inequality as
\begin{equation}
\operatorname{Tr}[(\left(  \mathcal{N}_{A\rightarrow B}\right)  ^{\dag} (\Pi_{B}^{G})-\Pi_{A}^{G})\rho_{A}]\geq0\quad\forall\rho_{A} \in\mathcal{D} (\mathcal{H}_{A}),
\end{equation}
we conclude that $\left(  \mathcal{N}_{A\rightarrow B}\right)  ^{\dag}(\Pi_{B}^{G})-\Pi_{A}^{G}$ is a positive semi-definite operator, which is equivalent to the condition in \eqref{eq:Bose-sym-channel}. Thus, $\mathcal{N}_{A\rightarrow B}$ is a Bose symmetric channel if and only if $\operatorname{Tr}[\Pi_{A}^{G}\rho_{A}]$ is a resource monotone.

\subsection{Resource theory of asymmetric unextendibility}

\label{sec:rt-asym-unext}We now propose a resource theory that generalizes that proposed in \cite{KDWW19,KDWW21}, just as the set of $G$-symmetric extendible states generalizes the set of $k$-extendible states (recall Example~\ref{ex:k-ext}). One of the main ideas is to use the notion of channel extension introduced in \cite{KDWW19,KDWW21}; additionally, this resource theory allows us to conclude that the acceptance probability of Algorithm~\ref{alg:sym-ext} (i.e., the maximum $G$-symmetric extendible fidelity) is a resource monotone and thus well motivated in this sense.

Let $G$ be a group, and let $\{U_{RS}(g)\}_{g\in G}$ and $\{V_{R^{\prime}S^{\prime}}(g)\}_{g\in G}$ be projective unitary representations of $G$ acting on $\mathcal{H}_{R}\otimes\mathcal{H}_{S}$ and $\mathcal{H}_{R^{\prime}}\otimes \mathcal{H}_{S^{\prime}}$, respectively.

\begin{definition}[$G$-symmetric extendible channel]\label{def:GSE-channels}
A channel $\mathcal{N}_{S\rightarrow S^{\prime}}$ is $G$-symmetric extendible if there exists a bipartite channel $\mathcal{M}_{RS\rightarrow R^{\prime}S^{\prime}}$ such that

\begin{enumerate}
\item $\mathcal{M}_{RS\rightarrow R^{\prime}S^{\prime}}$ is a channel extension of $\mathcal{N}_{S\rightarrow S^{\prime}}$:
\begin{equation}
\operatorname{Tr}_{R^{\prime}}\circ\mathcal{M}_{RS\rightarrow R^{\prime}S^{\prime}}=\mathcal{N}_{S\rightarrow S^{\prime}}\circ\operatorname{Tr}_{R},
\label{eq:channel-ext}
\end{equation}

\item $\mathcal{M}_{RS\rightarrow R^{\prime}S^{\prime}}$ is covariant with respect to $\{U_{RS}(g)\}_{g\in G}$ and $\{V_{R^{\prime}S^{\prime}}(g)\}_{g\in
G}$:
\begin{equation}
\mathcal{M}_{RS\rightarrow R^{\prime} S^{\prime}} \circ \mathcal{U}_{RS}(g) = \mathcal{V}_{R^{\prime} S^{\prime}}(g) \circ \mathcal{M}_{RS\rightarrow R^{\prime}S^{\prime}}
\label{eq:covariance-G-sym-ext}
\end{equation}
for all $g\in G$,
where $\mathcal{U}_{RS}(g)(\cdot)$ and $\mathcal{V}_{R^{\prime}S^{\prime}}(g)(\cdot)$ are defined similarly to \eqref{eq:uandv}.
\end{enumerate}
\end{definition}

The condition in \eqref{eq:channel-ext} implies that the extension channel $\mathcal{M}_{RS\rightarrow R^{\prime}S^{\prime}}$ is non-signaling from $R$ to $S^{\prime}$ \cite{BGNP01,ESW02,PHHH06}, in the sense that
\begin{equation}
\operatorname{Tr}_{R^{\prime}} \circ \mathcal{M}_{RS\rightarrow R^{\prime} S^{\prime}} = \operatorname{Tr}_{R^{\prime}} \circ \mathcal{M}_{RS\rightarrow R^{\prime}S^{\prime}} \circ \mathcal{R}_{R}^{\pi}, 
\label{eq:non-sig-1}
\end{equation}
where $\mathcal{R}_{R}^{\pi}(\cdot)\coloneqq \operatorname{Tr}[\cdot]\pi_{R}$ is a replacer channel that traces out its input and replaces it with the maximally mixed state $\pi_{R}$. This follows because
\begin{align}
\operatorname{Tr}_{R^{\prime}}\circ\mathcal{M}_{RS\rightarrow R^{\prime} S^{\prime}} \circ \mathcal{R}_{R}^{\pi}  &  = \mathcal{N}_{S\rightarrow S^{\prime}}\circ\operatorname{Tr}_{R}\circ\mathcal{R}_{R}^{\pi}\\
&  =\mathcal{N}_{S\rightarrow S^{\prime}}\circ\operatorname{Tr}_{R}\\
&  =\operatorname{Tr}_{R^{\prime}}\circ\mathcal{M}_{RS\rightarrow R^{\prime}S^{\prime}}, 
\label{eq:non-sig-4}
\end{align}
where we have exploited the identity in \eqref{eq:channel-ext} in the first and last lines, and in the second line used the fact that $\operatorname{Tr}_{R}\circ\mathcal{R}_{R}^{\pi}=\operatorname{Tr}_{R}$.

Definition~\ref{def:GSE-channels}\ leads to a consistent resource theory of $G$-asymmetric unextendibility, in the sense that the free states are $G$-symmetric extendible states and the output of a $G$-symmetric extendible channel acting on a $G$-symmetric extendible state is a $G$-symmetric extendible state.

\begin{proposition}
\label{prop:g-sym-ext-ch-triv-input-state}A $G$-symmetric extendible channel $\mathcal{N}_{S\rightarrow S^{\prime}}$\ with trivial input system is a $G$-symmetric extendible state.
\end{proposition}

\begin{proof}
If the input system $S$ of $\mathcal{N}_{S\rightarrow S^{\prime}}$ is trivial, then it follows that $\mathcal{N}_{S\rightarrow S^{\prime}}$ is a state (call it $\rho_{S^{\prime}}$); furthermore, we can choose the input system $R$ of the extension channel $\mathcal{M}_{RS\rightarrow R^{\prime}S^{\prime}}$ to be trivial, in which case $\mathcal{M}_{RS\rightarrow R^{\prime}S^{\prime}}$ is a state (call it $\omega_{R^{\prime}S^{\prime}}$) that extends $\rho_{S^{\prime}}$. The condition in \eqref{eq:covariance-G-sym-ext}\ then collapses to $\omega_{R^{\prime}S^{\prime}}=\mathcal{V}_{R^{\prime}S^{\prime}} (g)(\omega_{R^{\prime}S^{\prime}})$ for all $g\in G$. It follows by Definition~\ref{def:g-bose-sym-ext} that $\rho_{S^{\prime}}$ is a $G$-symmetric extendible state.
\end{proof}

\begin{proposition}
\label{prop:golden-rule-G-SE}Let $\mathcal{N}_{S\rightarrow S^{\prime}}$ be a $G$-symmetric extendible channel, and let $\rho_{S}$ be a $G$-symmetric extendible state. Then $\mathcal{N}_{S\rightarrow S^{\prime}}(\rho_{S})$ is a $G$-symmetric extendible state.
\end{proposition}

\begin{proof}
Since $\rho_{S}$ is a $G$-symmetric extendible state, by Definition~\ref{def:g-sym-ext}, there exists an extension state $\omega_{RS}$
satisfying the conditions stated there. Since $\mathcal{N}_{S\rightarrow S^{\prime}}$ is a $G$-symmetric extendible channel, by
Definition~\ref{def:GSE-channels}, there exists an extension channel $\mathcal{M}_{RS\rightarrow R^{\prime}S^{\prime}}$ satisfying the conditions stated there. It follows that $\mathcal{M}_{RS\rightarrow R^{\prime}S^{\prime}}(\omega_{RS})$ is an extension of $\mathcal{N}_{S\rightarrow S^{\prime}}(\rho_{S})$ as
\begin{align}
\operatorname{Tr}_{R^{\prime}}[\mathcal{M}_{RS\rightarrow R^{\prime}S^{\prime}} (\omega_{RS})]  &  =\mathcal{N}_{S\rightarrow S^{\prime}}(\operatorname{Tr}_{R}[\omega_{RS}])\\
&  =\mathcal{N}_{S\rightarrow S^{\prime}}(\rho_{S}),
\end{align}
where the first equality follows from \eqref{eq:channel-ext}. Also, consider
that the following holds for all $g\in G$:
\begin{align}
&  (\mathcal{V}_{R^{\prime}S^{\prime}}(g)\circ\mathcal{M}_{RS\rightarrow R^{\prime}S^{\prime}})(\omega_{RS})\nonumber\\
&  =(\mathcal{M}_{RS\rightarrow R^{\prime}S^{\prime}}\circ\mathcal{U}_{RS}(g))(\omega_{RS})\\
&  =\mathcal{M}_{RS\rightarrow R^{\prime} S^{\prime}}(\omega_{RS}),
\end{align}
where the first equality follows from \eqref{eq:covariance-G-sym-ext} and the second from \eqref{eq:G-ext-2}.
\end{proof}

\medskip
As a consequence of Proposition~\ref{prop:golden-rule-G-SE}\ and the data-processing inequality for fidelity, the maximum $G$-symmetric extendible fidelity is a resource monotone.

\begin{corollary}
Let $\rho_{S}$ be a state, and let $\mathcal{N}_{S\rightarrow S^{\prime}}$ be a $G$-symmetric extendible channel. Then the maximum $G$-symmetric extendible fidelity is a resource monotone,
\begin{multline}
\max_{\sigma_{S}\in\operatorname*{SymExt}_{G}}F(\rho_{S},\sigma_{S})\\
\leq \max_{\sigma_{S^{\prime}}\in\operatorname*{SymExt}_{G}}F(\mathcal{N}_{S\rightarrow S^{\prime}}(\rho_{S}),\sigma_{S^{\prime}}).
\end{multline}

\end{corollary}

\begin{example}
[$k$-unextendibility]The resource theory of $k$-unextendibility, proposed in \cite{KDWW19,KDWW21}, is a special case of the resource theory of $G$-asymmetric unextendibility. To see this, recall that a bipartite channel $\mathcal{N}_{AB\rightarrow A^{\prime}B^{\prime}}$ is $k$-extendible if there exists an extension channel $\mathcal{M}_{AB_{1}\cdots B_{k}\rightarrow A^{\prime} B_{1}^{\prime} \cdots B_{k}^{\prime}}$ satisfying
\begin{multline}
\operatorname{Tr}_{B_{2}^{\prime}\cdots B_{k}^{\prime}}\circ\mathcal{M}_{AB_{1}\cdots B_{k}\rightarrow A^{\prime}B_{1}^{\prime}\cdots B_{k}^{\prime} } \\ 
=\mathcal{N}_{AB\rightarrow A^{\prime}B^{\prime}}\circ\operatorname{Tr}_{B_{2}\cdots B_{k}}
\end{multline}
and
\begin{multline}
\mathcal{W}_{B_{1}^{\prime}\cdots B_{k}^{\prime}}^{\pi}\circ\mathcal{M}_{AB_{1}\cdots B_{k}\rightarrow A^{\prime}B_{1}^{\prime}\cdots B_{k}^{\prime}
}\\
=\mathcal{M}_{AB_{1}\cdots B_{k}\rightarrow A^{\prime}B_{1}^{\prime}\cdots B_{k}^{\prime}}\circ\mathcal{W}_{B_{1}\cdots B_{k}}^{\pi},
\end{multline}
for all $\pi \in S_k$, where $\mathcal{W}_{B_{1}\cdots B_{k}}^{\pi}$ and $\mathcal{W}_{B_{1}^{\prime}\cdots B_{k}^{\prime}}^{\pi}$ are unitary permutation channels. 
Thus, by setting
\begin{align}
S  &  =AB,\\
R  &  =B_{2}\cdots B_{k},\\
S^{\prime}  &  =A^{\prime}B^{\prime},\\
R^{\prime}  &  =B_{2}^{\prime}\cdots B_{k}^{\prime},\\
U_{RS}(g)  &  =\mathbb{I}_{A}\otimes W_{B_{1}\cdots B_{k}}(\pi),\\
V_{R^{\prime}S^{\prime}}(g)  &  =\mathbb{I}_{A^{\prime}}\otimes W_{B_{1}^{\prime}\cdots B_{k}^{\prime}}(\pi),
\end{align}
we see that a $k$-extendible channel is a special case of a $G$-symmetric extendible channel.
\end{example}

\subsection{Resource theory of Bose asymmetric unextendibility}

We finally consider the resource theory of Bose asymmetric unextendibility, with the goal being similar to that of the previous sections; we want to justify the acceptance probability of Algorithm~\ref{alg:G-BSE-test} (i.e., the maximum $G$-BSE fidelity) as a resource monotone. At the same time, we establish a novel resource theory that could have further applications in quantum information.

Let $G$, $\{U_{RS}(g)\}_{g\in G}$, and $\{V_{R^{\prime}S^{\prime}}(g)\}_{g\in G}$ be defined the same way as in Section~\ref{sec:rt-asym-unext}.

\begin{definition}
[$G$-BSE channel]\label{def:G-BSE-channels}A channel $\mathcal{N}_{S\rightarrow S^{\prime}}$ is $G$-Bose symmetric extendible ($G$-BSE) if there exists a bipartite channel $\mathcal{M}_{RS\rightarrow R^{\prime}S^{\prime}}$ such that

\begin{enumerate}
\item $\mathcal{M}_{RS\rightarrow R^{\prime}S^{\prime}}$ is a channel extension of $\mathcal{N}_{S\rightarrow S^{\prime}}$:
\begin{equation}
\operatorname{Tr}_{R^{\prime}}\circ\mathcal{M}_{RS\rightarrow R^{\prime}S^{\prime}}=\mathcal{N}_{S\rightarrow S^{\prime}}\circ\operatorname{Tr}_{R},
\label{eq:channel-ext-bose}
\end{equation}

\item $\mathcal{M}_{RS\rightarrow R^{\prime}S^{\prime}}$ is Bose symmetric:
\begin{equation}
(\mathcal{M}_{RS\rightarrow R^{\prime}S^{\prime}})^{\dag} (\Pi_{R^{\prime}S^{\prime}}^{G})\geq\Pi_{RS}^{G}, \label{eq:BSE-bose-condition}
\end{equation}
where $\Pi_{RS}^{G}$ and $\Pi_{R^{\prime}S^{\prime}}^{G}$ are defined as in \eqref{eq:Pi_RS-proj-again} as sums over $U_{RS}(g)$ and $V_{R^{\prime}S^{\prime}}(g)$ respectively.

\end{enumerate}
\end{definition}

As discussed in \eqref{eq:non-sig-1}--\eqref{eq:non-sig-4}, the condition in \eqref{eq:channel-ext-bose} can be understood as imposing a no-signaling constraint, from $R$ to $S^{\prime}$.

With the same line of reasoning given in the proof of Proposition~\ref{prop:g-sym-ext-ch-triv-input-state}, we conclude the following:

\begin{proposition}
A $G$-BSE channel $\mathcal{N}_{S\rightarrow S^{\prime}}$\ with trivial input system is a $G$-BSE state.
\end{proposition}

The following proposition demonstrates that the resource theory delineated by Definition~\ref{def:G-BSE-channels} is indeed a consistent resource theory.

\begin{proposition}
\label{prop:golden-rule-G-BSE}Let $\mathcal{N}_{S\rightarrow S^{\prime}}$ be a $G$-BSE channel, and let $\rho_{S}$ be a $G$-BSE state. Then $\mathcal{N}_{S\rightarrow S^{\prime}}(\rho_{S})$ is a $G$-BSE state.
\end{proposition}

As this proof is similar to that of Proposition~\ref{prop:golden-rule-G-SE}, we include it in Appendix~\ref{app:proof-prop-G-BSE}. As a consequence of Proposition~\ref{prop:golden-rule-G-BSE}\ and the data-processing inequality for fidelity, it follows that the maximum $G$-BSE fidelity is a resource monotone.

\begin{corollary}
Let $\rho_{S}$ be a state, and let $\mathcal{N}_{S\rightarrow S^{\prime}}$ be a $G$-BSE channel. Then the maximum $G$-BSE fidelity is a resource monotone in the following sense:
\begin{equation}
\max_{\sigma_{S}\in\operatorname*{BSE}_{G}}F(\rho_{S},\sigma_{S})\leq \max_{\sigma_{S^{\prime}}\in\operatorname*{BSE}_{G}}F(\mathcal{N}_{S\rightarrow S^{\prime}}(\rho_{S}),\sigma_{S^{\prime}}).
\end{equation}

\end{corollary}

To the best of our knowledge, the resource theory of $k$-Bose unextendibility has not been proposed in prior work. To define it, we establish the notion of a free channel (i.e., a $k$-Bose extendible bipartite channel) and discuss it in the following example. 
\begin{example}
[$k$-Bose unextendibility]We say that a bipartite channel $\mathcal{N}_{AB\rightarrow A^{\prime}B^{\prime}}$ is $k$-Bose-extendible if there exists an extension channel $\mathcal{M}_{AB_{1}\cdots B_{k}\rightarrow A^{\prime}B_{1}^{\prime} \cdots B_{k}^{\prime}}$ satisfying
\begin{multline}
\operatorname{Tr}_{B_{2}^{\prime}\cdots B_{k}^{\prime}}\circ\mathcal{M}_{AB_{1} \cdots B_{k}\rightarrow A^{\prime}B_{1}^{\prime}\cdots B_{k}^{\prime}}\\
=\mathcal{N}_{AB\rightarrow A^{\prime}B^{\prime}}\circ\operatorname{Tr}_{B_{2}\cdots B_{k}}
\end{multline}
and
\begin{equation}
(\mathcal{M}_{AB_{1}\cdots B_{k}\rightarrow A^{\prime}B_{1}^{\prime}\cdots B_{k}^{\prime}})^{\dag}(\Pi_{B_{1}^{\prime}\cdots B_{k}^{\prime}} ^{\operatorname{Sym}})\geq\Pi_{B_{1}\cdots B_{k}}^{\operatorname{Sym}},
\end{equation}
where $\Pi_{B_{1}^{\prime}\cdots B_{k}^{\prime}}^{\operatorname{Sym}}$ and $\Pi_{B_{1}\cdots B_{k}}^{\operatorname{Sym}}$ are projections onto symmetric subspaces,
\begin{align}
\Pi_{B_{1}\cdots B_{k}}^{\operatorname{Sym}}  &  \coloneqq \frac{1}{k!}\sum_{\pi\in S_{k}}W_{B_{1}\cdots B_{k}}^{\pi},\\
\Pi_{B_{1}^{\prime}\cdots B_{k}^{\prime}}^{\operatorname{Sym}}  &  \coloneqq \frac{1}{k!}\sum_{\pi\in S_{k}}W_{B_{1}^{\prime}\cdots B_{k}^{\prime}}^{\pi},
\end{align}
and $W_{B_{1}\cdots B_{k}}^{\pi}$ and $W_{B_{1}^{\prime}\cdots B_{k}^{\prime}}^{\pi}$ are unitary representations of the permutation $\pi\in S_{k}$. Thus, by setting
\begin{align}
S  &  =AB,\\
R  &  =B_{2}\cdots B_{k},\\
S^{\prime}  &  =A^{\prime}B^{\prime},\\
R^{\prime}  &  =B_{2}^{\prime}\cdots B_{k}^{\prime},\\
U_{RS}(g)  &  =\mathbb{I}_{A}\otimes W_{B_{1}\cdots B_{k}}(\pi),\\
V_{R^{\prime}S^{\prime}}(g)  &  =\mathbb{I}_{A^{\prime}}\otimes W_{B_{1}^{\prime}\cdots B_{k}^{\prime}}(\pi),
\end{align}
we see that a $k$-Bose-extendible channel is a special case of a $G$-Bose symmetric extendible channel.
\end{example}

\section{Conclusion}

\label{sec:conclusion}

In summary, we have proposed various quantum computational tests of symmetry, as well as various notions of symmetry like $G$-symmetric extendibility and $G$-Bose symmetric extendibility, which include previous notions of symmetry from \cite{MS13,MS14,W89a,DPS02,DPS04} as special cases, showing that these these new notions of symmetry provide a generalization with interesting applications. These tests have acceptance probabilities equal to various maximum symmetric fidelities, thus endowing these measures with operational meanings. We have also established resource theories of asymmetry beyond those proposed in \cite{MS13}, which put the maximum symmetric fidelities on firm ground in a resource-theoretic sense. Finally, we evaluated the quantum computational tests on existing quantum computers, by employing a variational algorithm to replace the role of the prover in a quantum interactive proof.

Going forward from here, one could generalize the approach we have taken to any quantum interactive proof by, for instance, replacing the prover with a parameterized circuit. This approach will allow for estimating distinguishability measures like the diamond distance \cite{RW05}. This method is not guaranteed to perform well in general, simply because a variational circuit cannot realize an arbitrarily powerful quantum computation like a quantum prover can. For sufficiently small examples, however, this seemly interesting approach has the potential to go beyond what can be estimated using a classical computer alone. After stating this observation in a preliminary version of this paper \cite{LW21}, this approach was pursued in \cite{RASW21} (see also \cite{KVPYB22}).

We are also interested in generalizing the quantum computational tests proposed here to test for extendibility and symmetry of quantum channels. The algorithm outlined in Section~\ref{sec:cov-sym-ch-test} is an initial finding in this direction, but more generally, we would like to test for $G$-symmetric extendibility and $G$-Bose symmetric extendibility of bipartite and multipartite channels. This would involve testing for the no-signaling constraint in addition to the symmetry constraint of $k$-extendible channels.

\bigskip 

\noindent \textbf{Acknowledgments}.
We dedicate this paper to the memory of Jonathan~P.~Dowling, an outstanding mentor and scientist whose brilliance and friendship will be greatly missed. May he live on through the lives and works of all those he inspired.

We thank Patrick Coles, Zo\"e Holmes, Ludovico Lami, Iman Marvian, Stephan Shipman, and Vishal Singh for discussions. We also acknowledge helpful discussions with ambassadors from Amazon and IBM.
MLL\ acknowledges support from the DoD SMART Scholarship program.
MMW\ and SR\ acknowledge support from NSF\ Grant No.~2315398.

\bigskip 
\noindent \textbf{Data Availability Statement}.
All code and data used to generate these results is available via the GitHub repository located at \url{https://github.com/Soorya-Rethin/Testing-Symmetry}.

\bibliographystyle{quantum}
\bibliography{Ref}

\appendix

\section{Proof of Theorem~\ref{thm:Bose-sym-purify}}

\label{app:Bose-sym-purify}

We give the proof for completeness, and we note here that it is very close to
the proof of \cite[Lemma~II.5]{CKMR08} (see also \cite[Lemma~3.6]{KW20book}).

We begin with the forward implication. Suppose that $\rho_{S}$ is
$G$-symmetric extendible. By definition, this means that there exists a state
$\omega_{RS}$ satisfying \eqref{eq:G-ext-1} and \eqref{eq:G-ext-2}. Suppose
that $\omega_{RS}$ has the following spectral decomposition:
\begin{equation}
\omega_{RS}=\sum_{k}\lambda_{k}\Pi_{RS}^{k},
\end{equation}
where $\lambda_{k}$ is an eigenvalue and $\Pi_{RS}^{k}$ is a spectral
projection. We can write $\Pi_{RS}^{k}$ as
\begin{equation}
\Pi_{RS}^{k}=\sum_{\ell}|\phi_{\ell}^{k}\rangle\!\langle\phi_{\ell}^{k}|_{RS},
\end{equation}
where $\{|\phi_{\ell}^{k}\rangle_{RS}\}_{\ell}$ is an orthonormal basis. Now
define
\begin{align}
|\Gamma^{k}\rangle_{RS\hat{R}\hat{S}}  &  \coloneqq \sum_{\ell}|\phi_{\ell
}^{k}\rangle_{RS}\otimes\overline{|\phi_{\ell}^{k}\rangle}_{\hat{R}\hat{S}},\\
|\psi^{\rho}\rangle_{RS\hat{R}\hat{S}}  &  \coloneqq \sum_{k}\sqrt{\lambda
_{k}}|\Gamma^{k}\rangle_{RS\hat{R}\hat{S}},
\end{align}
where $\overline{|\phi_{\ell}^{k}\rangle}_{\hat{R}\hat{S}}$ is the complex
conjugate of $|\phi_{\ell}^{k}\rangle_{RS}$ with respect to the standard
basis. Observe that $|\psi^{\rho}\rangle\!\langle\psi^{\rho}|_{RS\hat{R}
\hat{S}}$ is a purification of $\omega_{RS}$. Now let us establish
\eqref{eq:g-sym-pur-cond}. Given that $\omega_{RS}$ satisfies
\eqref{eq:G-ext-2}, it follows that
\begin{align}
U_{RS}(g)^{\dag}\omega_{RS}U_{RS}(g)|\phi_{\ell}^{k}\rangle_{RS}  &
=\omega_{RS}|\phi_{\ell}^{k}\rangle_{RS}\\
&  =\lambda_{k}|\phi_{\ell}^{k}\rangle_{RS},
\end{align}
for all $k$, $\ell$, and $g$. Left multiplying by $U_{RS}(g)$ implies that
\begin{equation}
\omega_{RS}U_{RS}(g)|\phi_{\ell}^{k}\rangle_{RS}=\lambda_{k}U_{RS}(g)|\phi_{\ell}^{k}\rangle_{RS},
\end{equation}
so that $U_{RS}(g)|\phi_{\ell}^{k}\rangle_{RS}$ is an eigenvector of $\omega_{RS}$ with eigenvalue $\lambda_{k}$. We conclude that the $k$th eigenspace corresponding to eigenvalue $\lambda_{k}$ is invariant under the action of $U_{RS}(g)$ because $|\phi_{\ell}^{k}\rangle_{RS}$ and $U_{RS}(g)|\phi_{\ell}^{k}\rangle_{RS}$ are eigenvectors of $\omega_{RS}$ with eigenvalue $\lambda_{k}$. This implies that the restriction of $U_{RS}(g)$ to the $k$th eigenspace is equivalent to a unitary $U_{RS}^{k}(g)$. Then it follows that
\begin{align}
&  (U_{RS}(g)\otimes\overline{U}_{\hat{R}\hat{S}}(g))|\Gamma^{k} \rangle_{RS\hat{R}\hat{S}}\nonumber\\
&  =(U_{RS}^{k}(g)\otimes\overline{U}_{\hat{R}\hat{S}}^{k}(g))|\Gamma^{k}\rangle_{RS\hat{R}\hat{S}}\\
&  =|\Gamma^{k}\rangle_{RS\hat{R}\hat{S}},
\end{align}
for all $g\in G$. The first equality follows from the fact stated just above. The second equality follows from the invariance of the maximally entangled vector $|\Gamma^{k}\rangle_{RS\hat{R}\hat{S}}$ under unitaries of the form $V\otimes\overline{V}$. Thus, it follows by linearity that
\begin{equation}
|\psi^{\rho}\rangle_{RS\hat{R}\hat{S}} = (U_{RS}(g)\otimes\overline{U}_{\hat{R}\hat{S}}(g))|\psi^{\rho}\rangle_{RS\hat{R}\hat{S}},
\label{eq:unitary-inv-purified}
\end{equation}
for all $g\in G$, which is the statement of \eqref{eq:g-sym-pur-cond}.

Let us now consider the opposite implication. Suppose that $\psi_{RS\hat{R}\hat{S}}^{\rho}$ is a purification of $\rho_{S}$ and $\psi_{RS\hat{R}\hat{S}}^{\rho}$ satisfies \eqref{eq:g-sym-pur-cond}. Set
\begin{equation}
\omega_{RS}=\operatorname{Tr}_{\hat{R}\hat{S}}[\psi_{RS\hat{R}\hat{S}}^{\rho
}].
\end{equation}
Then $\omega_{RS}$ is an extension of $\rho_{S}$. Furthermore, employing the shorthand $U_{RS}\equiv U_{RS}(g)$ and $\overline{U}_{\hat{R}\hat{S}} \equiv \overline{U}_{\hat{R}\hat{S}}(g)$, we find that $\omega_{RS}=U_{RS}(g)\omega_{RS}U_{RS}(g)^{\dag}$ for all $g\in G$ because
\begin{align}
&  \omega_{RS}\nonumber\\
&  =\operatorname{Tr}_{\hat{R}\hat{S}}[\psi_{RS\hat{R}\hat{S}}^{\rho}]\\
&  =\operatorname{Tr}_{\hat{R}\hat{S}}[(U_{RS}\otimes\overline{U}_{\hat{R} \hat{S}})\psi_{RS\hat{R}\hat{S}}^{\rho}(U_{RS}\otimes\overline{U}_{\hat{R} \hat{S}})^{\dag}]\\
&  =U_{RS}(g)\operatorname{Tr}_{\hat{R}\hat{S}}[\overline{U}_{\hat{R}\hat{S}}(g)\psi_{RS\hat{R}\hat{S}}^{\rho}\overline{U}_{\hat{R}\hat{S}}(g)^{\dag}]U_{RS}(g)^{\dag}\\
&  =U_{RS}(g)\operatorname{Tr}_{\hat{R}\hat{S}}[\overline{U}_{\hat{R}\hat{S}}(g)^{\dag}\overline{U}_{\hat{R}\hat{S}}(g)\psi_{RS\hat{R}\hat{S}}^{\rho}]U_{RS}(g)^{\dag}\\
&  =U_{RS}(g)\operatorname{Tr}_{\hat{R}\hat{S}}[\psi_{RS\hat{R}\hat{S}}^{\rho}]U_{RS}(g)^{\dag}\\
&  =U_{RS}(g)\omega_{RS}U_{RS}(g)^{\dag}.
\end{align}
Thus, it follows that $\rho_{S}$ is $G$-symmetric extendible.

We now justify the equivalence of \eqref{eq:g-sym-pur-cond} and \eqref{eq:g-sym-pur-cond-proj}. Using the result in \eqref{eq:unitary-inv-purified}, observe that
\begin{equation}
|\psi^{\rho}\rangle_{RS\hat{R}\hat{S}}= \frac{1}{|G|}\sum_{g\in G}(U_{RS}(g)\otimes\overline{U}_{\hat{R}\hat{S}}(g))|\psi^{\rho}\rangle_{RS\hat{R}\hat{S}},
\end{equation}
which simplifies to \eqref{eq:g-sym-pur-cond-proj} by substituting in \eqref{eq:projector-ref-unitaries}. Now starting with \eqref{eq:projector-ref-unitaries}, let us apply the property in \eqref{eq:unitaries-and-projs}, and we have that
\begin{equation}
    |\psi^{\rho}\rangle_{RS\hat{R}\hat{S}} = (U_{RS}(g)\otimes\overline{U}_{\hat{R}\hat{S}}(g))\Pi^G_{RS\hat{R}\hat{S}} |\psi^{\rho}\rangle_{RS\hat{R}\hat{S}},
\end{equation}
for all $g \in G$. This reduces to \eqref{eq:g-sym-pur-cond} by applying \eqref{eq:g-sym-pur-cond-proj}.
%Fleshed out explanation
%The equality in \eqref{eq:g-sym-pur-cond-proj}\ follows from \eqref{eq:unitary-inv-purified} and\ \eqref{eq:projector-ref-unitaries}. The equality in \eqref{eq:g-sym-pur-cond} follows from \eqref{eq:g-sym-pur-cond-proj} by applying the property in \eqref{eq:unitaries-and-projs}.

\section{Acceptance probabilities of Algorithms~\ref{alg:simple}--\ref{alg:sym-ext} as maximum symmetric fidelities}

In the subsections of this appendix, we prove that the acceptance probabilities of Algorithms~\ref{alg:simple}--\ref{alg:sym-ext} are given by maximum symmetric fidelities. That is, we prove Theorems~\ref{thm:acc-prob-g-Bose-sym}, \ref{thm:max-acc-prob-g-sym}, \ref{thm:G-BSE-acc-prob}, and \ref{thm:G-SE-acc-prob}.

\subsection{Proof of Theorem~\ref{thm:acc-prob-g-Bose-sym}}

\label{app:acc-prob-g-Bose-sym}

Let $\psi_{RS}$ be an arbitrary purification of $\rho_{S}$, and consider that
\begin{align}
\operatorname{Tr}[\Pi_{S}^{G}\rho_{S}]  & =\operatorname{Tr}[(\mathbb{I}_{R}\otimes
\Pi_{S}^{G})\psi_{RS}]\\
& =\left\Vert \left(  \mathbb{I}_{R}\otimes\Pi_{S}^{G}\right)  |\psi\rangle
_{RS}\right\Vert _{2}^{2}.
\end{align}
Recall the following property of the norm of an arbitrary vector
$|\varphi\rangle$:
\begin{equation}
\left\Vert |\varphi\rangle\right\Vert _{2}^{2}=\max_{|\phi\rangle:\left\Vert
|\phi\rangle\right\Vert _{2}=1}\left\vert \langle\phi|\varphi\rangle
\right\vert ^{2}.
\label{eq:euclidean-norm-opt}
\end{equation}
This follows from the Cauchy--Schwarz inequality and the conditions for
saturating it.
This implies that
\begin{multline}
 \left\Vert \left(  \mathbb{I}_{R}\otimes\Pi_{S}^{G}\right)  |\psi\rangle
_{RS}\right\Vert _{2}^{2}\\
 =\max_{|\phi\rangle:\left\Vert |\phi\rangle\right\Vert _{2}=1}\left\vert
\langle\phi|_{RS}\left(  \mathbb{I}_{R}\otimes\Pi_{S}^{G}\right)  |\psi\rangle
_{RS}\right\vert ^{2}
\end{multline}
Let us also recall Uhlmann's theorem \cite{U76}: For positive semi-definite operators $\omega_{A}$ and $\tau_{A}$ and
corresponding rank-one operators $\psi_{RA}^{\omega}$ and $\psi_{RA}^{\tau}$
satisfying
\begin{align}
\operatorname{Tr}_{R}[\psi_{RA}^{\omega}]  &  =\omega_{A},\label{eq:uhlmann-thm-1}\\
\operatorname{Tr}_{R}[\psi_{RA}^{\tau}]  &  =\tau_{A},
\end{align}
Uhlmann's theorem \cite{U76} states  that
\begin{align}
& F(\omega_A,\tau_A) \notag \\
& = \left\Vert \sqrt{\omega_{A}}\sqrt{\tau_{A}}\right\Vert _{1}^{2}\\
& =\max_{V_{R}
}\left\vert \langle\psi^{\omega}|_{RA}\left(  V_{R}\otimes \mathbb{I}_{A}\right)
|\psi^{\tau}\rangle_{RA}\right\vert ^{2},
\label{eq:uhlmann-thm-last}
\end{align}
where the optimization is over every unitary $V_{R}$ acting on the reference
system $R$. We also implicitly defined fidelity more generally for positive semi-definite operators. Considering that
\begin{equation}
\rho_S  = \operatorname{Tr}_R[\psi_{RS}], \qquad 
 \sigma_S  \coloneqq \operatorname{Tr}_R[  \phi_{RS} ],
\end{equation}
so that 
\begin{equation}
    \Pi^G_S \sigma_S\Pi^G_S = \operatorname{Tr}_R[ \Pi^G_S \phi_{RS} \Pi^G_S],
\end{equation}
we conclude that
\begin{align}
& \max_{|\phi\rangle:\left\Vert |\phi\rangle\right\Vert _{2}=1}\left\vert
\langle\phi|_{RS}\left(  \mathbb{I}_{R}\otimes\Pi_{S}^{G}\right)  |\psi\rangle
_{RS}\right\vert ^{2} \notag \\
& =  \max_{|\phi\rangle:\left\Vert |\phi\rangle\right\Vert _{2}=1}\max_{U_R}\left\vert
\langle\phi|_{RS}\left(  U_{R}\otimes\Pi_{S}^{G}\right)  |\psi\rangle
_{RS}\right\vert ^{2} \\
& =\max_{\sigma_{S}\in\mathcal{D}(\mathcal{H}_{S})}F(\rho_{S},\Pi_{S}
^{G}\sigma_{S}\Pi_{S}^{G}).
\end{align}
where the last equality follows from Uhlmann's theorem with the identifications $\ket{\psi^{\omega}} \leftrightarrow (\mathbb{I} \otimes \Pi^G ) \ket{\phi}$ and $\ket{\psi^{\tau}} \leftrightarrow  \ket{\psi}$. Clearly, we have that
\begin{align}
& \max_{\sigma_{S}\in\mathcal{D}(\mathcal{H}_{S})}F(\rho_{S},\Pi_{S}^{G}
\sigma_{S}\Pi_{S}^{G})  \notag \\
& \geq\max_{\sigma\in\text{B-Sym}_{G}}F(\rho_{S}
,\Pi_{S}^{G}\sigma_{S}\Pi_{S}^{G})\\
& =\max_{\sigma\in\text{B-Sym}_{G}}F(\rho_{S},\sigma_{S}),
\end{align}
because B-Sym$_{G}\subset\mathcal{D}(\mathcal{H})$. Now let us consider
showing the opposite inequality. Let $\sigma\in\mathcal{D}(\mathcal{H})$. If
$\Pi^{G}\sigma\Pi^{G}=0$, then this is a suboptimal choice as it follows that
the objective function $F(\rho_{S},\Pi_{S}^{G}\sigma_{S}\Pi_{S}^{G})=0$ in
this case. So, let us suppose this is not the case. Then define
\begin{align}
\sigma^{\prime}  & \coloneqq \frac{1}{p}\Pi^{G}\sigma\Pi^{G},\\
p  & \coloneqq \operatorname{Tr}[\Pi^{G}\sigma],
\end{align}
and observe that $\sigma_{S}^{\prime}\in$B-Sym$_{G}$. Consider that
\begin{align}
F(\rho_{S},\Pi_{S}^{G}\sigma_{S}\Pi_{S}^{G})  & =pF(\rho_{S},\sigma
_{S}^{\prime})\\
& \leq F(\rho_{S},\sigma_{S}^{\prime})\\
& \leq\max_{\sigma_{S}\in\text{B-Sym}_{G}}F(\rho_{S},\sigma_{S}).
\end{align}
We have thus proved the opposite inequality, concluding the proof.

\subsection{Proof of Theorem~\ref{thm:max-acc-prob-g-sym}}

\label{app:max-acc-prob-g-sym}

The formula in \eqref{eq:euclidean-norm-opt} implies that
\begin{multline}
\max_{V_{S^{\prime}E\rightarrow\hat{S}E^{\prime}}}\left\Vert \Pi_{S\hat{S}
}^{G}V_{S^{\prime}E\rightarrow\hat{S}E^{\prime}}|\psi\rangle_{S^{\prime}
S}|0\rangle_{E}\right\Vert _{2}^{2}\label{eq:max-acc-prob-proof-step-1} = \\
\max_{\substack{V_{S^{\prime}E\rightarrow\hat{S}E^{\prime}},
\\
|\phi\rangle_{S\hat
{S}E^{\prime}}}}\left\vert \langle\phi|_{S\hat{S}E^{\prime}}\Pi_{S\hat{S}}
^{G}V_{S^{\prime}E\rightarrow\hat{S}E^{\prime}}|\psi\rangle_{S^{\prime}
S}|0\rangle_{E}\right\vert ^{2}.
\end{multline}
Applying Uhlmann's theorem (see \eqref{eq:uhlmann-thm-1}--\eqref{eq:uhlmann-thm-last}) to \eqref{eq:max-acc-prob-proof-step-1} with
the identifications $R\leftrightarrow\hat{S}E^{\prime}\simeq S^{\prime}E$ and
$S\leftrightarrow A$ and noting that
\begin{align}
\operatorname{Tr}_{S^{\prime}E}[|\psi\rangle\!\langle\psi|_{S^{\prime}
S}\otimes|0\rangle\!\langle0|_{E}]  &  =\rho_{S},\\
\operatorname{Tr}_{\hat{S}E^{\prime}}[\Pi_{S\hat{S}}^{G}|\phi\rangle
\!\langle\phi|_{S\hat{S}E^{\prime}}\Pi_{S\hat{S}}^{G}]  &  =\operatorname{Tr}
_{\hat{S}}[\Pi_{S\hat{S}}^{G}\sigma_{S\hat{S}^{\prime}}\Pi_{S\hat{S}}^{G}],
\end{align}
where $\sigma_{S\hat{S}^{\prime}}$ is a quantum state satisfying
$\sigma_{S\hat{S}^{\prime}}=\operatorname{Tr}_{E^{\prime}}[|\phi
\rangle\!\langle\phi|_{S\hat{S}E^{\prime}}]$, we conclude that
\begin{multline}
\max_{\substack{V_{S^{\prime}E\rightarrow\hat{S}E^{\prime}},
\\
|\phi\rangle_{S\hat
{S}E^{\prime}}}}\left\vert \langle\phi|_{S\hat{S}E^{\prime}}\Pi_{S\hat{S}}
^{G}V_{S^{\prime}E\rightarrow\hat{S}E^{\prime}}|\psi\rangle_{S^{\prime}
S}|0\rangle_{E}\right\vert ^{2}\label{eq:1st-fid-formula-pf}\\
=\max_{\sigma_{S\hat{S}^{\prime}}}F(\rho_{S},\operatorname{Tr}_{\hat{S}}
[\Pi_{S\hat{S}}^{G}\sigma_{S\hat{S}^{\prime}}\Pi_{S\hat{S}}^{G}]),
\end{multline}
with the optimization in the last line over every quantum state $\sigma
_{S\hat{S}^{\prime}}$.

We finally prove that
\begin{equation}\label{eq:final-eq-steps-1}
\max_{\sigma_{S\hat{S}^{\prime}}}F(\rho_{S},\operatorname{Tr}_{\hat{S}}
[\Pi_{S\hat{S}}^{G}\sigma_{S\hat{S}^{\prime}}\Pi_{S\hat{S}}^{G}])=\max
_{\sigma_{S}\in\operatorname{Sym}_{G}}F(\rho_{S},\sigma_{S}).
\end{equation}
To justify the inequality $\geq$ in \eqref{eq:final-eq-steps-1}, let $\sigma_{S}\in\operatorname{Sym}_{G}$, and pick $\sigma_{S\hat{S}}$ to be the purification $\varphi_{S\hat{S}}$ of
$\sigma_{S}$ from Theorem~\ref{thm:Bose-sym-purify}\ (with trivial reference
systems $R\hat{R}$)\ that satisfies
\begin{equation}
\Pi_{S\hat{S}}^{G}\varphi_{S\hat{S}}\Pi_{S\hat{S}}^{G}=\varphi_{S\hat{S}}.
\end{equation}
Then we find that
\begin{equation}
\operatorname{Tr}_{\hat{S}}[\Pi_{S\hat{S}}^{G}\varphi_{S\hat{S}}\Pi_{S\hat{S}
}^{G}]=\operatorname{Tr}_{\hat{S}}[\varphi_{S\hat{S}}]=\sigma_{S},
\end{equation}
and so, given that $\sigma_{S}\in\operatorname{Sym}_{G}$ is arbitrary, it
follows that
\begin{equation}
\max_{\sigma_{S\hat{S}^{\prime}}}F(\rho_{S},\operatorname{Tr}_{\hat{S}}
[\Pi_{S\hat{S}}^{G}\sigma_{S\hat{S}^{\prime}}\Pi_{S\hat{S}}^{G}])\geq
\max_{\sigma_{S}\in\operatorname{Sym}_{G}}F(\rho_{S},\sigma_{S}).
\end{equation}
To justify the inequality $\leq$ in \eqref{eq:final-eq-steps-1}, let $\sigma_{S\hat{S}}$ be an arbitrary
state. If $\sigma_{S\hat{S}^{\prime}}$ is outside of the subspace onto which
$\Pi_{S\hat{S}}^{G}$ projects, then $\Pi_{S\hat{S}}^{G}\sigma_{S\hat
{S}^{\prime}}\Pi_{S\hat{S}}^{G}=0$ and the fidelity in
\eqref{eq:1st-fid-formula-pf} is equal to zero. Let us then suppose that this
is not the case, and let us define
\begin{align}
\sigma_{S\hat{S}}^{\prime}  &  \coloneqq \frac{1}{p}\Pi_{S\hat{S}}^{G}
\sigma_{S\hat{S}^{\prime}}\Pi_{S\hat{S}}^{G},\\
p  &  \coloneqq \operatorname{Tr}[\Pi_{S\hat{S}}^{G}\sigma_{S\hat{S}^{\prime}
}].
\end{align}
Then we find that
\begin{align}
F(\rho_{S},\operatorname{Tr}_{\hat{S}}[\Pi_{S\hat{S}}^{G}\sigma_{S\hat
{S}^{\prime}}\Pi_{S\hat{S}}^{G}])  &  =pF(\rho_{S},\tau_{S})\\
&  \leq F(\rho_{S},\tau_{S}),
\end{align}
where
\begin{equation}
\tau_{S}\coloneqq \operatorname{Tr}_{\hat{S}}[\sigma_{S\hat{S}}^{\prime}],
\end{equation}
and we used the fact that $p\leq1$. It remains to be proven that $\tau_{S}
\in\operatorname{Sym}_{G}$. To see this, consider that
\begin{align}
\tau_{S}  &  =\operatorname{Tr}_{\hat{S}}[\sigma_{S\hat{S}}^{\prime}]\\
&  =\operatorname{Tr}_{\hat{S}}[\Pi_{S\hat{S}}^{G}\sigma_{S\hat{S}}^{\prime
}\Pi_{S\hat{S}}^{G}]\\
&  =\operatorname{Tr}_{\hat{S}}[\left(  U_{S}\otimes\overline{U}_{\hat{S}
}\right)  \Pi_{S\hat{S}}^{G}\sigma_{S\hat{S}}^{\prime}\Pi_{S\hat{S}}
^{G}\left(  U_{S}\otimes\overline{U}_{\hat{S}}\right)  ^{\dag}]\\
&  =U_{S}\operatorname{Tr}_{\hat{S}}[\overline{U}_{\hat{S}}\Pi_{S\hat{S}}
^{G}\sigma_{S\hat{S}}^{\prime}\Pi_{S\hat{S}}^{G}\overline{U}_{\hat{S}}^{\dag
}]U_{S}^{\dag}\\
&  =U_{S}\operatorname{Tr}_{\hat{S}}[\overline{U}_{\hat{S}}^{\dag}\overline
{U}_{\hat{S}}\Pi_{S\hat{S}}^{G}\sigma_{S\hat{S}}^{\prime}\Pi_{S\hat{S}}
^{G}]U_{S}^{\dag}\\
&  =U_{S}\operatorname{Tr}_{\hat{S}}[\Pi_{S\hat{S}}^{G}\sigma_{S\hat{S}
}^{\prime}\Pi_{S\hat{S}}^{G}]U_{S}^{\dag}\\
&  =U_{S}(g)\operatorname{Tr}_{\hat{S}}[\sigma_{S\hat{S}}^{\prime}]U_{S}
^{\dag}(g)\\
&  =U_{S}(g)\tau_{S}U_{S}^{\dag}(g).
\end{align}
where we have used the shorthand $U_{S}\equiv U_{S}(g)$ and $\overline
{U}_{\hat{S}}\equiv\overline{U}_{\hat{S}}(g)$. Since the equality $\tau
_{S}=U_{S}(g)\tau_{S}U_{S}^{\dag}(g)$ holds for all $g\in G$, it follows that
\begin{equation}\label{eq:final-eq-steps-last}
\max_{\sigma_{S\hat{S}^{\prime}}}F(\rho_{S},\operatorname{Tr}_{\hat{S}}
[\Pi_{S\hat{S}}^{G}\sigma_{S\hat{S}^{\prime}}\Pi_{S\hat{S}}^{G}])\leq
\max_{\tau_{S}\in\operatorname{Sym}_{G}}F(\rho_{S},\sigma_{S}),
\end{equation}
concluding the proof.

\subsection{Proof of Theorem~\ref{thm:G-BSE-acc-prob}}

\label{app:proof-thm-g-bse}Following the same reasoning given in \eqref{eq:max-acc-prob-proof-step-1}--\eqref{eq:1st-fid-formula-pf}, by using Uhlmann's theorem, we conclude that
\begin{multline}
\max_{V_{S^{\prime}E\rightarrow RE^{\prime}}}\left\Vert \Pi_{RS}
^{G}V_{S^{\prime}E\rightarrow RE^{\prime}}|\psi\rangle_{S^{\prime}S}
|0\rangle_{E}\right\Vert _{2}^{2}\\
=\max_{\sigma_{RS}}F(\rho_{S},\operatorname{Tr}_{R}[\Pi_{RS}^{G}\sigma_{RS} 
\Pi_{RS}^{G}]),
\end{multline}
where the optimization is over every state $\sigma_{RS}$ and $\Pi_{RS}^{G}$ is defined in \eqref{eq:Pi_RS-proj-again}. The next part of the proof shows that
\begin{equation}
\max_{\sigma_{RS}}F(\rho_{S},\operatorname{Tr}_{R}[\Pi_{RS}^{G}\sigma_{RS}
\Pi_{RS}^{G}])=\max_{\sigma_{S}\in\operatorname*{BSE}_{G}}F(\rho_{S}
,\sigma_{S})
\end{equation}
and is similar to \eqref{eq:final-eq-steps-1}--\eqref{eq:final-eq-steps-last}.
To justify the inequality$~\geq$, let $\sigma_{S}$ be an arbitrary state in
$\operatorname*{BSE}_{G}$. Then by Definition~\ref{def:g-bose-sym-ext}, this
means that there exists a state $\omega_{RS}$ such that $\operatorname{Tr}
_{R}[\omega_{RS}]=\sigma_{S}$ and $\Pi_{RS}^{G}\omega_{RS}\Pi_{RS}^{G}
=\omega_{RS}$. We find that
\begin{align}
F(\rho_{S},\sigma_{S})  &  =F(\rho_{S},\operatorname{Tr}_{R}[\omega_{RS}])\\
&  =F(\rho_{S},\operatorname{Tr}_{R}[\Pi_{RS}^{G}\omega_{RS}\Pi_{RS}^{G}])\\
&  \leq\max_{\sigma_{RS}}F(\rho_{S},\operatorname{Tr}_{R}[\Pi_{RS}^{G}
\sigma_{RS}\Pi_{RS}^{G}]),
\end{align}
which implies that
\begin{equation}
\max_{\sigma_{RS}}F(\rho_{S},\operatorname{Tr}_{R}[\Pi_{RS}^{G}\sigma_{RS}\Pi_{RS}^{G}])
\geq
\max_{\sigma_{S}\in\operatorname*{BSE}_{G}}F(\rho_{S},\sigma_{S}).
\end{equation}
To justify the inequality~$\leq$, let $\sigma_{RS}$ be an arbitrary state. If $\Pi_{RS}^{G}\sigma_{RS}\Pi_{RS}^{G}=0$, then the desired inequality trivially follows. Supposing then that this is not the case, let us define
\begin{align}
\sigma_{RS}^{\prime}  &  \coloneqq \frac{1}{p}\Pi_{RS}^{G}\sigma_{RS}\Pi
_{RS}^{G},\\
p  &  \coloneqq \operatorname{Tr}[\Pi_{RS}^{G}\sigma_{RS}].
\end{align}
We then find that
\begin{align}
& F(\rho_{S},\operatorname{Tr}_{R}[\Pi_{RS}^{G}\sigma_{RS}\Pi_{RS}^{G}]) \notag \\
&
=pF(\rho_{S},\operatorname{Tr}_{R}[\sigma_{RS}^{\prime}])\\
&  \leq F(\rho_{S},\operatorname{Tr}_{R}[\sigma_{RS}^{\prime}]).
\end{align}
Consider that $\sigma_{S}^{\prime}\coloneqq \operatorname{Tr}_{R}[\sigma
_{RS}^{\prime}]$ is $G$-Bose symmetric extendible because $\sigma_{RS}
^{\prime}$ is an extension of it that satisfies $\Pi_{RS}^{G}\sigma
_{RS}^{\prime}\Pi_{RS}^{G}=\sigma_{RS}^{\prime}$. We conclude that
\begin{equation}
F(\rho_{S},\operatorname{Tr}_{R}[\Pi_{RS}^{G}\sigma_{RS}\Pi_{RS}^{G}])\leq
\max_{\sigma_{S}\in\operatorname*{BSE}_{G}}F(\rho_{S},\sigma_{S}).
\end{equation}
Since this inequality holds for every state $\sigma_{RS}$, we surmise the
desired result
\begin{equation}
\max_{\sigma_{RS}}F(\rho_{S},\operatorname{Tr}_{R}[\Pi_{RS}^{G}\sigma_{RS}
\Pi_{RS}^{G}])\leq\max_{\sigma_{S}\in\operatorname*{BSE}_{G}}F(\rho_{S}
,\sigma_{S}).
\end{equation}

\subsection{Proof of Theorem~\ref{thm:G-SE-acc-prob}}

\label{app:proof-thm-g-se}

Following the same reasoning given in
\eqref{eq:max-acc-prob-proof-step-1}--\eqref{eq:1st-fid-formula-pf}, by using
Uhlmann's theorem, we conclude that
\begin{multline}
\max_{V_{S^{\prime}E\rightarrow R\hat{R}\hat{S}E^{\prime}}}\left\Vert
\Pi_{RS\hat{R}\hat{S}}^{G}V_{S^{\prime}E\rightarrow R\hat{R}\hat{S}E^{\prime}
}|\psi\rangle_{S^{\prime}S}|0\rangle_{E}\right\Vert _{2}^{2}\\
=\max_{\sigma_{R\hat{R}S\hat{S}}}F(\rho_{S},\operatorname{Tr}_{R\hat{R}\hat
{S}}[\Pi_{RS\hat{R}\hat{S}}^{G}\sigma_{R\hat{R}S\hat{S}}\Pi_{RS\hat{R}\hat{S}
}^{G}]),
\end{multline}
where the optimization is over every state $\sigma_{RS\hat{R}\hat{S}}$ and $\Pi_{RS\hat
{R}\hat{S}}^{G}$ is defined in \eqref{eq:projector-ref-unitaries}. The next
part of the proof shows that
\begin{multline}
\max_{\sigma_{R\hat{R}S\hat{S}}}F(\rho_{S},\operatorname{Tr}_{R\hat{R}\hat{S}
}[\Pi_{RS\hat{R}\hat{S}}^{G}\sigma_{R\hat{R}S\hat{S}}\Pi_{RS\hat{R}\hat{S}
}^{G}])\\
=\max_{\sigma_{S}\in\operatorname*{SymExt}_{G}}F(\rho_{S},\sigma_{S})
\end{multline}
and is similar to \eqref{eq:final-eq-steps-1}--\eqref{eq:final-eq-steps-last}.
To justify the inequality $\geq$, let $\sigma_{S}$ be a state in
$\operatorname*{SymExt}_{G}$. Then by Theorem~\ref{thm:Bose-sym-purify}, there
exists a purification $\varphi_{RS\hat{R}\hat{S}}$ of $\sigma_{S}$\ satisfying
$\varphi_{RS\hat{R}\hat{S}}=\Pi_{RS\hat{R}\hat{S}}^{G}\varphi_{RS\hat{R}
\hat{S}}\Pi_{RS\hat{R}\hat{S}}^{G}$. We find that
\begin{align}
&  F(\rho_{S},\sigma_{S})\nonumber\\
&  =F(\rho_{S},\operatorname{Tr}_{R\hat{R}\hat{S}}[\varphi_{RS\hat{R}\hat{S}
}])\\
&  =F(\rho_{S},\operatorname{Tr}_{R\hat{R}\hat{S}}[\Pi_{RS\hat{R}\hat{S}}
^{G}\varphi_{RS\hat{R}\hat{S}}\Pi_{RS\hat{R}\hat{S}}^{G}])\\
&  \leq\max_{\sigma_{R\hat{R}S\hat{S}}}F(\rho_{S},\operatorname{Tr}_{R\hat
{R}\hat{S}}[\Pi_{RS\hat{R}\hat{S}}^{G}\sigma_{R\hat{R}S\hat{S}}\Pi_{RS\hat
{R}\hat{S}}^{G}]).
\end{align}
Since the inequality holds for all $\sigma_{S}\in\operatorname*{SymExt}_{G}$,
we conclude that
\begin{multline}
\max_{\sigma_{S}\in\operatorname*{SymExt}_{G}}F(\rho_{S},\sigma_{S})\\
\leq\max_{\sigma_{R\hat{R}S\hat{S}}}F(\rho_{S},\operatorname{Tr}_{R\hat{R}
\hat{S}}[\Pi_{RS\hat{R}\hat{S}}^{G}\sigma_{R\hat{R}S\hat{S}}\Pi_{RS\hat{R}
\hat{S}}^{G}]).
\end{multline}
To justify the inequality $\leq$, let $\sigma_{R\hat{R}S\hat{S}}$ be an
arbitrary state. If $\Pi_{RS\hat{R}\hat{S}}^{G}\sigma_{R\hat{R}S\hat{S}}
\Pi_{RS\hat{R}\hat{S}}^{G}=0$, then the desired inequality follows trivially.
Supposing this is not the case, then define
\begin{align}
\sigma_{R\hat{R}S\hat{S}}^{\prime}  &  \coloneqq \frac{1}{p}\Pi_{RS\hat{R}
\hat{S}}^{G}\sigma_{R\hat{R}S\hat{S}}\Pi_{RS\hat{R}\hat{S}}^{G},\\
p  &  \coloneqq \operatorname{Tr}[\Pi_{RS\hat{R}\hat{S}}^{G}\sigma_{R\hat
{R}S\hat{S}}].
\end{align}
Then we find that
\begin{align}
&  F(\rho_{S},\operatorname{Tr}_{R\hat{R}\hat{S}}[\Pi_{RS\hat{R}\hat{S}}
^{G}\sigma_{R\hat{R}S\hat{S}}\Pi_{RS\hat{R}\hat{S}}^{G}])\nonumber\\
&  =pF(\rho_{S},\operatorname{Tr}_{R\hat{R}\hat{S}}[\sigma_{R\hat{R}S\hat{S}
}^{\prime}])\\
&  \leq F(\rho_{S},\operatorname{Tr}_{R\hat{R}\hat{S}}[\sigma_{R\hat{R}
S\hat{S}}^{\prime}])\\
&  =F(\rho_{S},\tau_{S}),
\end{align}
where $\tau_{S}\coloneqq \operatorname{Tr}_{R\hat{R}\hat{S}}[\sigma_{R\hat
{R}S\hat{S}}^{\prime}]$. We now aim to show that $\tau_{S}\in
\operatorname*{SymExt}_{G}$. To do so, it suffices to prove that $\sigma
_{RS}^{\prime}=U_{RS}(g)\sigma_{RS}^{\prime}U_{RS}(g)^{\dag}$ for all $g\in
G$. Abbreviating $U\otimes\overline{U}\equiv U_{RS}(g)\otimes\overline
{U}_{\hat{R}\hat{S}}(g)$, consider that
\begin{align}
&  \sigma_{RS}^{\prime}\nonumber\\
&  =\operatorname{Tr}_{\hat{R}\hat{S}}[\sigma_{RS\hat{R}\hat{S}}^{\prime}]\\
&  =\operatorname{Tr}_{\hat{R}\hat{S}}[\Pi_{RS\hat{R}\hat{S}}^{G}
\sigma_{RS\hat{R}\hat{S}}^{\prime}\Pi_{RS\hat{R}\hat{S}}^{G}]\\
&  =\operatorname{Tr}_{\hat{R}\hat{S}}[(U\otimes\overline{U})\Pi_{RS\hat
{R}\hat{S}}^{G}\sigma_{RS\hat{R}\hat{S}}^{\prime}\Pi_{RS\hat{R}\hat{S}}
^{G}(U\otimes\overline{U})^{\dag}]\\
&  =U\operatorname{Tr}_{\hat{R}\hat{S}}[\overline{U}\Pi_{RS\hat{R}\hat{S}}
^{G}\sigma_{RS\hat{R}\hat{S}}^{\prime}\Pi_{RS\hat{R}\hat{S}}^{G}\overline
{U}^{\dag}]U^{\dag}\\
&  =U\operatorname{Tr}_{\hat{R}\hat{S}}[\overline{U}^{\dag}\overline{U}
\Pi_{RS\hat{R}\hat{S}}^{G}\sigma_{RS\hat{R}\hat{S}}^{\prime}\Pi_{RS\hat{R}
\hat{S}}^{G}]U^{\dag}\\
&  =U\operatorname{Tr}_{\hat{R}\hat{S}}[\Pi_{RS\hat{R}\hat{S}}^{G}
\sigma_{RS\hat{R}\hat{S}}^{\prime}\Pi_{RS\hat{R}\hat{S}}^{G}]U^{\dag}\\
&  =U\operatorname{Tr}_{\hat{R}\hat{S}}[\sigma_{RS\hat{R}\hat{S}}^{\prime
}]U^{\dag}\\
&  =U_{RS}(g)\sigma_{RS}^{\prime}U_{RS}(g)^{\dag}.
\end{align}
It follows that $\tau_{S}\in\operatorname*{SymExt}_{G}$, and we conclude
that
\begin{multline}
F(\rho_{S},\operatorname{Tr}_{R\hat{R}\hat{S}}[\Pi_{RS\hat{R}\hat{S}}
^{G}\sigma_{R\hat{R}S\hat{S}}\Pi_{RS\hat{R}\hat{S}}^{G}])\\
\leq\max_{\sigma_{S}\in\operatorname*{SymExt}_{G}}F(\rho_{S},\sigma_{S}).
\end{multline}
Since the inequality holds for every state $\sigma_{R\hat{R}S\hat{S}}$, we
conclude that
\begin{multline}
\max_{\sigma_{R\hat{R}S\hat{S}}}F(\rho_{S},\operatorname{Tr}_{R\hat{R}\hat{S}
}[\Pi_{RS\hat{R}\hat{S}}^{G}\sigma_{R\hat{R}S\hat{S}}\Pi_{RS\hat{R}\hat{S}
}^{G}])\\
\leq\max_{\sigma_{S}\in\operatorname*{SymExt}_{G}}F(\rho_{S},\sigma_{S}).
\end{multline}

\section{Proof of Proposition~\ref{prop:golden-rule-G-BSE}}
\label{app:proof-prop-G-BSE}

The idea of the proof is similar to that for
Proposition~\ref{prop:golden-rule-G-SE}. Since $\rho_{S}$ is a $G$-BSE state,
by Definition~\ref{def:g-bose-sym-ext}, there exists an extension state
$\omega_{RS}$ satisfying the conditions stated there. Since $\mathcal{N}
_{S\rightarrow S^{\prime}}$ is a $G$-BSE channel, by
Definition~\ref{def:G-BSE-channels}, there exists an extension channel
$\mathcal{M}_{RS\rightarrow R^{\prime}S^{\prime}}$ satisfying the conditions
stated there. It follows that $\mathcal{M}_{RS\rightarrow R^{\prime}S^{\prime
}}(\omega_{RS})$ is an extension of $\mathcal{N}_{S\rightarrow S^{\prime}
}(\rho_{S})$ because
\begin{align}
\operatorname{Tr}_{R^{\prime}}[\mathcal{M}_{RS\rightarrow R^{\prime}S^{\prime
}}(\omega_{RS})]  &  =\mathcal{N}_{S\rightarrow S^{\prime}}(\operatorname{Tr}
_{R}[\omega_{RS}])\\
&  =\mathcal{N}_{S\rightarrow S^{\prime}}(\rho_{S}),
\end{align}
where the first equality follows from \eqref{eq:channel-ext-bose}. Also,
consider that the following holds
\begin{align}
1  &  \geq\operatorname{Tr}[\Pi_{R^{\prime}S^{\prime}}^{G}\mathcal{M}
_{RS\rightarrow R^{\prime}S^{\prime}}(\omega_{RS})]\nonumber\\
&  =\operatorname{Tr}[(\mathcal{M}_{RS\rightarrow R^{\prime}S^{\prime}}
)^{\dag}(\Pi_{R^{\prime}S^{\prime}}^{G})\omega_{RS}]\\
&  \geq\operatorname{Tr}[\Pi_{RS}^{G}\omega_{RS}]\\
&  =1.
\end{align}
The first inequality follows because $\mathcal{M}_{RS\rightarrow R^{\prime
}S^{\prime}}(\omega_{RS})$ is a state and $\Pi_{R^{\prime}S^{\prime}}^{G}$ is
projection. The first equality follows from the definition of channel adjoint.
The second inequality follows from \eqref{eq:BSE-bose-condition}. We conclude that $\operatorname{Tr}[\Pi_{R^{\prime
}S^{\prime}}^{G}\mathcal{M}_{RS\rightarrow R^{\prime}S^{\prime}}(\omega
_{RS})]=1$, which by \eqref{eq:Bose-symmetric-equiv-cond}, implies that
$\mathcal{M}_{RS\rightarrow R^{\prime}S^{\prime}}(\omega_{RS})$ is a $G$-Bose
symmetric state. It then follows that $\mathcal{N}_{S\rightarrow S^{\prime}
}(\rho_{S})$ is $G$-Bose symmetric extendible.

\section{Cyclic group \texorpdfstring{$C_3$}{C3}}

\label{app:cyclic-c-3}

\begin{figure*}[t!]
\begin{center}
\includegraphics[width=\linewidth]{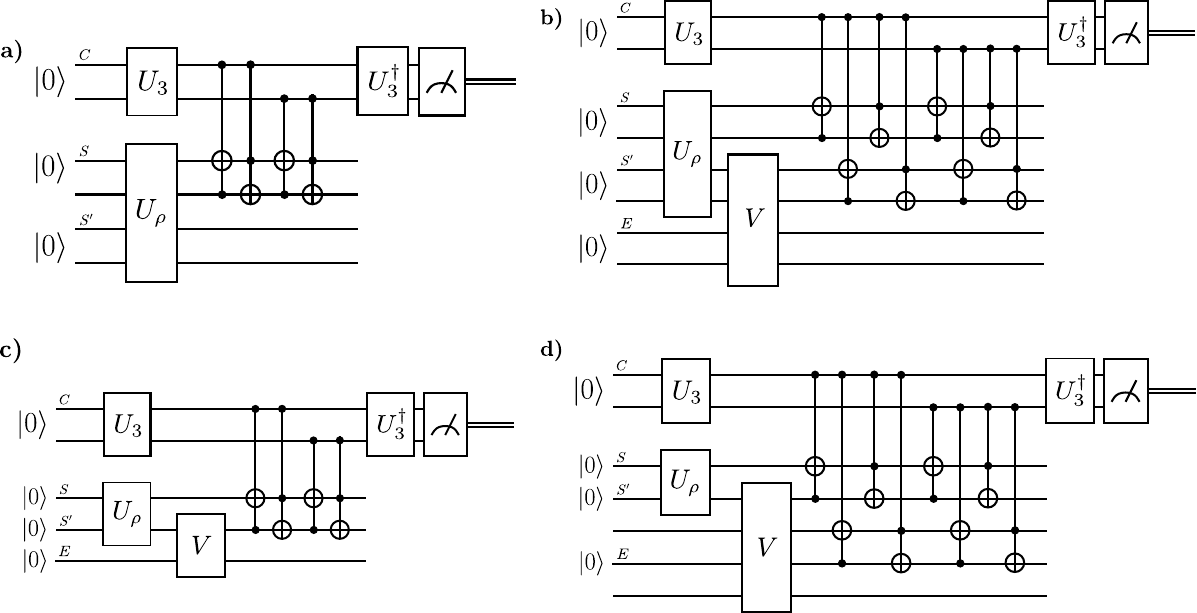}
\end{center}
\caption{Symmetry tests for the $C_3$ group: a) $G$-Bose symmetry, b) $G$-symmetry, c) $G$-Bose symmetric extendibility, and d) $G$-symmetric extendibility. }
\label{fig:CS3_Circuits}
\end{figure*}

Cyclic groups, denoted by $C_n$, are abelian groups formed by cyclic shifts of $n$ elements and always have order $n$. Consider first $C_3$, the cyclic group on three elements. The group table for $C_3$ is given by
\begin{center}
\begin{tabular}{
>{\centering\arraybackslash}p{0.15\textwidth} | >{\centering\arraybackslash}p{0.05\textwidth} | 
>{\centering\arraybackslash}p{0.05\textwidth} | 
>{\centering\arraybackslash}p{0.05\textwidth}}
 \hline
  Group element &$e$& $a$ & $b$ \\ 
  \hline
$e$ & $e$ & $a$ & $b$ \\ 
 $a$ & $a$ & $b$ &$e$\\ 
 $b$ & $b$ & $e$ & $a$ \\ 
 \hline
\end{tabular}
\end{center}

The $C_3$ group has a one-dimensional representation given by the third roots of unity, but here we instead opt for a two-qubit unitary representation corresponding more closely to the standard representation of $C_3$: $\{e \rightarrow \mathbb{I}, a \rightarrow \textrm{SWAP}\ \circ\ \textrm{CNOT}, b \rightarrow \textrm{SWAP}\ \circ\ \textrm{CNOT}\ \circ\ \textrm{SWAP}\ \circ\ \textrm{CNOT}\}$. The $C_3$ group has three elements, and thus, the $\ket{+}_C$ state is a uniform superposition of three elements. We use two qubits and the same unitary~$U_3$ shown in Figure~\ref{fig:U3_Superposition} to generate an equal superposition of three elements:
\begin{equation}
    \label{eq:CS3_Superposition}
    U_3\ket{00} = \frac{1}{\sqrt{3}} (\ket{00} + \ket{01} + \ket{11}).
\end{equation}

The control register states need to be mapped to group elements. We employ the mapping $\{\ket{00} \rightarrow e, \ket{01} \rightarrow a, \ket{11} \rightarrow b\}$ for our circuit constructions. The circuits required for all tests are given in Figure~\ref{fig:CS3_Circuits}.

\subsection{\texorpdfstring{$G$}{G}-Bose symmetry}

Figure~\ref{fig:CS3_Circuits}a) shows the circuit that tests for $G$-Bose symmetry. Table~\ref{tab:CS3_GBS} shows the results for various input states. The true fidelity value is calculated using \eqref{eq:acc-prob-bose-test}, where $\Pi^G_S$ is defined in \eqref{eq:group_proj_GBS}.

\begin{table}[h]
\centering
\vspace{.05in}
\begin{tabular}{
>{\centering\arraybackslash}p{0.12\textwidth} | >{\centering\arraybackslash}p{0.08\textwidth} | 
>{\centering\arraybackslash}p{0.08\textwidth} | 
>{\centering\arraybackslash}p{0.08\textwidth}}
\hline
\textrm{State} & 
\textrm{True Fidelity} &
\textrm{Noiseless} &
\textrm{Noisy}\\
\hline\hline 
$\outerproj{00}$ & 1.0 & 1.0000 & 0.8415 \\
$\ket{\!-\!+}\!\bra{-\!+\!}$ & 0.3333 & 0.3333 & 0.3408\\ 
$\rho$ & 1.0  & 1.0000 & 0.8524\\ 
$\pi^{\otimes 2}$ & 0.5 & 0.5000 & 0.4698\\ 
\hline
\end{tabular}
\caption{Results of $C_3$-Bose symmetry tests. The state $\rho$ is defined as $\outerproj{\psi}$ where $\ket{\psi} = \frac{1}{\sqrt{3}}(\ket{01} + \ket{10} + \ket{11})$.}
\label{tab:CS3_GBS}
\end{table}

\subsection{\texorpdfstring{$G$}{G}-symmetry}
A circuit that tests for $G$-symmetry is shown in Figure~\ref{fig:CS3_Circuits}b). It involves variational parameters, and an example of the training process is shown in Figure~\ref{fig:CS3_GS_Training}. Table~\ref{tab:CS3_GS} shows the final results after training for various input states. The true fidelity is calculated using the semi-definite program given in \eqref{eq:SDP-rootfid-GS}.

\begin{figure}
\begin{center}
\includegraphics[width=\linewidth]{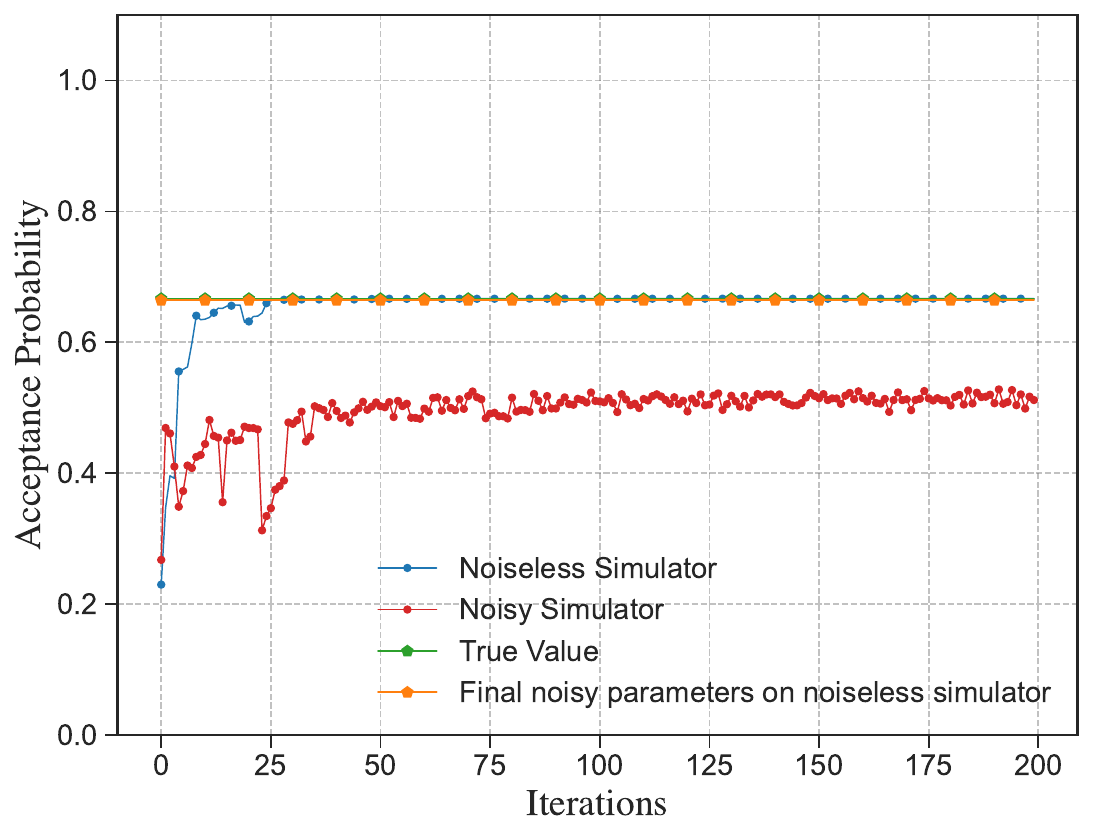}
\end{center}
\caption{Example of the training process for testing $C_3$-symmetry of $\Phi^+$. We see that the training exhibits a noise resilience.}
\label{fig:CS3_GS_Training}
\end{figure}

\begin{table}[h]
\centering
\vspace{.05in}
\begin{tabular}{
>{\centering\arraybackslash}p{0.09\textwidth} | >{\centering\arraybackslash}p{0.07\textwidth} | 
>{\centering\arraybackslash}p{0.07\textwidth} | 
>{\centering\arraybackslash}p{0.05\textwidth} |
>{\centering\arraybackslash}p{0.08\textwidth}}
\hline
\textrm{State} & 
\textrm{True Fidelity} &
\textrm{Noiseless} &
\textrm{Noisy} & 
\textrm{Noise Resilient}
\\
\hline\hline 
$\ket{-,\!+}\!\bra{-,\!+}$& 0.3339 & 0.3333 & 0.3084 & 0.3333 \\
$\Phi^+$ & 0.6666 & 0.6666 & 0.5118 & 0.6639 \\
$\rho$ & 0.7778 & 0.7775 & 0.5694 & 0.7760 \\ 
$\pi^{\otimes 2} $ & 1.0000 & 0.9998 & 0.6756 & 0.9864 \\ 
\hline
\end{tabular}
\caption{Results of $C_3$-symmetry tests. The state $\rho$ is defined as $\outerproj{\psi}$ where $\ket{\psi} = \frac{1}{\sqrt{3}}(\ket{00}+\ket{11}+\ket{10})$.}
\label{tab:CS3_GS}
\end{table}

\subsection{\texorpdfstring{$G$}{G}-Bose symmetric extendibility}

A circuit that tests for $G$-Bose symmetric extendibility is shown in Figure~\ref{fig:CS3_Circuits}c). It involves variational parameters, and an example of the training process is shown in Figure~\ref{fig:CS3_GBSE_Training}. Table~\ref{tab:CS3_GBSE} shows the final results after training for various input states. The true fidelity is calculated using the semi-definite program given in \eqref{eq:SDP-rootfid-GBSE}.

\begin{figure}
\begin{center}
\includegraphics[width=\linewidth]{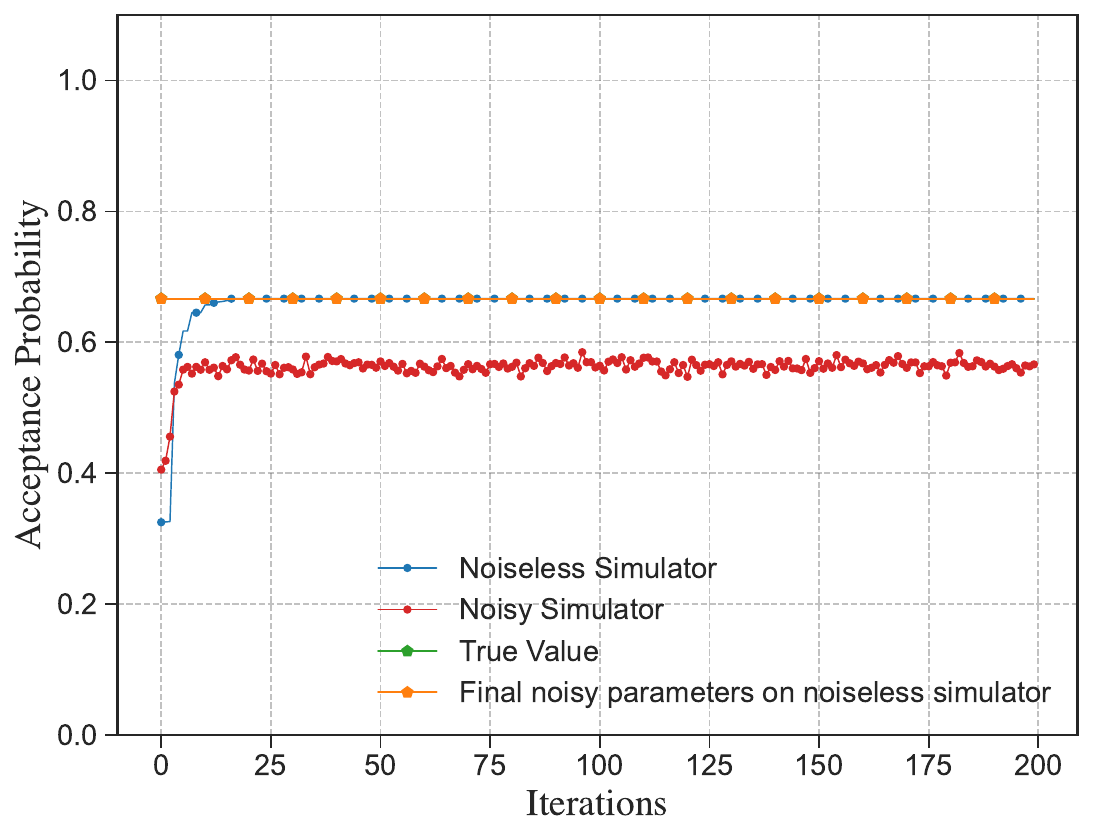}
\end{center}
\caption{Example of the training process for testing $C_3$-Bose symmetric extendibility of $\outerproj{1}$. We see that the training exhibits a noise resilience.}
\label{fig:CS3_GBSE_Training}
\end{figure}

\begin{table}[h]
\centering
\vspace{.05in}
\begin{tabular}{
>{\centering\arraybackslash}p{0.05\textwidth} | >{\centering\arraybackslash}p{0.07\textwidth} | 
>{\centering\arraybackslash}p{0.07\textwidth} | 
>{\centering\arraybackslash}p{0.07\textwidth} |
>{\centering\arraybackslash}p{0.08\textwidth}}
\hline
\textrm{State} & 
\textrm{True Fidelity} &
\textrm{Noiseless} &
\textrm{Noisy} & 
\textrm{Noise Resilient}
\\
\hline\hline 
$\outerproj{0}$ & 0.6670 & 0.6667 & 0.5662 & 0.6665 \\
$\pi$ & 1.0000 & 1.0000 & 0.8066 & 0.9979 \\
$\rho$ & 0.8382 & 0.8380 & 0.7093 & 0.8377 \\
\hline
\end{tabular}
\caption{Results of $C_3$-Bose symmetric extendibility tests. The state $\rho$ is defined as $\outerproj{\psi}$ where $\ket{\psi} = \frac{1}{2}(\sqrt{3}\ket{0} - \ket{1})$.}
\label{tab:CS3_GBSE}
\end{table}

\subsection{\texorpdfstring{$G$}{G}-symmetric extendibility}

A circuit that tests for $G$-symmetric extendibility is shown in Figure~\ref{fig:CS3_Circuits}d). It involves variational parameters, and an example of the training process is shown in Figure~\ref{fig:CS3_GSE_Training}. Table~\ref{tab:CS3_GSE} shows the final results after training for various input states. The true fidelity is calculated using the semi-definite program given in \eqref{eq:SDP-rootfid-GSE}.

\begin{figure}
\begin{center}
\includegraphics[width=\linewidth]{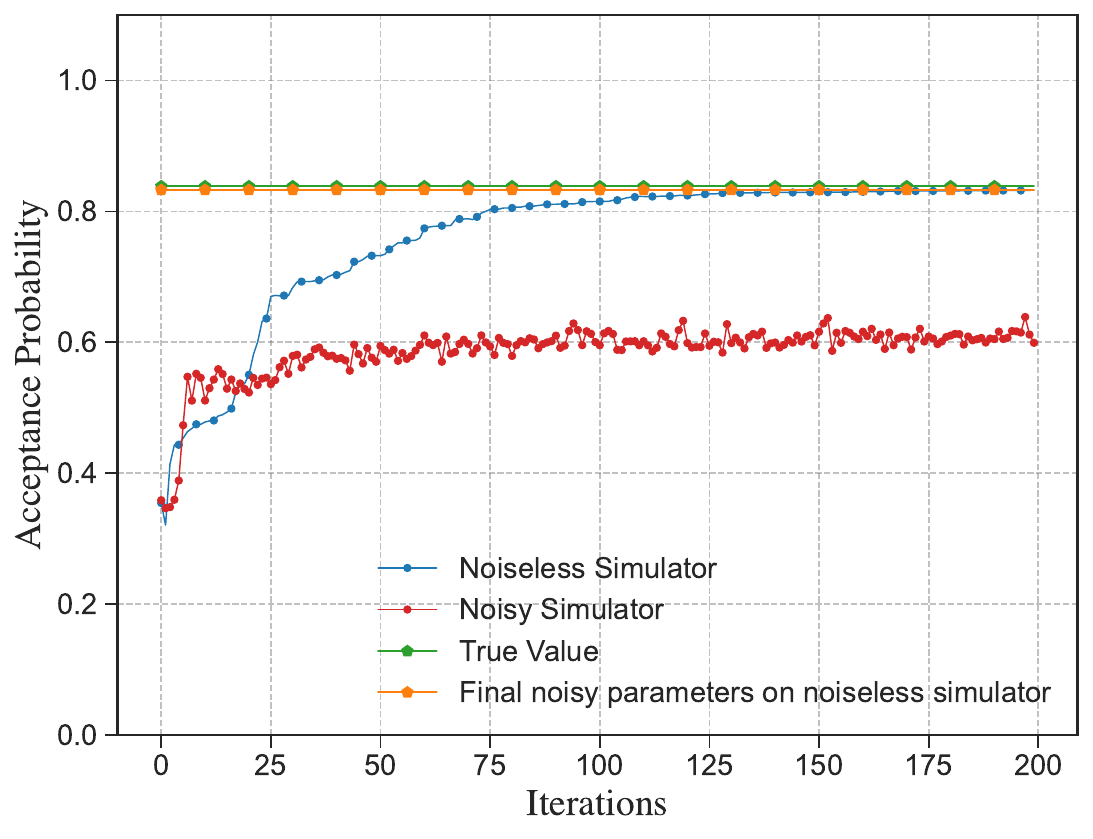}
\end{center}
\caption{Example of the training process for testing $C_3$-symmetric extendibility of $\rho = \outerproj{\psi}$, where $\ket{\psi} = \frac{1}{2}(\sqrt{3}\ket{0} - \ket{1})$. We see that the training exhibits a noise resilience.}
\label{fig:CS3_GSE_Training}
\end{figure}

\begin{table}[h]
\centering
\vspace{.05in}
\begin{tabular}{
>{\centering\arraybackslash}p{0.05\textwidth} | >{\centering\arraybackslash}p{0.07\textwidth} | 
>{\centering\arraybackslash}p{0.07\textwidth} | 
>{\centering\arraybackslash}p{0.07\textwidth} |
>{\centering\arraybackslash}p{0.08\textwidth}}
\hline
\textrm{State} & 
\textrm{True Fidelity} &
\textrm{Noiseless} &
\textrm{Noisy} & 
\textrm{Noise Resilient}
\\
\hline\hline 
$\outerproj{1}$ & 0.6667 & 0.6660 & 0.4809 & 0.6620 \\
$\pi$ & 1.0000 & 0.9942 & 0.6818 & 0.9812 \\
$\rho$ & 0.8383 & 0.8322 & 0.5992 & 0.8327 \\
\hline
\end{tabular}
\caption{Results of $C_3$-symmetric extendibility tests. The state $\rho$ is defined as $\outerproj{\psi}$ where $\ket{\psi} = \frac{1}{2}(\sqrt{3}\ket{0} - \ket{1})$.}
\label{tab:CS3_GSE}
\end{table}

\section{Cyclic group \texorpdfstring{$C_4$}{C4}}
\label{sec:CS4}

In this appendix, we consider $C_4$, the cyclic group on four elements. Again, as an abelian group, there exists a one-dimensional representation that we choose not to employ here. Instead, we consider again a two-qubit representation.

\begin{figure*}[t!]
\begin{center}
\includegraphics[width=\linewidth]{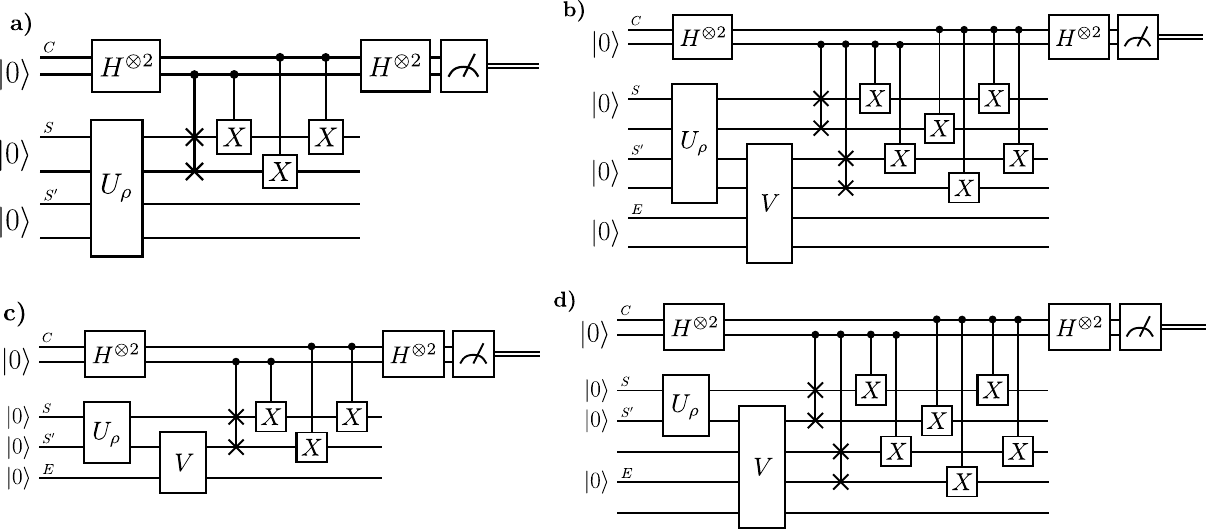}
\end{center}
\caption{Symmetry tests for the $C_4$ group: a) $G$-Bose symmetry, b) $G$-symmetry, c) $G$-Bose symmetric extendibility, and d) $G$-symmetric extendibility. }
\label{fig:CS4_Circuits}
\end{figure*}

The group table for $C_4$ is given by
\begin{center}
\begin{tabular}{
>{\centering\arraybackslash}p{0.14\textwidth} | 
>{\centering\arraybackslash}p{0.03\textwidth} | 
>{\centering\arraybackslash}p{0.03\textwidth} |
>{\centering\arraybackslash}p{0.03\textwidth} |
>{\centering\arraybackslash}p{0.03\textwidth}}
 \hline
  Group element & $e$ & $a$ & $b$ & $c$\\ 
  \hline
 $e$ & $e$ & $a$ & $b$ & $c$\\ 
 $a$ & $a$ & $b$ & $c$ & $e$ \\ 
 $b$ & $b$ & $c$ & $e$ & $a$ \\
 $c$ & $c$ & $e$ & $a$ & $b$ \\
 \hline
\end{tabular}
\end{center}
This group has a two-qubit unitary representation $\{e \rightarrow \mathbb{I}, a \rightarrow X_0 \circ \textrm{SWAP}, b \rightarrow X_0X_1, c \rightarrow X_1 \circ \textrm{SWAP}\}$, where $X_i$ denotes the Pauli $\sigma_x$ operator acting on qubit $i$, for $i \in \{0,1\}$. The $C_4$ group has four elements, and thus, the $\ket{+}_C$ state is a uniform superposition of four elements. We use two qubits and the Hadamard gate to generate the control state, as follows:
\begin{equation}
    H^{\otimes 2}\ket{00} = \frac{1}{2}\left(\ket{00} + \ket{01} + \ket{10} + \ket{11}\right).
\end{equation}
The control register states need to be mapped to group elements. We employ the mapping $\{\ket{00} \rightarrow e, \ket{01} \rightarrow a, \ket{10} \rightarrow b, \ket{11} \rightarrow c\}$ for our circuit constructions.

\subsection{\texorpdfstring{$G$}{G}-Bose symmetry}

Figure~\ref{fig:CS4_Circuits}a) shows a circuit that tests for $G$-Bose symmetry. Table~\ref{tab:CS4_GBS} shows the results for various input states. The true fidelity value is calculated using \eqref{eq:acc-prob-bose-test}, where $\Pi^G_S$ is defined in \eqref{eq:group_proj_GBS}.

\begin{table}[h]
\vspace{.05in}
\centering
\begin{tabular}{
>{\centering\arraybackslash}p{0.10\textwidth} | >{\centering\arraybackslash}p{0.08\textwidth} | 
>{\centering\arraybackslash}p{0.08\textwidth} | 
>{\centering\arraybackslash}p{0.08\textwidth}}
\hline
\textrm{State} & 
\textrm{True Fidelity} &
\textrm{Noiseless} &
\textrm{Noisy}\\
\hline\hline 
$\outerproj{00}$ & 0.25 & 0.2500 & 0.2579 \\
$\ket{\!+\!+}\!\bra{+\!+\!}$ & 1.0   & 1.0000 & 0.9276 \\ 
$\ket{\!+\!0}\!\bra{+0}$ & 0.5 & 0.5000 & 0.5002 \\ 
$\pi^{\otimes 2} $ & 0.25 & 0.2500 & 0.2449\\ 
\hline
\end{tabular}
\caption{Results of $C_4$-Bose symmetry tests.}
\label{tab:CS4_GBS}
\end{table}

\subsection{\texorpdfstring{$G$}{G}-symmetry}
A circuit that tests for $G$-symmetry is shown in Figure~\ref{fig:CS4_Circuits}b). It involves variational parameters, and an example of the training process is shown in Figure~\ref{fig:CS4_GS_Training}. Table~\ref{tab:CS4_GS} shows the final results after training for various input states. The true fidelity is calculated using the semi-definite program given in \eqref{eq:SDP-rootfid-GS}.

\begin{figure}
\begin{center}
\includegraphics[width=\linewidth]{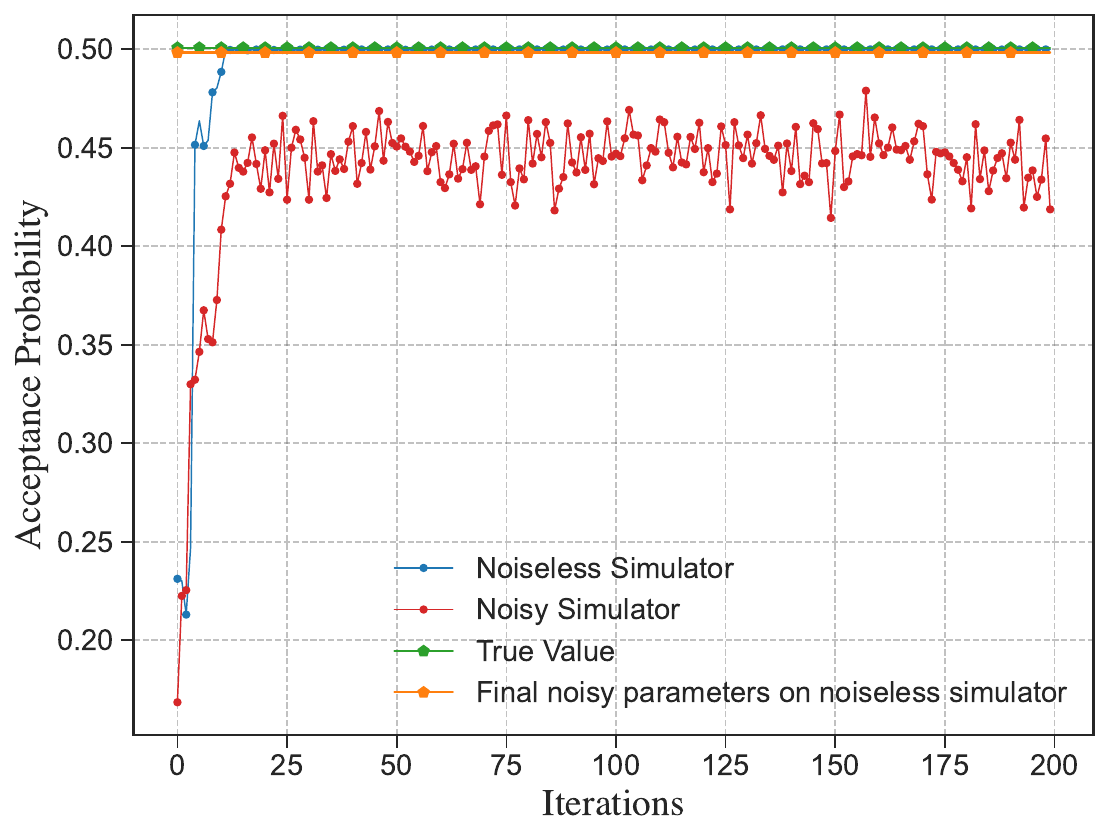}
\end{center}
\caption{Example of the training process for testing $C_4$-symmetry of $\rho = \outerproj{\psi}$, where $\ket{\psi} = \ket{\!+\!-}$. We see that the training exhibits a noise resilience.}
\label{fig:CS4_GS_Training}
\end{figure}

\begin{table}[h]
\vspace{.05in}
\centering
\begin{tabular}{
>{\centering\arraybackslash}p{0.08\textwidth} | >{\centering\arraybackslash}p{0.07\textwidth} | 
>{\centering\arraybackslash}p{0.07\textwidth} | 
>{\centering\arraybackslash}p{0.06\textwidth} |
>{\centering\arraybackslash}p{0.08\textwidth}}
\hline
\textrm{State} & 
\textrm{True Fidelity} &
\textrm{Noiseless} &
\textrm{Noisy} & 
\textrm{Noise Resilient}
\\
\hline\hline 
$\outerproj{00}$ & 0.2502 & 0.2500 & 0.2562 & 0.2500 \\
$\ket{\!+\!-}\!\bra{-\!+\!}$ & 0.5008 & 0.5000 & 0.4187 & 0.4984 \\
$\pi \otimes \outerproj{0}$ & 0.7501 & 0.7498 & 0.6140 & 0.7480 \\
$\pi^{\otimes 2}$ & 1.0000 & 0.9992 & 0.7606 & 0.9912 \\
\hline
\end{tabular}
\caption{Results of $C_4$-symmetry tests.}
\label{tab:CS4_GS}
\end{table}

\subsection{\texorpdfstring{$G$}{G}-Bose symmetric extendibility}

A circuit that tests for $G$-Bose symmetric extendibility is shown in Figure~\ref{fig:CS4_Circuits}c). It involves variational parameters, and an example of the training process is shown in Figure~\ref{fig:CS4_GBSE_Training}. Table~\ref{tab:CS4_GBSE} shows the final results after training for various input states. The true fidelity is calculated using the semi-definite program given in \eqref{eq:SDP-rootfid-GBSE}.

\begin{figure}
\begin{center}
\includegraphics[width=\linewidth]{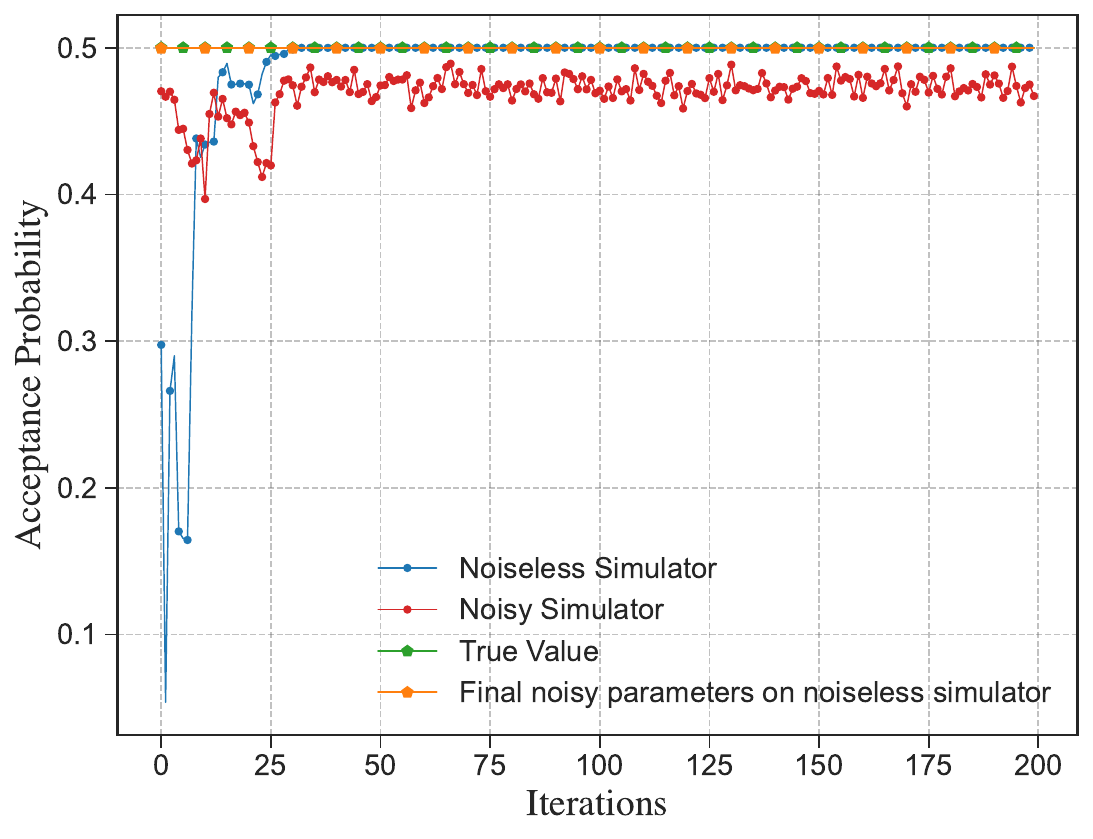}
\end{center}
\caption{Example of the training process for testing $C_4$-Bose symmetric extendibility of $\outerproj{00}$. We see that the training exhibits a noise resilience.}
\label{fig:CS4_GBSE_Training}
\end{figure}

\begin{table}[h]
\centering
\vspace{.05in}
\begin{tabular}{
>{\centering\arraybackslash}p{0.05\textwidth} | >{\centering\arraybackslash}p{0.08\textwidth} | 
>{\centering\arraybackslash}p{0.08\textwidth} | 
>{\centering\arraybackslash}p{0.06\textwidth} |
>{\centering\arraybackslash}p{0.08\textwidth}}
\hline
\textrm{State} & 
\textrm{True Fidelity} &
\textrm{Noiseless} &
\textrm{Noisy} & 
\textrm{Noise Resilient}
\\
\hline\hline 
$\outerproj{0}$ & 0.5000 & 0.5000 & 0.4671 & 0.4995 \\
$\outerproj{+}$ & 1.0000 & 1.0000 & 0.9195 & 1.0000 \\
$\rho$ & 0.9330 & 0.9330 & 0.8689 & 0.9329 \\
\hline
\end{tabular}
\caption{Results of $C_4$-Bose symmetric extendibility tests. The state $\rho$ is defined as $\begin{bmatrix} 0.75 & 0.4330\\ 0.4430 & 0.25 \end{bmatrix}$.}
\label{tab:CS4_GBSE}
\end{table}

\subsection{\texorpdfstring{$G$}{G}-symmetric extendibility}

A circuit that tests for $G$-symmetric extendibility is shown in Figure~\ref{fig:CS4_Circuits}d). It involves variational parameters, and an example of the training process is shown in Figure~\ref{fig:CS4_GSE_Training}. Table~\ref{tab:CS4_GSE} shows the final results after training for various input states. The true fidelity is calculated using the semi-definite program given in \eqref{eq:SDP-rootfid-GSE}.

\begin{figure}
\begin{center}
\includegraphics[width=\linewidth]{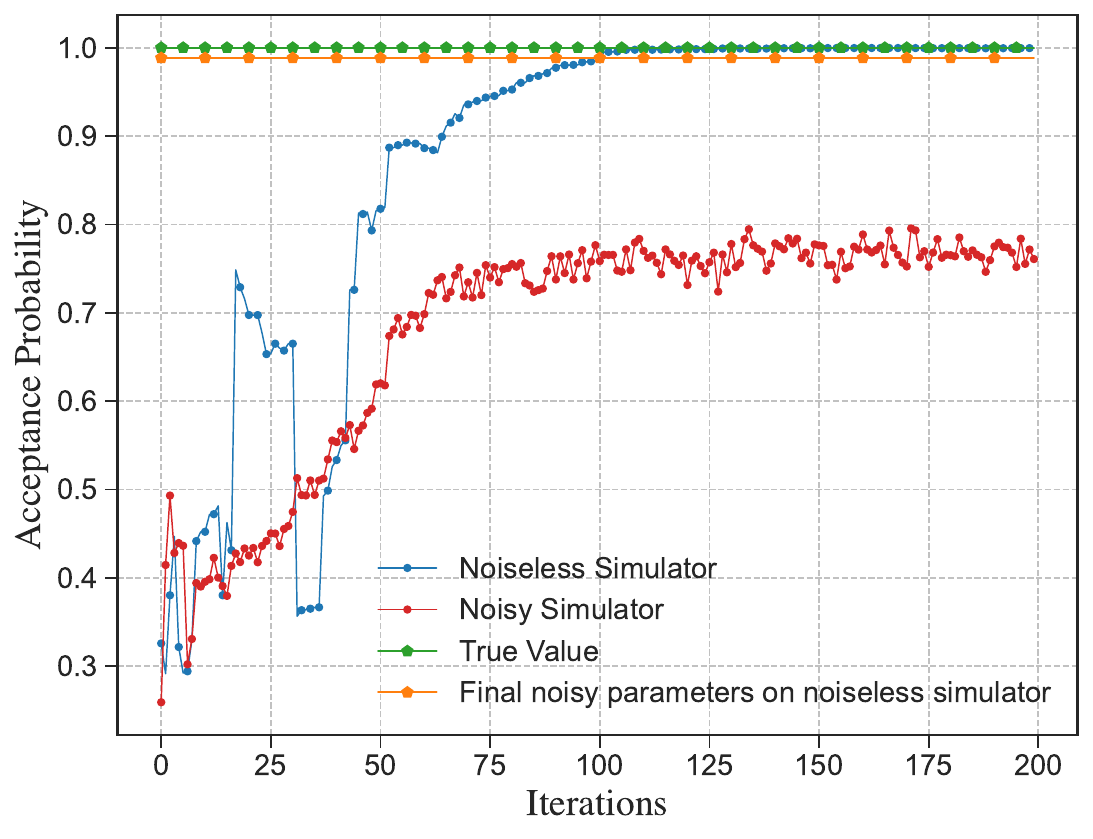}
\end{center}
\caption{Example of the training process for testing $C_4$-symmetry extendibility of $\pi$. We see that the training exhibits a noise resilience.}
\label{fig:CS4_GSE_Training}
\end{figure}

\begin{table}[h]
\centering
\vspace{.05in}
\begin{tabular}{
>{\centering\arraybackslash}p{0.06\textwidth} | >{\centering\arraybackslash}p{0.08\textwidth} | 
>{\centering\arraybackslash}p{0.07\textwidth} | 
>{\centering\arraybackslash}p{0.06\textwidth} |
>{\centering\arraybackslash}p{0.08\textwidth}}
\hline
\textrm{State} & 
\textrm{True Fidelity} &
\textrm{Noiseless} &
\textrm{Noisy} & 
\textrm{Noise Resilient}
\\
\hline\hline 
$\outerproj{0}$ & 0.5000 & 0.4997 & 0.4191 & 0.4982 \\
$\pi$ & 1.0000 & 0.9996 & 0.7608 & 0.9884 \\
$\rho$ & 0.8535 & 0.8533 & 0.6838 & 0.8459 \\
\hline
\end{tabular}
\caption{Results of $C_4$-symmetric extendibility tests. The state $\rho$ is defined as $\begin{bmatrix} 0.854 & 0\\ 0 & 0.146 \end{bmatrix}$.}
\label{tab:CS4_GSE}
\end{table}

\section{Quaternion group \texorpdfstring{$Q_8$}{Q8}}

\label{app:Q8}

\begin{figure*}[t!]
\begin{center}
\includegraphics[width=\linewidth]{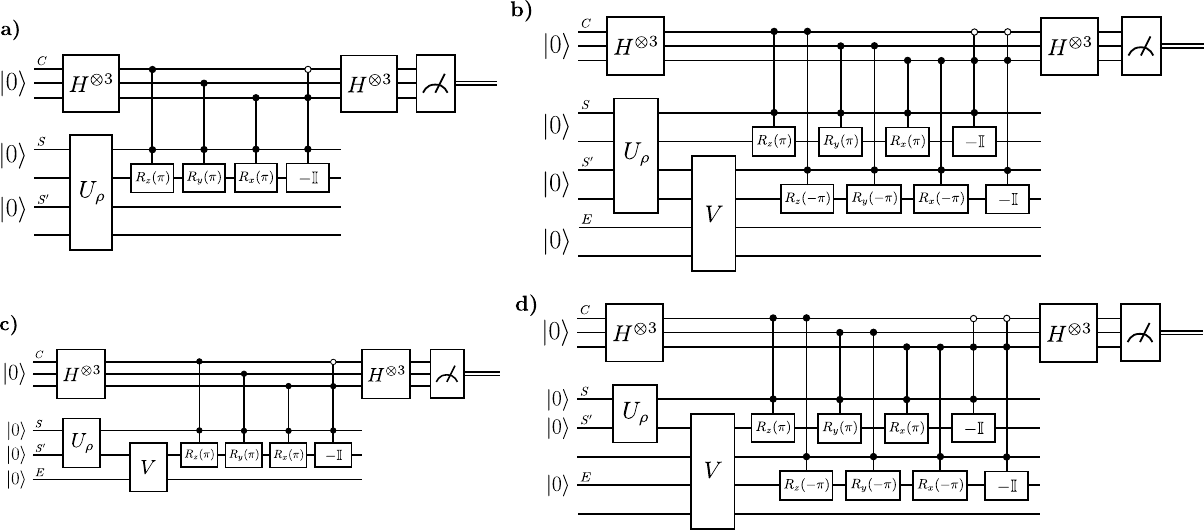}
\end{center}
\caption{Symmetry tests for the $Q_8$ group: a) $G$-Bose symmetry, b) $G$-symmetry, c) $G$-Bose symmetric extendibility, and d) $G$-symmetric extendibility. }
\label{fig:Q8_Circuits}
\end{figure*}

The Quaternion group is defined as
\begin{equation}
    Q_8 = \langle \bar{e}, i, j, k\ \vert\ \bar{e}^2 = e, i^2 = j^2 = k^2 = ijk = \bar{e} \rangle.
\end{equation}
The inverse elements of $e, i, j, k$ are given by $\bar{e}, \bar{i}, \bar{j}, \bar{k}$ respectively. The $Q_8$ group has a two-qubit unitary representation 
\begin{eqnarray}
    e = \begin{bmatrix} \mathbb{I} & 0\\ 0 & \mathbb{I} \end{bmatrix}
    &,\quad \bar{e} = \begin{bmatrix} \mathbb{I} & 0\\ 0 & -\mathbb{I} \end{bmatrix}, \nonumber \\
    i = \begin{bmatrix} \mathbb{I} & 0\\ 0 & -i\sigma_x \end{bmatrix}&,\quad \bar{i} = \begin{bmatrix} \mathbb{I} & 0\\ 0 & i\sigma_x \end{bmatrix}, \nonumber \\
    j = \begin{bmatrix} \mathbb{I} & 0\\ 0 & -i\sigma_y \end{bmatrix}&,\quad \bar{j} = \begin{bmatrix} \mathbb{I} & 0\\ 0 & i\sigma_y \end{bmatrix}, \nonumber \\
    k = \begin{bmatrix} \mathbb{I} & 0\\ 0 & -i\sigma_z \end{bmatrix}&,\quad \bar{k} = \begin{bmatrix} \mathbb{I} & 0\\ 0 & i\sigma_z \end{bmatrix}.
\end{eqnarray}

The $Q_8$ group has eight elements and thus, the $\ket{+}_C$ state is a uniform superposition of eight elements. We use three qubits and the Hadamard gate to generate it as follows:
\begin{multline}
    H^{\otimes 3}\ket{000} = \frac{1}{\sqrt{8}} \Big(\ket{000} + \ket{001} + \ket{010} +\\ \ket{011} + \ket{100} + \ket{101} + \ket{110} + \ket{111}\Big).
\end{multline}
The control register states need to be mapped to group elements. We employ the mapping $\{\ket{000}\rightarrow e, \ket{001} \rightarrow \bar{i}, \ket{010} \rightarrow j, \ket{011} \rightarrow \bar{k},\ket{100} \rightarrow k, \ket{101} \rightarrow \bar{j}, \ket{110} \rightarrow i, \ket{111} \rightarrow \bar{e}\}$ for our circuit constructions.

\subsection{\texorpdfstring{$G$}{G}-Bose symmetry}

Figure~\ref{fig:Q8_Circuits}a) shows the circuit needed to test for $G$-Bose symmetry. Table~\ref{tab:Q8_GBS} shows the results for various input states. The true fidelity value is calculated using \eqref{eq:acc-prob-bose-test}, where $\Pi^G_S$ is defined in \eqref{eq:group_proj_GBS}.

\begin{table}[h]
\centering
\vspace{.05in}
\begin{tabular}{
>{\centering\arraybackslash}p{0.08\textwidth} | >{\centering\arraybackslash}p{0.08\textwidth} | 
>{\centering\arraybackslash}p{0.08\textwidth} | 
>{\centering\arraybackslash}p{0.08\textwidth}}
\hline
\textrm{State} & 
\textrm{True Fidelity} &
\textrm{Noiseless} &
\textrm{Noisy}\\
\hline\hline 
$\outerproj{00}$ & 1.0 & 1.0000 & 0.7416 \\
$\ket{1+}\!\bra{1\!+\!}$ & 0.0 & 0.0000 & 0.0709\\ 
$\ket{\!+\!0}\!\bra{0\!+\!}$ & 0.5 & 0.4999 & 0.3961\\ 
$\pi^{\otimes 2}$ & 0.5 & 0.4999 & 0.3842\\ 
\hline
\end{tabular}
\caption{Results of $Q_8$-Bose symmetry tests.}
\label{tab:Q8_GBS}
\end{table}

\subsection{\texorpdfstring{$G$}{G}-symmetry}
A circuit that tests for $G$-symmetry is shown in Figure~\ref{fig:Q8_Circuits}b). It involves variational parameters, and an example of the training process is shown in Figure~\ref{fig:Q8_GS_Training}. Table~\ref{tab:Q8_GS} shows the final results after training for various input states. The true fidelity is calculated using the semi-definite program given in \eqref{eq:SDP-rootfid-GS}.

\begin{figure}
\begin{center}
\includegraphics[width=\linewidth]{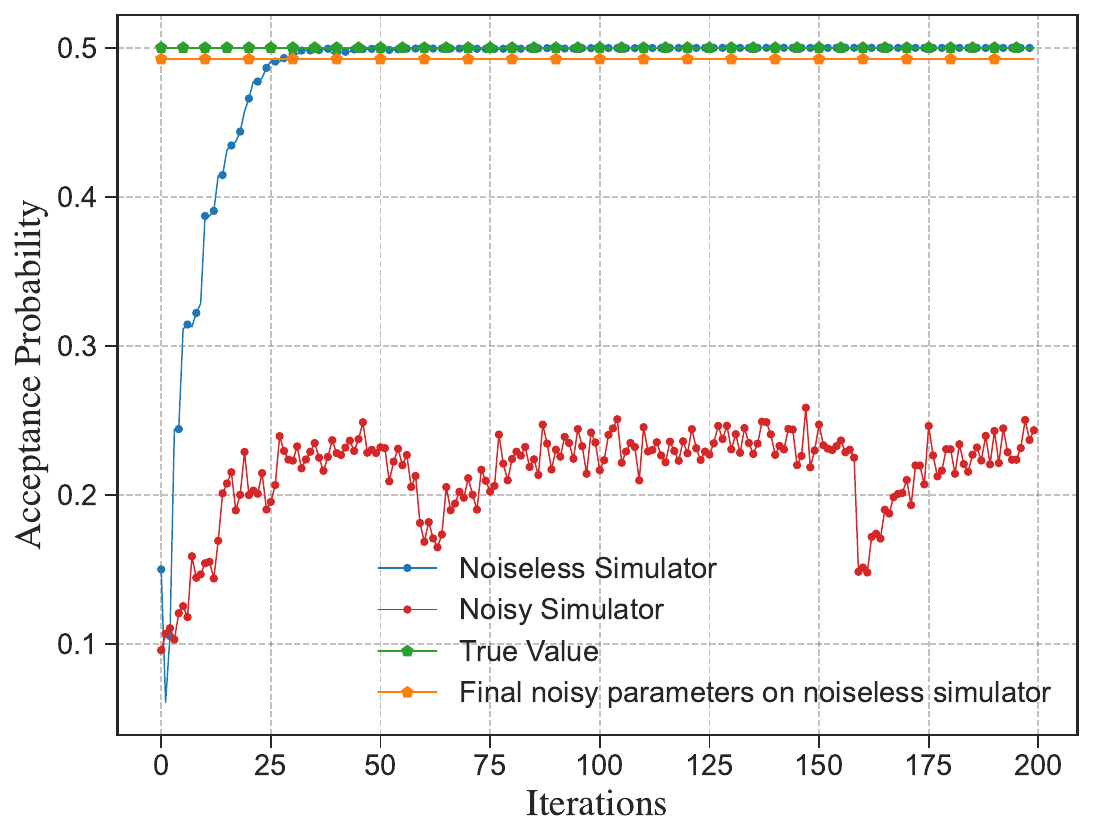}
\end{center}
\caption{Example of the training process for testing $Q_8$-symmetry of $\rho = \outerproj{\psi}$, where $\ket{\psi} = \ket{1+}$. We see that the training exhibits a noise resilience.}
\label{fig:Q8_GS_Training}
\end{figure}

\begin{table}[h]
\vspace{.05in}
\centering
\begin{tabular}{
>{\centering\arraybackslash}p{0.08\textwidth} | >{\centering\arraybackslash}p{0.07\textwidth} | 
>{\centering\arraybackslash}p{0.07\textwidth} | 
>{\centering\arraybackslash}p{0.06\textwidth} |
>{\centering\arraybackslash}p{0.08\textwidth}}
\hline
\textrm{State} & 
\textrm{True Fidelity} &
\textrm{Noiseless} &
\textrm{Noisy} & 
\textrm{Noise Resilient}
\\
\hline\hline 
$\outerproj{00}$ & 1.0000 & 0.9998 & 0.5430 & 0.9960\\
$\ket{1+}\!\bra{1\!+\!}$ & 0.5000 & 0.4999 & 0.2433 & 0.4924 \\
$\rho$ & 0.7500 & 0.7499 & 0.4581 & 0.7447\\
$\pi^{\otimes 2}$ & 1.0000 & 0.9998 & 0.2448 & 0.3774\\
\hline
\end{tabular}
\caption{Results of $Q_8$-symmetry tests. The state $\rho$ is defined as $\outerproj{\psi}$ where $\ket{\psi} = \frac{1}{2}(\sqrt{3}\ket{00} + \ket{11})$.}
\label{tab:Q8_GS}
\end{table}

\subsection{\texorpdfstring{$G$}{G}-Bose symmetric extendibility}

A circuit that tests for $G$-Bose symmetric extendibility is shown in Figure~\ref{fig:Q8_Circuits}c). It involves variational parameters, and an example of the training process is shown in Figure~\ref{fig:Q8_GBSE_Training}. Table~\ref{tab:Q8_GBSE} shows the final results after training for various input states. The true fidelity is calculated using the semi-definite program given in \eqref{eq:SDP-rootfid-GBSE}.

\begin{figure}
\begin{center}
\includegraphics[width=\linewidth]{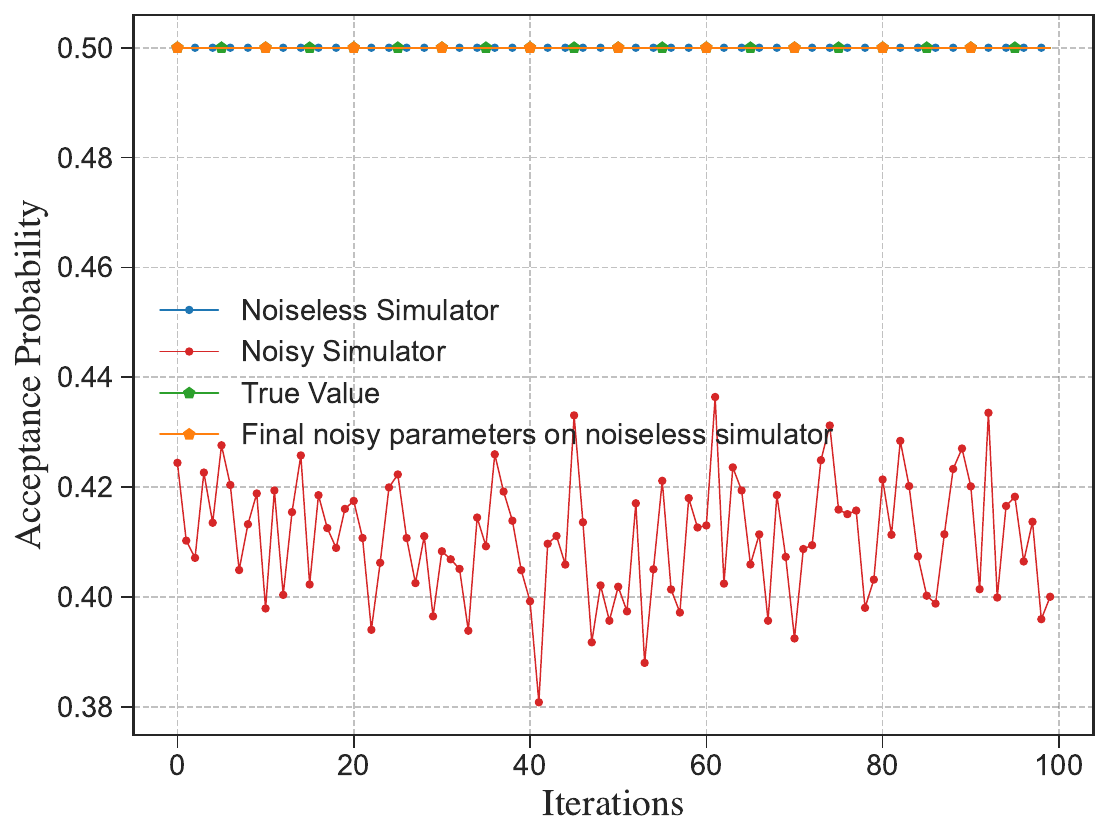}
\end{center}
\caption{Example of the training process for testing $Q_8$-Bose symmetric extendibility of $\outerproj{+}$. We see that the training exhibits a noise resilience.}
\label{fig:Q8_GBSE_Training}
\end{figure}

\begin{table}[h]
\centering
\vspace{.05in}
\begin{tabular}{
>{\centering\arraybackslash}p{0.05\textwidth} | >{\centering\arraybackslash}p{0.08\textwidth} | 
>{\centering\arraybackslash}p{0.08\textwidth} | 
>{\centering\arraybackslash}p{0.06\textwidth} |
>{\centering\arraybackslash}p{0.08\textwidth}}
\hline
\textrm{State} & 
\textrm{True Fidelity} &
\textrm{Noiseless} &
\textrm{Noisy} & 
\textrm{Noise Resilient}
\\
\hline\hline 
$\outerproj{0}$ & 1.0000 & 1.0000 & 0.7161 & 1.0000 \\
$\pi$ & 0.5000 & 0.5000 & 0.4086 & 0.5000 \\
$\rho$ & 0.9330 & 0.9330 & 0.6519 & 0.9330 \\
\hline
\end{tabular}
\caption{Results of $Q_8$-Bose symmetric extendibility tests. The state $\rho$ is defined as $\begin{bmatrix} 0.933 & 0.25\\ 0.25 & 0.067 \end{bmatrix}$.}
\label{tab:Q8_GBSE}
\end{table}

\subsection{\texorpdfstring{$G$}{G}-symmetric extendibility}

A circuit that tests for $G$-symmetric extendibility is shown in Figure~\ref{fig:Q8_Circuits}d). It involves variational parameters, and an example of the training process is shown in Figure~\ref{fig:Q8_GSE_Training}. Table~\ref{tab:Q8_GSE} shows the final results after training for various input states. The true fidelity is calculated using the semi-definite program given in \eqref{eq:SDP-rootfid-GSE}.

\begin{figure}
\begin{center}
\includegraphics[width=\linewidth]{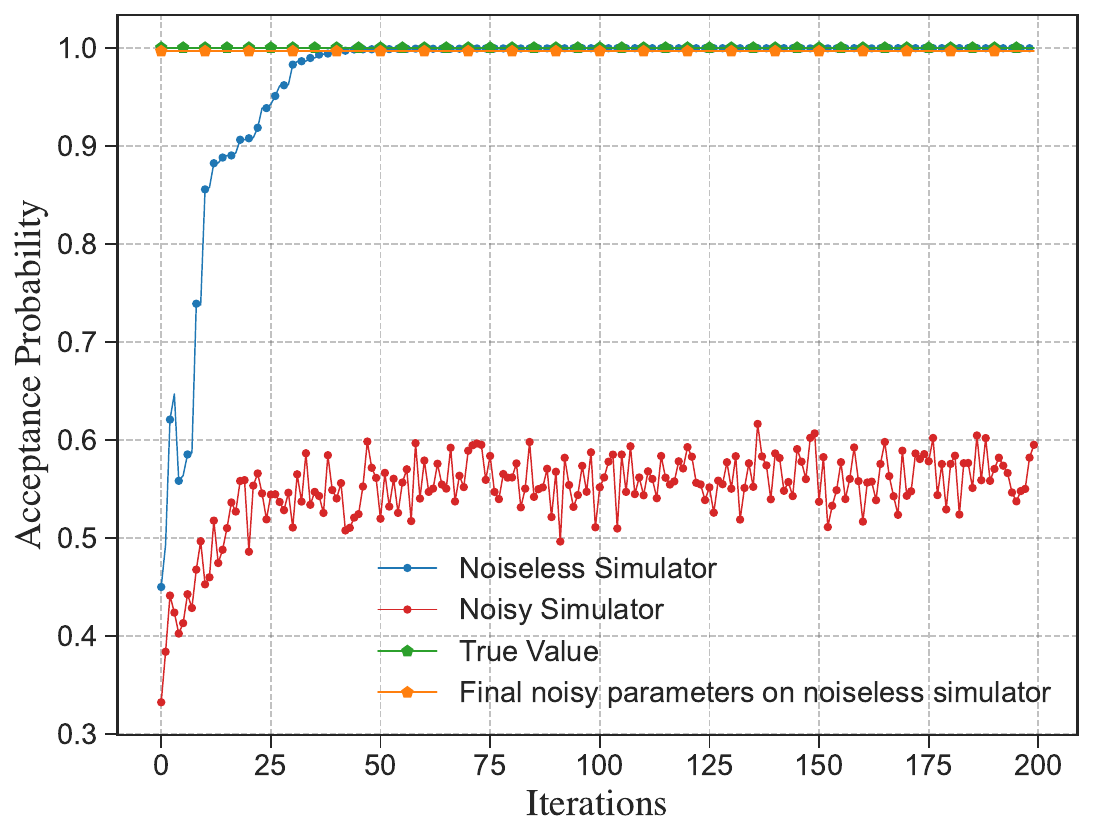}
\end{center}
\caption{Example of the training process for testing $Q_8$-symmetry extendibility of $\outerproj{0}$. We see that the training exhibits a noise resilience.}
\label{fig:Q8_GSE_Training}
\end{figure}

\begin{table}[h]
\centering
\vspace{.05in}
\begin{tabular}{
>{\centering\arraybackslash}p{0.05\textwidth} | >{\centering\arraybackslash}p{0.08\textwidth} | 
>{\centering\arraybackslash}p{0.08\textwidth} | 
>{\centering\arraybackslash}p{0.06\textwidth} |
>{\centering\arraybackslash}p{0.08\textwidth}}
\hline
\textrm{State} & 
\textrm{True Fidelity} &
\textrm{Noiseless} &
\textrm{Noisy} & 
\textrm{Noise Resilient}
\\
\hline\hline 
$\outerproj{0}$ & 1.0000 & 0.9995 & 0.5951 & 0.9964 \\
$\outerproj{+}$ & 0.5000 &  0.5000 & 0.2918 & 0.4974 \\
$\pi$ & 1.0 & 0.9985 & 0.4605 & 0.8778 \\
\hline
\end{tabular}
\caption{Results of $Q_8$-symmetric extendibility tests.}
\label{tab:Q8_GSE}
\end{table}

\end{document}